\documentclass[journal]{IEEEtran}

\usepackage[latin1]{inputenc}
\usepackage{amsmath}
\usepackage{amssymb}
\usepackage{graphicx}

\usepackage{enumitem}
\usepackage[font=footnotesize,caption=false]{subfig}
\usepackage{stfloats}
\usepackage{cite}
\usepackage{color}
\usepackage{tikz}
\usetikzlibrary{shapes,calc}
\usepackage{multirow}
\usepackage{hyperref}
\usepackage{enumitem}
\hypersetup{colorlinks, citecolor=black, filecolor=black, linkcolor=black, urlcolor=black}
\usepackage{array}
\usepackage{arydshln} % must be loaded AFTER array-package
\usepackage{mathrsfs}
\usepackage{eucal}
\usepackage{booktabs}

\newtheorem{theorem}{Theorem}

\newtheorem{definition}{Definition}
\newtheorem{lemma}{Lemma}
\newtheorem{example}{Example}

\makeatletter
\long\def\@makecaption#1#2{\ifx\@captype\@IEEEtablestring%
\footnotesize\begin{center}{\normalfont\footnotesize #1}\\
{\normalfont\footnotesize\scshape #2}\end{center}%
\@IEEEtablecaptionsepspace
\else
\@IEEEfigurecaptionsepspace
\setbox\@tempboxa\hbox{\normalfont\footnotesize {#1.}~~ #2}%
\ifdim \wd\@tempboxa >\hsize%
\setbox\@tempboxa\hbox{\normalfont\footnotesize {#1.}~~ }%
\parbox[t]{\hsize}{\normalfont\footnotesize \noindent\unhbox\@tempboxa#2}%
\else
\hbox to\hsize{\normalfont\footnotesize\hfil\box\@tempboxa\hfil}\fi\fi}
\makeatother

\mathchardef\ordinarycolon\mathcode`\:
  \mathcode`:= \string"8000
  \begingroup \catcode`\:= \active
    \gdef:{\mathrel{\mathop\ordinarycolon}}
  \endgroup

\makeatletter

\newcommand*\wthelper[2]{%
        \hbox{\dimen@\accentfontxheight#1%
                \accentfontxheight#11.25\dimen@
                $\m@th#1\hat{#2}$%
                \accentfontxheight#1\dimen@
        }%
}

\newcommand*\accentfontxheight[1]{%
        \fontdimen5\ifx#1\displaystyle
                \textfont
        \else\ifx#1\textstyle
                \textfont
        \else\ifx#1\scriptstyle
                \scriptfont
        \else
                \scriptscriptfont
        \fi\fi\fi3
}
\makeatother

\IEEEoverridecommandlockouts

\begin{document}

\title{Absorbing Set Analysis and Design of LDPC Codes from Transversal Designs over the AWGN Channel}

\author{Alexander~Gruner,~\IEEEmembership{Student Member,~IEEE}, and Michael~Huber,~\IEEEmembership{Member,~IEEE} %
\thanks{The work of A.~Gruner was supported within the Promotionsverbund `Kombinatorische Strukturen' (LGFG Baden-W\"urttemberg).
The work of M.~Huber was supported by the Deutsche Forschungsgemeinschaft (DFG) via a Heisenberg grant (Hu954/4) and a Heinz Maier-Leibnitz Prize grant (Hu954/5).}
\thanks{The authors are with the Wilhelm Schickard Institute for Computer Science, Eberhard Karls Universit\"at T\"ubingen, Sand~13,
D-72076 T\"ubingen, Germany (corresponding author's e-mail address: alexander.gruner@uni-tuebingen.de).}}%

\maketitle

\begin{abstract}
In this paper we construct low-density parity-check (LDPC) codes from transversal designs with low error-floors over the additive white Gaussian noise (AWGN) channel. The constructed codes are based on transversal designs that arise from sets of mutually orthogonal Latin squares (MOLS) with cyclic structure. 
For lowering the error-floors, our approach is twofold: First, we give an exhaustive classification of so-called absorbing sets that may occur in the factor graphs of the given codes. These purely combinatorial substructures are known to be the main cause of decoding errors in the error-floor region over the AWGN channel by decoding with the standard sum-product algorithm (SPA). Second, based on this classification, we exploit the specific structure of the presented codes to eliminate the most harmful absorbing sets and derive powerful constraints for the proper choice of code parameters in order to obtain codes with an optimized error-floor performance.
\end{abstract}

\begin{IEEEkeywords}
Low-density parity-check (LDPC) code, additive white gaussian noise (AWGN) channel, absorbing set, transversal design (TD), mutually orthogonal latin squares (MOLS).
\end{IEEEkeywords}

\IEEEpeerreviewmaketitle

\section{Introduction}

\IEEEPARstart{L}{DPC} codes are linear block codes that potentially achieve the limit of Shannon's famous coding theorem and thus reveal excellent error-correcting properties under iterative decoding. An important and challenging task for designing LDPC codes is to lower the so-called \emph{error-floors}. This phenomenon is a significant flattening of the \emph{bit-error-rate} (BER) curve beyond a certain \emph{signal-to-noise-ratio} (SNR). It has been discovered that error-floors are caused by special substructures in the code's \emph{factor (or Tanner) graph} that act as internal states in which the iterative decoder can be trapped. Richardson \cite{Richardson2003} introduced the notion of \emph{trapping sets} to describe such internal states for iterative decoders. Depending on the channel and the iterative decoding algorithm, trapping sets have quite different characteristics. 
Over the \emph{binary erasure channel} (BEC), trapping sets have a purely combinatorial character and are known as \emph{stopping sets} (e.g. \cite{Di02,Ka03,SchwVar06}), which completely determine the decoding performance over this channel \cite{Di02}. For more complex non-erasure channels such as the AWGN channel, trapping sets have a more subtle nature and can not easily be described by a simple combinatorial notion such as stopping sets. However, a subclass of the occurring trapping sets over the AWGN channel with standard SPA decoding can be described by combinatorial objects called \emph{(fully) absorbing sets} which has been introduced in \cite{Zhang2006} as special subgraphs of a code's factor graph. It has been demonstrated by extensive hardware simulations \cite{Zhang2006,Zhang2009} that these entities are the main contributors to the error-floors over the AWGN channel under SPA decoding.
It is therefore an important step to identify the dominant absorbing sets that may occur in the factor graph of an LDPC code and to eliminate the harmful ones in order to improve the decoding performance in the error-floor region. 
Recent papers have proposed methods to characterize and improve the absorbing set spectrum of certain LDPC codes \cite{DolWang2010,Dolecek2010,Wang2013}. 
For instance, the work of Dolecek et al. \cite{Dolecek2010} provides an extensive analysis of the absorbing sets occurring in a family of array-based LDPC codes.
We also emphasize \cite{Wang2013} which presents a powerful approach to eliminate dominant absorbing sets in a wide class of circulant-based LDPC codes.

In the present paper, we exploit a special class of transversial designs that arise from cyclic-structured MOLS in order to design LDPC codes with low error-floors over the AWGN channel. These codes provide a simple setting to investigate and eliminate harmful absorbing sets in a closed algebraic form.
In \cite{GrunHub2013} we have demonstrated that this code family can also be utilized to generate LDPC codes with excellent decoding performances over the BEC by eliminating the smallest stopping sets. 
Moreover, we have shown that these codes possess quasi-cyclic structure and thus can be encoded with linear complexity via simple feedback shift registers \cite{TowWel67}. Notice that codes based on transversal designs has been first considered and investigated in \cite{JohnWell2004} and \cite{JohnsonDiss} in terms of \emph{partial geometries}.
It is worth noting here that absorbing sets are stable under bit-flipping decoding and thus also greatly contribute to the decoding failures over the binary symmetric channel (BSC). Therefore, our approach should also produce excellent LDPC codes over the BSC.

The paper is organized as follows: In Section~\ref{preliminaries}, we give a summarization of the theoretical concepts that are important for our purposes. 
In Section~\ref{section:classification}, we elaborate a classification of the smallest absorbing set candidates that typically occur in the factor graph of LDPC codes based on transversal designs. 
In Section~\ref{codes_on_simple_structured_MOLS}, we thoroughly describe the construction of LDPC codes that arise from sets of cyclic-structured MOLS and investigate the properties of these codes. 
Based on this class of codes and the absorbing set classification of Section~\ref{section:classification}, we develop a method to eliminate harmful absorbing sets in Section~\ref{section:elimination} and present the main results in Section~\ref{section:main_results}. In Section~\ref{simulations}, we demonstrate the strength of our elimination technique by extensive simulations and conclude the paper in Section~\ref{conclusion}.

\section{Preliminaries}\label{preliminaries}

\subsection{Latin Squares}

A \emph{Latin square} $L$ of order $n$ is an array of $n\times n$ cells, where each row and each column contains every symbol of an $n$-set $S$ exactly once \cite{crc_handbook}. Let $L[x,y]$ denote the symbol at row $x\in X$ and column $y\in Y$, where $X$ and $Y$ are $n$-sets indexing the rows and columns of $L$, respectively. Two Latin squares $L_1$ and $L_2$ of order $n$ are \emph{orthogonal}, if they share a common row and column set $X$ and $Y$, respectively, and if the ordered pairs $(L_1[x,y],L_2[x,y])$ are unique for all $(x,y) \in X \times Y$. In other words, there can not be two cell positions $[x_1,y_1]$ and $[x_2,y_2]$ such that $L_1[x_1,y_1]=L_1[x_2,y_2]$ and $L_2[x_1,y_1]=L_2[x_2,y_2]$. A set of Latin squares $L_1,...,L_m$ is called \emph{mutually orthogonal}, if for every $1\leq i<j\leq m$, $L_i$ and $L_j$ are orthogonal. These are also referred to as \emph{MOLS}, \emph{mutually orthogonal Latin squares}.

\subsection{Transversal Designs}\label{TDs}

A \emph{transversal design} TD$(k,n)$ of order (or group size) $n$ and block size $k$ is a triple $(\CMcal{P},\CMcal{G},\CMcal{B})$, where

\begin{enumerate}[label=(\arabic*),leftmargin=22pt]
	\item $\CMcal{P}$ is a set of $kn$ points.
	\item $\CMcal{G}$ is a partition of $\CMcal{P}$ into $k$ classes of size $n$, called \emph{groups}.
	\item $\CMcal{B}$ is a collection of $k$-subsets of $\CMcal{P}$, called \emph{blocks}.
	\item Every unordered pair of points from $\CMcal{P}$ is contained either in exactly one group or in exactly one block (cf.  \cite{crc_handbook}).
\end{enumerate}

It follows from (1)-(4) that any point of $\CMcal{P}$ occurs in exactly $n$ blocks and that $|\CMcal{B}|=n^2$. Furthermore, axiom (4) implies that every block of $\CMcal{B}$ consists of exactly one point per group.

\begin{theorem}\label{equiv_MOLS_TDs}
For $k\ge 3$, the existence of a set of $m:= k-2$ mutually orthogonal Latin squares (MOLS) of order $n$ is equivalent to the existence of a TD$(k,n)$ \cite{crc_handbook, BosShri90, Wilson74}.
\end{theorem}

\begin{proof}
We will outline the proof of this known result, since it is important for the understanding of our paper. 
Let $L_1,\hdots,L_m$ be $m$ MOLS with symbol sets $S_1,\hdots,S_m$, and with common row and column sets $X$ and $Y$, respectively. We may assume that the sets $X,Y,S_1,\hdots,S_m$ are pairwise disjoint, which can easily be achieved by renaming the elements. Then we obtain a TD$(k,n)$ with points $\CMcal{P}=\{X \cup Y \cup S_1 \cup \hdots \cup S_m\}$, groups $\CMcal{G}=\{X,Y,S_1,\hdots,S_m\}$ and blocks $\CMcal{B}=\{\{x, y, L_1[x,y], \hdots, L_{m}[x,y]\}: (x,y)\in X\times Y\}$. This process can be reversed to recover a set of $m=k-2$ MOLS from a TD$(k,n)$ for $k\geq 3$.
\end{proof}

A transversal design, denoted by $\CMcal{D}$, can be described by a binary $|\CMcal{P}|\times|\CMcal{B}|$ \emph{incidence matrix} $\CMcal{N}(\CMcal{D})$ with rows indexed by the points of $\CMcal{P}$, columns indexed by the blocks of $\CMcal{B}$, and
\[\CMcal{N}(\CMcal{D})_{ij}=\begin{cases}
  1,  & \text{if the $i$-th point is in the $j$-th block}\\
  0, & \text{otherwise}.
\end{cases}\]

\begin{example}
Fig.~\ref{fig:MOLS} depicts the incidence matrix of the transversal design TD$(4,5)$ which is equivalent to the orthogonal Latin squares given in the same figure by using the correspondence detailed in the proof of Theorem~\ref{equiv_MOLS_TDs}.
\end{example}

\begin{figure}[t!]
\renewcommand{\arraystretch}{1.2}
\subfloat{
		\scalebox{0.97}{$
		\begin{array}[b]{|ccccc|}
		\hline
		0&1&2&3&4\\
		1&2&3&4&0\\
		2&3&4&0&1\\
		3&4&0&1&2\\
		4&0&1&2&3\\
		\hline
		\end{array}
		$}
}
\subfloat{
		\scalebox{0.97}{$
		\begin{array}[b]{|ccccc|}
		\hline
		0&1&2&3&4\\
		2&3&4&0&1\\
		4&0&1&2&3\\
		1&2&3&4&0\\
		3&4&0&1&2\\
		\hline
		\end{array}
		$}
}
\subfloat{
	\includegraphics[scale = 0.54]{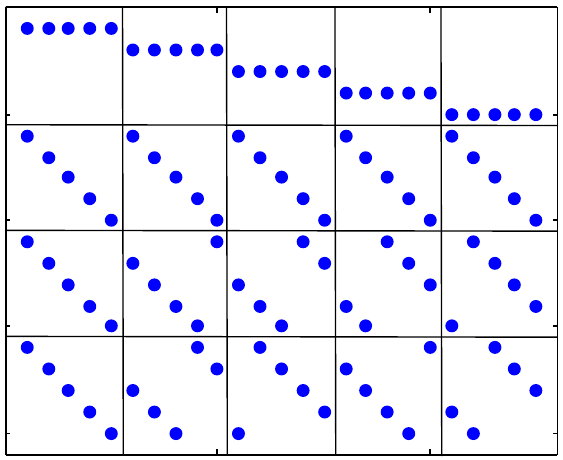}
}
\caption{Orthogonal Latin squares of order 5 and the incidence matrix of the resulting TD$(4,5)$. The black dots represent 1-entries in the incidence matrix.}
\label{fig:MOLS}
\end{figure}

\subsection{LDPC Codes from Transversal Designs}\label{codes_from_TDs}

Let $\CMcal{D}$ denote a TD$(k,n)$ with points $\CMcal{P}$ and blocks $\CMcal{B}$. 
The incidence matrix $\CMcal{N}(\CMcal{D})$ can directly be used as the parity-check matrix $H$ of a \textit{TD LDPC code}, such that the $|\CMcal{P}|=kn$ points correspond to the parity-check equations (rows of $H$) and the $|\CMcal{B}|=n^2$ blocks correspond to the code bits (columns of $H$). The resulting TD LDPC code has block length $N=n^2$, rate $R\geq(n-k)/n$, and a parity-check matrix of column weight $k$ and row weight $n$. The column weight $k$ corresponds to the block size of $\CMcal{D}$ and the row weight $n$ arises from the fact that every point is incident to exactly $n$ blocks. The associated factor (or Tanner) graph of $H$ is free of 4-cycles and has girth $g=6$ (e.g.~\cite{JohnWell2004}).
Note that these codes were first considered in \cite{JohnWell2004} as a subclass of codes from partial geometries.

\subsection{Set Systems}\label{set_systems}

A \emph{set system} $(\mathcal{P},\mathcal{B})$ consists of a point set $\mathcal{P}$ and a block set $\mathcal{B}$ which is a family of subsets of $\mathcal{P}$. For a consistent representation, we only allow set systems where $\mathcal{P}=\bigcup \mathcal{B}$ such that each point must be contained in at least one block. The \emph{size} of $(\mathcal{P},\mathcal{B})$ is given by the number of blocks of $\mathcal{B}$. The \emph{degree} of a point is the number of blocks containing the point and by $O(\mathcal{P})$ we mean the subset of points of $\mathcal{P}$ having odd degree.
Two set systems $(\mathcal{P},\mathcal{B})$ and $(\mathcal{P}',\mathcal{B}')$ are \emph{isomorphic}, if there exists a bijection between $\mathcal{P}$ and $\mathcal{P}'$ that maps $\mathcal{B}$ to $\mathcal{B}'$.
Given a transversal design $\CMcal{D}$ with points $\CMcal{P}$ and blocks $\CMcal{B}$, a set system $(\mathcal{P},\mathcal{B})$ is called a \emph{configuration} of $\CMcal{D}$ if $\mathcal{P}\subseteq\CMcal{P}$ and $\mathcal{B}\subseteq \CMcal{B}$.\footnote{The notion of \emph{configurations} in designs is frequently used in the literature about combinatorial designs and codes from designs (e.g. \cite{Ka03}).} 
Every set system can be equivalently described by a matrix which allows a convenient representation of the system. More precisely, a set system is equivalent to a $|\mathcal{P}|\times|\mathcal{B}|$ matrix where the entry at $(i,j)$ is one if the $i$-th point of $\mathcal{P}$ is in the $j$-th block of $\mathcal{B}$, else zero.

\subsection{New Concept: $t$-Colourings of Set Systems}

Let $Q:=\{1,\hdots,t\}$ be a set of $t$ colours. 
A \emph{\mbox{$t$-colouring}} of a set system $(\mathcal{P},\mathcal{B})$ is a mapping $\varphi: \mathcal{P}\to Q$ such that all points of a block of $\mathcal{B}$ are coloured with a different colour of $Q$. 
For each $c\in Q$, the set $\varphi^{-1}(c)=\{x\in \mathcal{P}:\varphi(x)=c\}$ is called a \emph{colour class}.
For any $t$-colouring $\varphi$, we use the ordered set $\hat{\varphi}=(\varphi^{-1}(1), \varphi^{-1}(2),\hdots, \varphi^{-1}(t))$ as a short notation of this colouring.
If a set system has at least one $t$-colouring, we call it \emph{$t$-colourable}. Clearly, there can be different $t$-colourings of a set system.
Let $(\mathcal{P},\mathcal{B})$ and $(\mathcal{P}',\mathcal{B}')$ be two set systems with $t$-colourings $\varphi$ and $\varphi'$, respectively. Then, $(\mathcal{P},\mathcal{B},\varphi)$ and $(\mathcal{P}',\mathcal{B}',\varphi')$ are isomorphic, if there are bijections $\sigma: \mathcal{P}\to \mathcal{P}'$ and $\pi: Q\to Q$ such that $\sigma(\mathcal{B})=\mathcal{B}'$ with $\sigma(\mathcal{B})=\big\{\{\sigma(x):x\in B\} : B\in\mathcal{B}\big\}$ and 
$\pi(\varphi(x))=\varphi'(\sigma(x))$.

\begin{figure}[t!]
\setlength{\arraycolsep}{5pt}
\renewcommand{\arraystretch}{1.2}
\centering
\small
\centerline{
\subfloat{
	\scalebox{1}{$
	\begin{array}[b]{c|cccc|}
	\cline{2-5}
	1 & 1 & 1 & . & .\\
	2 & 1 & . & 1 & .\\
	3 & 1 & . & . & 1\\
	4 & 1 & . & . & .\\
	5 & . & 1 & 1 & .\\
	6 & . & 1 & . & 1\\
	7 & . & 1 & . & .\\
	8 & . & . & 1 & 1\\
	9 & . & . & 1 & .\\
	10 & . & . & . & 1\\
	\cline{2-5}
	\end{array}
	$}
}
\subfloat{
	\scalebox{1}{$
	\begin{array}[b]{c|cccc|}
	\cline{2-5}
	1 & 1 & 1 & . & .\\
	8 & . & . & 1 & 1\\
	\hdashline[1pt/2pt]
	2 & 1 & . & 1 & .\\
	6 & . & 1 & . & 1\\
	\hdashline[1pt/2pt]
	3 & 1 & . & . & 1\\
	7 & . & 1 & . & .\\
	9 & . & . & 1 & .\\
	\hdashline[1pt/2pt]
	4 & 1 & . & . & .\\
	5 & . & 1 & 1 & .\\
	10 & . & . & . & 1\\
	\cline{2-5}
	\end{array}
	$}
}
\subfloat{
	\scalebox{1}{$
	\begin{array}[b]{c|cccc|}
	\cline{2-5}
	1 & 1 & 1 & . & .\\
	8 & . & . & 1 & 1\\
	\hdashline[1pt/2pt]
	2 & 1 & . & 1 & .\\
	6 & . & 1 & . & 1\\
	\hdashline[1pt/2pt]
	3 & 1 & . & . & 1\\
	5 & . & 1 & 1 & .\\
	\hdashline[1pt/2pt]
	4 & 1 & . & . & .\\
	7 & . & 1 & . & .\\
	9 & . & . & 1 & .\\
	10 & . & . & . & 1\\
	\cline{2-5}
	\end{array}
	$}
}
}
\caption{Matrix representation of a set system $\mathcal{S}$ with points $\mathcal{P}=\{1,\hdots,10\}$ (left) and the two non-isomorphic 4-colourings of $\mathcal{S}$ grouped by their colour classes. The dots represent zero-entries.}
\label{fig:set_system}
\end{figure}

\begin{example}\label{example:colouring_set_system}
Let $\mathcal{S}=(\mathcal{P},\mathcal{B})$ be a set system with point set $\mathcal{P}=\{1,\hdots,10\}$ and block set $\mathcal{B}=\big\{\{1,2,3,4\}, \{1,5,6,7\},\allowbreak \{2,5,8,9\},\allowbreak \{3,6,8,10\}\big\}$ as depicted in Fig.~\ref{fig:set_system}, and let $Q=\{1,2,3,4\}$ be a set of four colours. 
Then, there are two non-isomorphic 4-colourings of $\mathcal{S}$,
\[
\varphi_1(x)=
\begin{cases}
1,\ x=1,8\\
2,\ x=2,6\\
3,\ x=3,7,9\\
4,\  x=4,5,10,
\end{cases}
\varphi_2(x)=
\begin{cases}
1,\  x=1,8\\
2,\  x=2,6\\
3,\  x=3,5\\
4,\  x=4,7,9,10.
\end{cases}
\]
\end{example}
The colour classes are given by $\varphi_1^{-1}(1)=\{1,8\},\ \varphi_1^{-1}(2)=\{2,6\},\ \varphi_1^{-1}(3)=\{3,7,9\}$ and $\varphi_1^{-1}(4)=\{4,5,10\}$ for the first colouring $\varphi_1$, and $\varphi_2^{-1}(1)=\{1,8\},\ \varphi_2^{-1}(2)=\{2,6\},\ \varphi_2^{-1}(3)=\{3,5\}$ and $\varphi_2^{-1}(4)=\{4,7,9,10\}$ for the second colouring $\varphi_2$.
It can easily be seen that the points of any block of $\mathcal{B}$ are contained in different colour classes and thus are coloured uniquely within a block. 
The short notations of the colourings are $\hat{\varphi}_1=(\{1,8\},\{2,6\},\{3,7,9\},\{4,5,10\})$ and $\hat{\varphi}_2=(\{1,8\},\{2,6\},\{3,5\},\{4,7,9,10\})$, respectively.

\subsection{Factor Graph of LDPC Codes}\label{factor_graph}

The parity-check matrix $H$ of an LDPC code can be equivalently represented by its \emph{factor graph} $G_H$. More precisely, let $G_H=(V,F,E)$ denote a bipartite graph, where

\begin{itemize}[leftmargin=15pt,labelindent=10pt]
\item the nodes $V$ are associated with the rows of $H$,
\item the nodes $F$ are associated with the columns of $H$,
\item the edge set $E$ is given by the structure of $H$. In particular, an undirected edge $e(i,j)\in E$ exists iff $H_{ij}=1$.
\end{itemize}
The elements of $V$ are called \emph{bit nodes}, since they correspond to the code bits, and the elements of $F$ are called \emph{check nodes}, since they correspond to the parity-check equations of the code.

\subsection{Absorbing Sets}

Let $G_H=(V,F,E)$ be the factor graph of a given LDPC code with parity-check matrix $H$.
For any subset $X \subseteq V$, let $N_F(X)$ be the set of neighboring check nodes of $X$ in $F$, i.e., $N_F(X):= \{y\in F: \exists x\in X,\ \allowbreak e(x,y)\in E \}$. We further subdivide $N_F(X)$ into the sets of check nodes with even and odd degrees, denoted by $\CMcal{E}(X)$ and $O(X)$, respectively. Furthermore, we denote $N_E(X)$ as the set of adjacent edges of $X$ in $E$, i.e., $N_E(X):= \{e(x,y)\in E: x\in X\}$.

\begin{definition}\label{def:absorbing_set}
An \emph{$(a,b)$ absorbing set} of $G_H$, denoted by $\CMcal{A}$, is a subset $\CMcal{A}\subseteq V$ with $|\CMcal{A}|=a$ and $O(\CMcal{A})=b$, such that each node of $\CMcal{A}$ has strictly fewer neighbors in $O(\CMcal{A})$ than in $\CMcal{E}(\CMcal{A})$ \cite{Zhang2006}. We say that $a$ is the \emph{size} of the absorbing set and $b$ is the \emph{syndrome}. 
An $(a,b)$ absorbing set is an \emph{$(a,b)$ fully absorbing set} if, in addition, all bit nodes in $V\setminus \CMcal{A}$ have strictly fewer neighbors in $O(\CMcal{A})$ than in $F\setminus O(\CMcal{A})$.
An absorbing set $\CMcal{A}$ is $\emph{elementary}$, if all check nodes of $N_F(\CMcal{A})$ have degree of at most 2.
The $\emph{induced subgraph}$ of $\CMcal{A}$, denoted by $\CMcal{I}(\CMcal{A})$, is a bipartite subgraph of $G_H$ consisting of the bit nodes $\CMcal{A}$, the neighboring check nodes $N_F(\CMcal{A})$ and the adjacent edges $N_E(\CMcal{A})$.
\end{definition}

\begin{figure}[!t]
\centering
\begin{tikzpicture}[xscale=1.0,yscale=1.0]

\tikzstyle{bit_node}=[circle,draw=black,fill=black,minimum size=6pt]
\tikzstyle{check_node_filled}=[rectangle,draw=black,fill=black,minimum size=8pt]
\tikzstyle{check_node_empty}=[rectangle,draw=black,fill=white,minimum size=8pt]

\node [bit_node] (b1) at (0,0) {};
\node [bit_node] (b2) at ($(b1)+(2,0)$) {};
\node [bit_node] (b3) at ($(b2)+(2,0)$) {};
\node [bit_node] (b4) at ($(b3)+(2,0)$) {};

\def\cnodeoffset{0.9}
\node [check_node_empty] (c1) at (-1.1,-1.8) {};
\node [check_node_empty] (c2) at  ($(c1)+(\cnodeoffset,0)$) {};
\node [check_node_empty] (c3) at  ($(c2)+(\cnodeoffset,0)$) {};
\node [check_node_empty] (c4) at  ($(c3)+(\cnodeoffset,0)$) {};
\node [check_node_empty] (c5) at  ($(c4)+(\cnodeoffset,0)$) {};
\node [check_node_empty] (c6) at  ($(c5)+(\cnodeoffset,0)$) {};
\node [check_node_filled] (c7) at  ($(c6)+(\cnodeoffset,0)$) {};
\node [check_node_filled] (c8) at  ($(c7)+(\cnodeoffset,0)$) {};
\node [check_node_filled] (c9) at  ($(c8)+(\cnodeoffset,0)$) {};
\node [check_node_filled] (c10) at  ($(c9)+(\cnodeoffset,0)$) {};

\draw (b1) edge (c1);
\draw (b1) edge (c2);
\draw (b1) edge (c3);
\draw (b1) edge (c7);

\draw (b2) edge (c1);
\draw (b2) edge (c4);
\draw (b2) edge (c5);
\draw (b2) edge (c8);

\draw (b3) edge (c2);
\draw (b3) edge (c4);
\draw (b3) edge (c6);
\draw (b3) edge (c9);

\draw (b4) edge (c3);
\draw (b4) edge (c5);
\draw (b4) edge (c6);
\draw (b4) edge (c10);

\draw (3,0) node[rectangle, rounded corners=0.4cm, label=above:$\CMcal{A}$, minimum height=0.8cm, minimum width=6.6cm, draw, dashed] {};
\draw (1.15,-1.8) node[rectangle, rounded corners=0.4cm, label=below:$\CMcal{E}(\CMcal{A})$, minimum height=0.8cm, minimum width=5.2cm, draw, dashed] {};
\draw (5.65,-1.8) node[rectangle, rounded corners=0.4cm, label=below:$O(\CMcal{A})$, minimum height=0.8cm, minimum width=3.4cm, draw, dashed] {};

\end{tikzpicture}
\caption{Geometric representation of a $(4,4)$ absorbing set.}
\label{fig:4_4_absorbing_set}
\end{figure}

\begin{example}
An example of a $(4,4)$ absorbing set is shown in Fig.~\ref{fig:4_4_absorbing_set}, where the bit nodes $\CMcal{A}$ are represented by black circles. The neighboring check  nodes of odd degree $O(\CMcal{A})$ are visualized by black squares and the neighboring check nodes with even degree $\CMcal{E}(\CMcal{A})$ are visualized by white squares. Observe that all bit nodes of $\CMcal{A}$ are connected with strictly more even-degree check nodes than odd-degree check nodes and thus fulfill the definition of an absorbing set. 
\end{example}

\section{Classification of Absorbing Set Candidates \\ in LDPC Codes from Transversal Designs}
\label{section:classification}

Absorbing sets are known to be the main cause of decoding errors in the error-floor region of LDPC codes over the AGWN channel under SPA decoding. Therefore, it is an important step to identify and categorize the most harmful absorbing sets of an LDPC code.\footnote{By an absorbing set of an LDPC code we actually mean an absorbing set that occurs in the code's factor graph. For simplicity, we prefer this short notation in the remainder of the paper.} In this section, we give an exhaustive classification of small absorbing set candidates that may occur in a TD LDPC code of column weight 3 and 4.

Absorbing sets are graph-based objects that has been established in the field of coding theory. However, they have a purely combinatorial nature and thus can be conveniently described from a combinatorial viewpoint. 
More precisely, let $\CMcal{D}$ be a transversal design of block size $k$ with point set $\CMcal{P}$ and block set $\CMcal{B}$ and let $\CMcal{C}(\CMcal{D})$ be the corresponding TD LDPC code of column weight $k$ based on $\CMcal{D}$.

\begin{lemma}\label{lemma:set_system_representation}
An absorbing set $\CMcal{A}$ of $\CMcal{C}(\CMcal{D})$ can be equivalently described by a set system $(\mathcal{P},\mathcal{B})$ which is a configuration of $\CMcal{D}$, i.e., $\mathcal{P}\subseteq \CMcal{P}$ and $\mathcal{B}\subseteq \CMcal{B}$. Then, $(\mathcal{P},\mathcal{B})$ is called the \emph{set system representation} of $\CMcal{A}$.
\end{lemma}
\begin{proof}
By taking the points and blocks of $\CMcal{D}$ that correspond to the bit nodes and check nodes of $\CMcal{I}(\CMcal{A})$, respectively, we obtain a set system that is clearly a configuration of $\CMcal{D}$.
\end{proof}

\begin{lemma}\label{lemma:induced_colouring}
Any configuration $(\mathcal{P},\mathcal{B})$ of $\CMcal{D}$ has a $k$-colouring that is induced by the groups of $\CMcal{D}$.
\end{lemma}
\begin{proof}
Recall that $\CMcal{D}$ has $k$ groups and that every block consists of exactly one point per group. By colouring all points of $\mathcal{P}$ that lie in the same group of $\CMcal{D}$ with the same colour and each group with a different colour, we obtain an $k$-colouring of $(\mathcal{P},\mathcal{B})$ which is induced by $\CMcal{D}$. 
\end{proof}

Let $\CMcal{A}$ be an absorbing set with set system representation $(\mathcal{P}^*,\mathcal{B}^*)$ and a $k$-colouring $\varphi^*$ induced by the groups of $\CMcal{D}$. 
By considering any set system $(\mathcal{P},\mathcal{B})$, we say that $\CMcal{A}$ is of \emph{type} $(\mathcal{P},\mathcal{B})$, if the set system is isomorphic to $(\mathcal{P}^*,\mathcal{B}^*)$. Moreover, by considering any $k$-colouring $\varphi$ of $(\mathcal{P},\mathcal{B})$, we say that $\CMcal{A}$ is of \emph{type} $(\mathcal{P},\mathcal{B},\varphi)$ if the triple is isomorphic to $(\mathcal{P}^*,\mathcal{B}^*,\varphi^*)$. Obviously, $(\mathcal{P},\mathcal{B})$ and $(\mathcal{P},\mathcal{B},\varphi)$ may represent a large number of isomorphic absorbing sets of an LDPC code.
Also notice that an absorbing set of type $(\mathcal{P},\mathcal{B},\varphi)$ is automatically an absorbing set of type $(\mathcal{P},\mathcal{B})$ but not vice versa.

\begin{lemma}\label{lemma:combinatorial_constraints_of_absorbing_sets}
By considering an absorbing set of type $(\mathcal{P},\mathcal{B})$ that occurs in any TD LDPC code of column weight $k$, the following combinatorial constraints must be valid:
\begin{enumerate}[label=(\Alph*),leftmargin=22pt]
	\item Each block of $\mathcal{B}$ contains exactly $k$ points.
	\item Any two blocks of $\mathcal{B}$ share at most one point of $\mathcal{P}$.
	\item $(\mathcal{P},\mathcal{B})$ is $k$-colourable.
	\item For every $B\in \mathcal{B}$, only a minority of the points of $B$ have odd degree, i.e., $|B\cap O(\mathcal{P})|\leq \lfloor \frac{k-1}{2}\rfloor$.
\end{enumerate}
\end{lemma}
\begin{proof}
Since $(\mathcal{P},\mathcal{B})$ corresponds to a configuration of $\CMcal{D}$, the constraints (A) and (B) directly follow from the axioms of a transversal design and (C) follows from Lemma~\ref{lemma:induced_colouring}.
The constraint (D) is the combinatorial counterpart to the postulation of absorbing sets that each bit node of $\CMcal{A}$ has strictly fewer neighbors in $O(\CMcal{A})$ than in $\CMcal{E}(\CMcal{A})$.
\end{proof}

\begin{definition}
An \emph{absorbing set candidate} is a set system that satisfies the constraints (A)-(D) of Lemma~\ref{lemma:combinatorial_constraints_of_absorbing_sets}.
\end{definition}

Absorbing set candidates can be considered as combinatorial patterns for absorbing sets independent of any code. By contrast, absorbing sets are inextricably linked with a concrete code since they are defined as subgraphs of the code's factor graph.

\begin{theorem}\label{theorem:combinatorial_absorbing_constraints}
	Let $\CMcal{A}$ be an absorbing set of type $(\mathcal{P},\mathcal{B},\varphi)$. 
	For the case of $k=3$ or $4$, the absorbing set $\CMcal{A}$ is fully if and only if the points of $O(\mathcal{P})$ are coloured with the same colour by $\varphi$, i.e., $|\{\varphi(x):x\in O(\mathcal{P})\}|\leq 1$.
\end{theorem}
\begin{proof}
We assume w.l.o.g. that $(\mathcal{P},\mathcal{B})$ is the set system representation of $\CMcal{A}$ with induced colouring $\varphi$.
If all points of $O(\mathcal{P})$ are coloured with the same colour by $\varphi$, then all points of odd degree belong to the same group of $\CMcal{D}$.
We also know that every block of $\CMcal{D}$ has exactly one point per group and thus every block of $\CMcal{B}\setminus\mathcal{B}$ intersects with at most one point of $O(\mathcal{P})$. Hence, every block of $\CMcal{B}\setminus\mathcal{B}$ has strictly fewer points in $O(\mathcal{P})$ than in $\CMcal{P}\setminus O(\mathcal{P})$ such that $\CMcal{A}$ must be fully.

Conversely, if we assume that there are two points of $O(\mathcal{P})$ with different colours (which means that they are contained in different groups of $\CMcal{D}$), there must be a block in $\CMcal{B}$ that contains both points (cf. (4) of Subsection~\ref{TDs}). Since $k=3$ or $4$, this block must be in $\CMcal{B}\setminus\mathcal{B}$, otherwise it contradicts the definition of an absorbing set. Hence, this block contradicts the definition of a fully absorbing set.
\end{proof}

\begin{figure*}[t!]
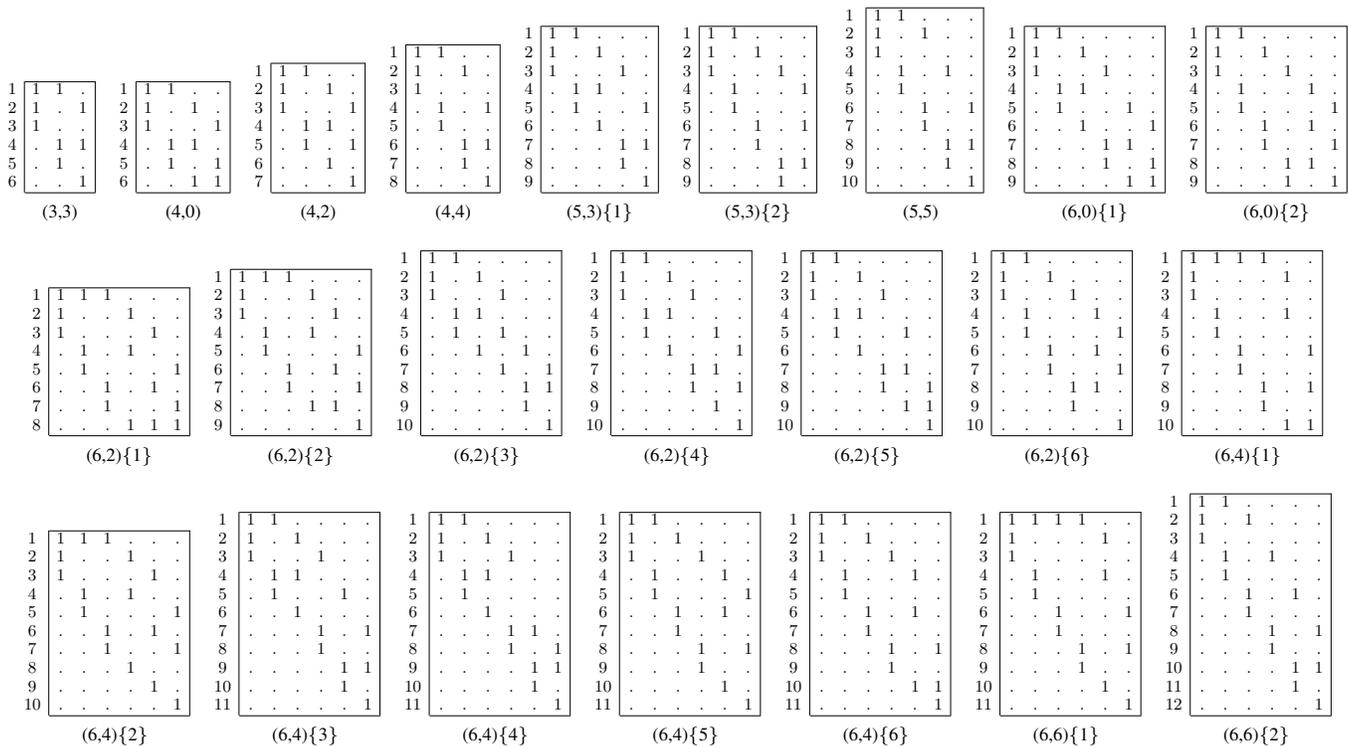

\vspace{1cm}
\captionsetup[subfigure]{labelformat=empty,margin={9pt,0pt},font={scriptsize}}
\newcommand{\lborder}[1]{\multicolumn{1}{;{1pt/2pt}c}{#1}}
\newcommand{\rborder}[1]{\multicolumn{1}{c;{1pt/2pt}}{#1}}
\newcommand{\lrborder}[1]{\multicolumn{1}{;{1pt/2pt}c;{1pt/2pt}}{#1}}
\setlength{\arraycolsep}{4pt}
\renewcommand{\arraystretch}{1}
\def\scaleboxes {0.7}
\centering
\small
\subfloat[(3,3)]{
\scalebox{\scaleboxes}{$
\begin{array}[b]{c|ccc|}
\cline{2-4}
1 & 1 & 1 & .\\
2 & 1 & . & 1\\
3 & 1 & . & .\\
4 & . & 1 & 1\\
5 & . & 1 & .\\
6 & . & . & 1\\
\cline{2-4}
\end{array}
$}
}
\subfloat[(4,0)]{
\scalebox{\scaleboxes}{$
\begin{array}[b]{c|cccc|}
\cline{2-5}
1 & 1 & 1 & . & .\\
2 & 1 & . & 1 & .\\
3 & 1 & . & . & 1\\
4 & . & 1 & 1 & .\\
5 & . & 1 & . & 1\\
6 & . & . & 1 & 1\\
\cline{2-5}
\end{array}
$}
}
\subfloat[(4,2)]{
\scalebox{\scaleboxes}{$
\begin{array}[b]{c|cccc|}
\cline{2-5}
1 & 1 & 1 & . & .\\
2 & 1 & . & 1 & .\\
3 & 1 & . & . & 1\\
4 & . & 1 & 1 & .\\
5 & . & 1 & . & 1\\
6 & . & . & 1 & .\\
7 & . & . & . & 1\\
\cline{2-5}
\end{array}
$}
}
\subfloat[(4,4)]{
\scalebox{\scaleboxes}{$
\begin{array}[b]{c|cccc|}
\cline{2-5}
1 & 1 & 1 & . & .\\
2 & 1 & . & 1 & .\\
3 & 1 & . & . & .\\
4 & . & 1 & . & 1\\
5 & . & 1 & . & .\\
6 & . & . & 1 & 1\\
7 & . & . & 1 & .\\
8 & . & . & . & 1\\
\cline{2-5}
\end{array}
$}
}
\subfloat[(5,3)\{1\}]{
\scalebox{\scaleboxes}{$
\begin{array}[b]{c|ccccc|}
\cline{2-6}
1 & 1 & 1 & . & . & .\\
2 & 1 & . & 1 & . & .\\
3 & 1 & . & . & 1 & .\\
4 & . & 1 & 1 & . & .\\
5 & . & 1 & . & . & 1\\
6 & . & . & 1 & . & .\\
7 & . & . & . & 1 & 1\\
8 & . & . & . & 1 & .\\
9 & . & . & . & . & 1\\
\cline{2-6}
\end{array}
$}
}
\subfloat[(5,3)\{2\}]{
\scalebox{\scaleboxes}{$
\begin{array}[b]{c|ccccc|}
\cline{2-6}
1 & 1 & 1 & . & . & .\\
2 & 1 & . & 1 & . & .\\
3 & 1 & . & . & 1 & .\\
4 & . & 1 & . & . & 1\\
5 & . & 1 & . & . & .\\
6 & . & . & 1 & . & 1\\
7 & . & . & 1 & . & .\\
8 & . & . & . & 1 & 1\\
9 & . & . & . & 1 & .\\
\cline{2-6}
\end{array}
$}
}
\subfloat[(5,5)]{
\scalebox{\scaleboxes}{$
\begin{array}[b]{c|ccccc|}
\cline{2-6}
1 & 1 & 1 & . & . & .\\
2 & 1 & . & 1 & . & .\\
3 & 1 & . & . & . & .\\
4 & . & 1 & . & 1 & .\\
5 & . & 1 & . & . & .\\
6 & . & . & 1 & . & 1\\
7 & . & . & 1 & . & .\\
8 & . & . & . & 1 & 1\\
9 & . & . & . & 1 & .\\
10 & . & . & . & . & 1\\
\cline{2-6}
\end{array}
$}
}
\subfloat[(6,0)\{1\}]{
\scalebox{\scaleboxes}{$
\begin{array}[b]{c|cccccc|}
\cline{2-7}
1 & 1 & 1 & . & . & . & .\\
2 & 1 & . & 1 & . & . & .\\
3 & 1 & . & . & 1 & . & .\\
4 & . & 1 & 1 & . & . & .\\
5 & . & 1 & . & . & 1 & .\\
6 & . & . & 1 & . & . & 1\\
7 & . & . & . & 1 & 1 & .\\
8 & . & . & . & 1 & . & 1\\
9 & . & . & . & . & 1 & 1\\
\cline{2-7}
\end{array}
$}
}
\subfloat[(6,0)\{2\}]{
\scalebox{\scaleboxes}{$
\begin{array}[b]{c|cccccc|}
\cline{2-7}
1 & 1 & 1 & . & . & . & .\\
2 & 1 & . & 1 & . & . & .\\
3 & 1 & . & . & 1 & . & .\\
4 & . & 1 & . & . & 1 & .\\
5 & . & 1 & . & . & . & 1\\
6 & . & . & 1 & . & 1 & .\\
7 & . & . & 1 & . & . & 1\\
8 & . & . & . & 1 & 1 & .\\
9 & . & . & . & 1 & . & 1\\
\cline{2-7}
\end{array}
$}
}

\subfloat[(6,2)\{1\}]{
\scalebox{\scaleboxes}{$
\begin{array}[b]{c|cccccc|}
\cline{2-7}
1 & 1 & 1 & 1 & . & . & .\\
2 & 1 & . & . & 1 & . & .\\
3 & 1 & . & . & . & 1 & .\\
4 & . & 1 & . & 1 & . & .\\
5 & . & 1 & . & . & . & 1\\
6 & . & . & 1 & . & 1 & .\\
7 & . & . & 1 & . & . & 1\\
8 & . & . & . & 1 & 1 & 1\\
\cline{2-7}
\end{array}
$}
}
\subfloat[(6,2)\{2\}]{
\scalebox{\scaleboxes}{$
\begin{array}[b]{c|cccccc|}
\cline{2-7}
1 & 1 & 1 & 1 & . & . & .\\
2 & 1 & . & . & 1 & . & .\\
3 & 1 & . & . & . & 1 & .\\
4 & . & 1 & . & 1 & . & .\\
5 & . & 1 & . & . & . & 1\\
6 & . & . & 1 & . & 1 & .\\
7 & . & . & 1 & . & . & 1\\
8 & . & . & . & 1 & 1 & .\\
9 & . & . & . & . & . & 1\\
\cline{2-7}
\end{array}
$}
}
\subfloat[(6,2)\{3\}]{
\scalebox{\scaleboxes}{$
\begin{array}[b]{c|cccccc|}
\cline{2-7}
1 & 1 & 1 & . & . & . & .\\
2 & 1 & . & 1 & . & . & .\\
3 & 1 & . & . & 1 & . & .\\
4 & . & 1 & 1 & . & . & .\\
5 & . & 1 & . & 1 & . & .\\
6 & . & . & 1 & . & 1 & .\\
7 & . & . & . & 1 & . & 1\\
8 & . & . & . & . & 1 & 1\\
9 & . & . & . & . & 1 & .\\
10 & . & . & . & . & . & 1\\
\cline{2-7}
\end{array}
$}
}
\subfloat[(6,2)\{4\}]{
\scalebox{\scaleboxes}{$
\begin{array}[b]{c|cccccc|}
\cline{2-7}
1 & 1 & 1 & . & . & . & .\\
2 & 1 & . & 1 & . & . & .\\
3 & 1 & . & . & 1 & . & .\\
4 & . & 1 & 1 & . & . & .\\
5 & . & 1 & . & . & 1 & .\\
6 & . & . & 1 & . & . & 1\\
7 & . & . & . & 1 & 1 & .\\
8 & . & . & . & 1 & . & 1\\
9 & . & . & . & . & 1 & .\\
10 & . & . & . & . & . & 1\\
\cline{2-7}
\end{array}
$}
}
\subfloat[(6,2)\{5\}]{
\scalebox{\scaleboxes}{$
\begin{array}[b]{c|cccccc|}
\cline{2-7}
1 & 1 & 1 & . & . & . & .\\
2 & 1 & . & 1 & . & . & .\\
3 & 1 & . & . & 1 & . & .\\
4 & . & 1 & 1 & . & . & .\\
5 & . & 1 & . & . & 1 & .\\
6 & . & . & 1 & . & . & .\\
7 & . & . & . & 1 & 1 & .\\
8 & . & . & . & 1 & . & 1\\
9 & . & . & . & . & 1 & 1\\
10 & . & . & . & . & . & 1\\
\cline{2-7}
\end{array}
$}
}
\subfloat[(6,2)\{6\}]{
\scalebox{\scaleboxes}{$
\begin{array}[b]{c|cccccc|}
\cline{2-7}
1 & 1 & 1 & . & . & . & .\\
2 & 1 & . & 1 & . & . & .\\
3 & 1 & . & . & 1 & . & .\\
4 & . & 1 & . & . & 1 & .\\
5 & . & 1 & . & . & . & 1\\
6 & . & . & 1 & . & 1 & .\\
7 & . & . & 1 & . & . & 1\\
8 & . & . & . & 1 & 1 & .\\
9 & . & . & . & 1 & . & .\\
10 & . & . & . & . & . & 1\\
\cline{2-7}
\end{array}
$}
}
\subfloat[(6,4)\{1\}]{
\scalebox{\scaleboxes}{$
\begin{array}[b]{c|cccccc|}
\cline{2-7}
1 & 1 & 1 & 1 & 1 & . & .\\
2 & 1 & . & . & . & 1 & .\\
3 & 1 & . & . & . & . & .\\
4 & . & 1 & . & . & 1 & .\\
5 & . & 1 & . & . & . & .\\
6 & . & . & 1 & . & . & 1\\
7 & . & . & 1 & . & . & .\\
8 & . & . & . & 1 & . & 1\\
9 & . & . & . & 1 & . & .\\
10 & . & . & . & . & 1 & 1\\
\cline{2-7}
\end{array}
$}
}

\subfloat[(6,4)\{2\}]{
\scalebox{\scaleboxes}{$
\begin{array}[b]{c|cccccc|}
\cline{2-7}
1 & 1 & 1 & 1 & . & . & .\\
2 & 1 & . & . & 1 & . & .\\
3 & 1 & . & . & . & 1 & .\\
4 & . & 1 & . & 1 & . & .\\
5 & . & 1 & . & . & . & 1\\
6 & . & . & 1 & . & 1 & .\\
7 & . & . & 1 & . & . & 1\\
8 & . & . & . & 1 & . & .\\
9 & . & . & . & . & 1 & .\\
10 & . & . & . & . & . & 1\\
\cline{2-7}
\end{array}
$}
}
\subfloat[(6,4)\{3\}]{
\scalebox{\scaleboxes}{$
\begin{array}[b]{c|cccccc|}
\cline{2-7}
1 & 1 & 1 & . & . & . & .\\
2 & 1 & . & 1 & . & . & .\\
3 & 1 & . & . & 1 & . & .\\
4 & . & 1 & 1 & . & . & .\\
5 & . & 1 & . & . & 1 & .\\
6 & . & . & 1 & . & . & .\\
7 & . & . & . & 1 & . & 1\\
8 & . & . & . & 1 & . & .\\
9 & . & . & . & . & 1 & 1\\
10 & . & . & . & . & 1 & .\\
11 & . & . & . & . & . & 1\\
\cline{2-7}
\end{array}
$}
}
\subfloat[(6,4)\{4\}]{
\scalebox{\scaleboxes}{$
\begin{array}[b]{c|cccccc|}
\cline{2-7}
1 & 1 & 1 & . & . & . & .\\
2 & 1 & . & 1 & . & . & .\\
3 & 1 & . & . & 1 & . & .\\
4 & . & 1 & 1 & . & . & .\\
5 & . & 1 & . & . & . & .\\
6 & . & . & 1 & . & . & .\\
7 & . & . & . & 1 & 1 & .\\
8 & . & . & . & 1 & . & 1\\
9 & . & . & . & . & 1 & 1\\
10 & . & . & . & . & 1 & .\\
11 & . & . & . & . & . & 1\\
\cline{2-7}
\end{array}
$}
}
\subfloat[(6,4)\{5\}]{
\scalebox{\scaleboxes}{$
\begin{array}[b]{c|cccccc|}
\cline{2-7}
1 & 1 & 1 & . & . & . & .\\
2 & 1 & . & 1 & . & . & .\\
3 & 1 & . & . & 1 & . & .\\
4 & . & 1 & . & . & 1 & .\\
5 & . & 1 & . & . & . & 1\\
6 & . & . & 1 & . & 1 & .\\
7 & . & . & 1 & . & . & .\\
8 & . & . & . & 1 & . & 1\\
9 & . & . & . & 1 & . & .\\
10 & . & . & . & . & 1 & .\\
11 & . & . & . & . & . & 1\\
\cline{2-7}
\end{array}
$}
}
\subfloat[(6,4)\{6\}]{
\scalebox{\scaleboxes}{$
\begin{array}[b]{c|cccccc|}
\cline{2-7}
1 & 1 & 1 & . & . & . & .\\
2 & 1 & . & 1 & . & . & .\\
3 & 1 & . & . & 1 & . & .\\
4 & . & 1 & . & . & 1 & .\\
5 & . & 1 & . & . & . & .\\
6 & . & . & 1 & . & 1 & .\\
7 & . & . & 1 & . & . & .\\
8 & . & . & . & 1 & . & 1\\
9 & . & . & . & 1 & . & .\\
10 & . & . & . & . & 1 & 1\\
11 & . & . & . & . & . & 1\\
\cline{2-7}
\end{array}
$}
}
\subfloat[(6,6)\{1\}]{
\scalebox{\scaleboxes}{$
\begin{array}[b]{c|cccccc|}
\cline{2-7}
1 & 1 & 1 & 1 & 1 & . & .\\
2 & 1 & . & . & . & 1 & .\\
3 & 1 & . & . & . & . & .\\
4 & . & 1 & . & . & 1 & .\\
5 & . & 1 & . & . & . & .\\
6 & . & . & 1 & . & . & 1\\
7 & . & . & 1 & . & . & .\\
8 & . & . & . & 1 & . & 1\\
9 & . & . & . & 1 & . & .\\
10 & . & . & . & . & 1 & .\\
11 & . & . & . & . & . & 1\\
\cline{2-7}
\end{array}
$}
}
\subfloat[(6,6)\{2\}]{
\scalebox{\scaleboxes}{$
\begin{array}[b]{c|cccccc|}
\cline{2-7}
1 & 1 & 1 & . & . & . & .\\
2 & 1 & . & 1 & . & . & .\\
3 & 1 & . & . & . & . & .\\
4 & . & 1 & . & 1 & . & .\\
5 & . & 1 & . & . & . & .\\
6 & . & . & 1 & . & 1 & .\\
7 & . & . & 1 & . & . & .\\
8 & . & . & . & 1 & . & 1\\
9 & . & . & . & 1 & . & .\\
10 & . & . & . & . & 1 & 1\\
11 & . & . & . & . & 1 & .\\
12 & . & . & . & . & . & 1\\
\cline{2-7}
\end{array}
$}
}
\caption{Matrix representation of $(a,b)$ absorbing set candidates with $a\leq 6$. All depicted candidates are 3-colourable. If there are multiple candidates of the same size $(a,b)$, we use a postfix $\{i\}$ to determine an order. The dots represent zero-entries.}
\vspace{1cm}
\label{fig:absorbing_set_classification_k3}
\end{figure*}

\begin{figure*}[t!]
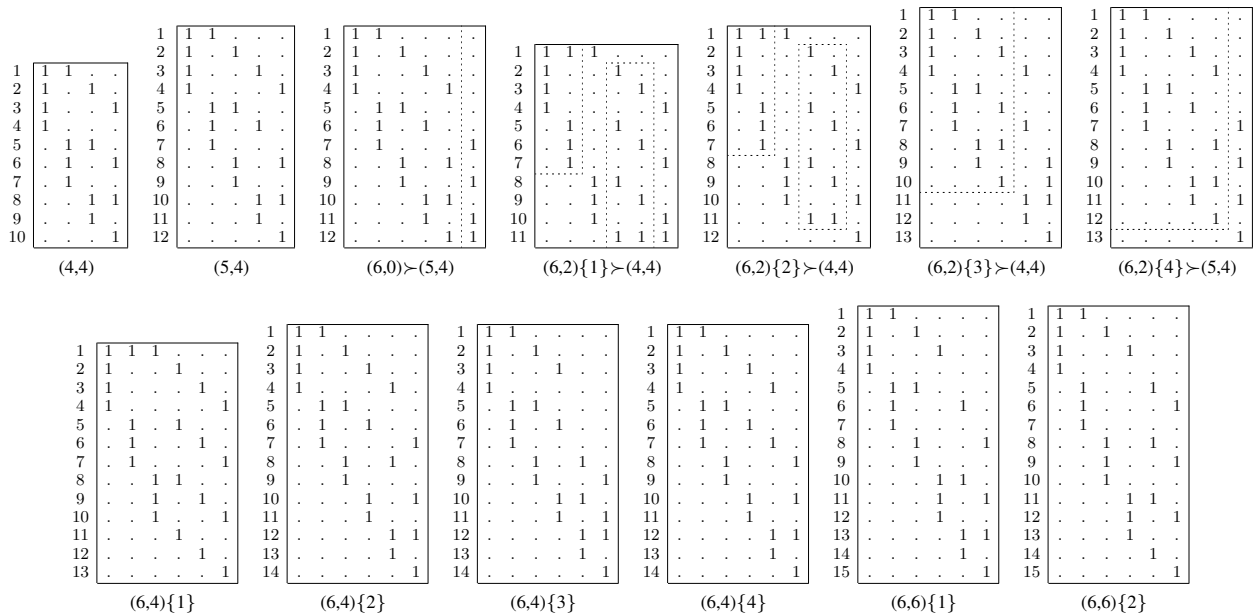

\captionsetup[subfigure]{labelformat=empty,margin={9pt,0pt},font={scriptsize}}
\newcommand{\lborder}[1]{\multicolumn{1}{;{1pt/2pt}c}{#1}}
\newcommand{\rborder}[1]{\multicolumn{1}{c;{1pt/2pt}}{#1}}
\newcommand{\lrborder}[1]{\multicolumn{1}{;{1pt/2pt}c;{1pt/2pt}}{#1}}
\setlength{\arraycolsep}{4pt}
\renewcommand{\arraystretch}{1}
\def\scaleboxes {0.7}
\centering
\small
\subfloat[(4,4)]{
	\scalebox{\scaleboxes}{$
	\begin{array}[b]{c|cccc|}
	\cline{2-5}
	1 & 1 & 1 & . & .\\
	2 & 1 & . & 1 & .\\
	3 & 1 & . & . & 1\\
	4 & 1 & . & . & .\\
	5 & . & 1 & 1 & .\\
	6 & . & 1 & . & 1\\
	7 & . & 1 & . & .\\
	8 & . & . & 1 & 1\\
	9 & . & . & 1 & .\\
	10 & . & . & . & 1\\
	\cline{2-5}
	\end{array}
	$}
}
\subfloat[(5,4)]{
	\scalebox{\scaleboxes}{$
	\begin{array}[b]{c|ccccc|}
	\cline{2-6}
	1 & 1 & 1 & . & . & .\\
	2 & 1 & . & 1 & . & .\\
	3 & 1 & . & . & 1 & .\\
	4 & 1 & . & . & . & 1\\
	5 & . & 1 & 1 & . & .\\
	6 & . & 1 & . & 1 & .\\
	7 & . & 1 & . & . & .\\
	8 & . & . & 1 & . & 1\\
	9 & . & . & 1 & . & .\\
	10 & . & . & . & 1 & 1\\
	11 & . & . & . & 1 & .\\
	12 & . & . & . & . & 1\\
  	\cline{2-6}
	\end{array}
	$}
}
\subfloat[(6,0)$\succ$(5,4)]{
	\scalebox{\scaleboxes}{$
	\begin{array}[b]{c|ccccc;{1pt/2pt}c|}
	\cline{2-7}
	1 & 1 & 1 & . & . & . & .\\
	2 & 1 & . & 1 & . & . & .\\
	3 & 1 & . & . & 1 & . & .\\
	4 & 1 & . & . & . & 1 & .\\
	5 & . & 1 & 1 & . & . & .\\
	6 & . & 1 & . & 1 & . & .\\
	7 & . & 1 & . & . & . & 1\\
	8 & . & . & 1 & . & 1 & .\\
	9 & . & . & 1 & . & . & 1\\
	10 & . & . & . & 1 & 1 & .\\
	11 & . & . & . & 1 & . & 1\\
	12 & . & . & . & . & 1 & 1\\
	\cline{2-7}
	\end{array}
	$}
}
\subfloat[(6,2)\{1\}$\succ$(4,4)]{
	\scalebox{\scaleboxes}{$
	\begin{array}[b]{c|cccccc|}
	\cline{2-7}
	1 & 1 & \rborder{1} & 1 & . & . & .\\
	\cdashline{5-6}[1pt/2pt]
	2 & 1 & \rborder{.} & . & \lborder{1} & \rborder{.} & .\\
	3 & 1 & \rborder{.} & . & \lborder{.} & \rborder{1} & .\\
	4 & 1 & \rborder{.} & . & \lborder{.} & \rborder{.} & 1\\
	5 & . & \rborder{1} & . & \lborder{1} & \rborder{.} & .\\
	6 & . & \rborder{1} & . & \lborder{.} & \rborder{1} & .\\
	7 & . & \rborder{1} & . & \lborder{.} & \rborder{.} & 1\\
	\cdashline{2-3}[1pt/2pt]
	8 & . & . & 1 & \lborder{1} & \rborder{.} & .\\
	9 & . & . & 1 & \lborder{.} & \rborder{1} & .\\
	10 & . & . & 1 & \lborder{.} & \rborder{.} & 1\\
	11 & . & . & . & \lborder{1} & \rborder{1} & 1\\
	\cline{2-7}
	\end{array}
	$}
}
\subfloat[(6,2)\{2\}$\succ$(4,4)]{
	\scalebox{\scaleboxes}{$
	\begin{array}[b]{c|cccccc|}
	\cline{2-7}
	1 & 1 & \rborder{1} & 1 & . & . & .\\
	\cdashline{5-6}[1pt/2pt]
	2 & 1 & \rborder{.} & . & \lborder{1} & \rborder{.} & .\\
	3 & 1 & \rborder{.} & . & \lborder{.} & \rborder{1} & .\\
	4 & 1 & \rborder{.} & . & \lborder{.} & \rborder{.} & 1\\
	5 & . & \rborder{1} & . & \lborder{1} & \rborder{.} & .\\
	6 & . & \rborder{1} & . & \lborder{.} & \rborder{1} & .\\
	7 & . & \rborder{1} & . & \lborder{.} & \rborder{.} & 1\\
	\cdashline{2-3}[1pt/2pt]
	8 & . & . & 1 & \lborder{1} & \rborder{.} & .\\
	9 & . & . & 1 & \lborder{.} & \rborder{1} & .\\
	10 & . & . & 1 & \lborder{.} & \rborder{.} & 1\\
	11 & . & . & . & \lborder{1} & \rborder{1} & .\\
	\cdashline{5-6}[1pt/2pt]
	12 & . & . & . & . & . & 1\\
	\cline{2-7}
	\end{array}
	$}
}
\subfloat[(6,2)\{3\}$\succ$(4,4)]{
	\scalebox{\scaleboxes}{$
	\begin{array}[b]{c|cccccc|}
	\cline{2-7}
	1 & 1 & 1 & . & \rborder{.} & . & .\\
	2 & 1 & . & 1 & \rborder{.} & . & .\\
	3 & 1 & . & . & \rborder{1} & . & .\\
	4 & 1 & . & . & \rborder{.} & 1 & .\\
	5 & . & 1 & 1 & \rborder{.} & . & .\\
	6 & . & 1 & . & \rborder{1} & . & .\\
	7 & . & 1 & . & \rborder{.} & 1 & .\\
	8 & . & . & 1 & \rborder{1} & . & .\\
	9 & . & . & 1 & \rborder{.} & . & 1\\
	10 & . & . & . & \rborder{1} & . & 1\\
	\cdashline{2-5}[1pt/2pt]
	11 & . & . & . & . & 1 & 1\\
	12 & . & . & . & . & 1 & .\\
	13 & . & . & . & . & . & 1\\
	\cline{2-7}
	\end{array}
	$}
}
\subfloat[(6,2)\{4\}$\succ$(5,4)]{
	\scalebox{\scaleboxes}{$
	\begin{array}[b]{c|cccccc|}
	\cline{2-7}
	1 & 1 & 1 & . & . & \rborder{.} & .\\
	2 & 1 & . & 1 & . & \rborder{.} & .\\
	3 & 1 & . & . & 1 & \rborder{.} & .\\
	4 & 1 & . & . & . & \rborder{1} & .\\
	5 & . & 1 & 1 & . & \rborder{.} & .\\
	6 & . & 1 & . & 1 & \rborder{.} & .\\
	7 & . & 1 & . & . & \rborder{.} & 1\\
	8 & . & . & 1 & . & \rborder{1} & .\\
	9 & . & . & 1 & . & \rborder{.} & 1\\
	10 & . & . & . & 1 & \rborder{1} & .\\
	11 & . & . & . & 1 & \rborder{.} & 1\\
	12 & . & . & . & . & \rborder{1} & .\\
	\cdashline{2-6}[1pt/2pt]
	13 & . & . & . & . & . & 1\\
	\cline{2-7}
	\end{array}
	$}
}

\subfloat[(6,4)\{1\}]{
	\scalebox{\scaleboxes}{$
	\begin{array}[b]{c|cccccc|}
	\cline{2-7}
	1 & 1 & 1 & 1 & . & . & .\\
	2 & 1 & . & . & 1 & . & .\\
	3 & 1 & . & . & . & 1 & .\\
	4 & 1 & . & . & . & . & 1\\
	5 & . & 1 & . & 1 & . & .\\
	6 & . & 1 & . & . & 1 & .\\
	7 & . & 1 & . & . & . & 1\\
	8 & . & . & 1 & 1 & . & .\\
	9 & . & . & 1 & . & 1 & .\\
	10 & . & . & 1 & . & . & 1\\
	11 & . & . & . & 1 & . & .\\
	12 & . & . & . & . & 1 & .\\
	13 & . & . & . & . & . & 1\\
	\cline{2-7}
	\end{array}
	$}
}
\subfloat[(6,4)\{2\}]{
	\scalebox{\scaleboxes}{$
	\begin{array}[b]{c|cccccc|}
	\cline{2-7}
	1 & 1 & 1 & . & . & . & .\\
	2 & 1 & . & 1 & . & . & .\\
	3 & 1 & . & . & 1 & . & .\\
	4 & 1 & . & . & . & 1 & .\\
	5 & . & 1 & 1 & . & . & .\\
	6 & . & 1 & . & 1 & . & .\\
	7 & . & 1 & . & . & . & 1\\
	8 & . & . & 1 & . & 1 & .\\
	9 & . & . & 1 & . & . & .\\
	10 & . & . & . & 1 & . & 1\\
	11 & . & . & . & 1 & . & .\\
	12 & . & . & . & . & 1 & 1\\
	13 & . & . & . & . & 1 & .\\
	14 & . & . & . & . & . & 1\\
	\cline{2-7}
	\end{array}
	$}
}
\subfloat[(6,4)\{3\}]{
	\scalebox{\scaleboxes}{$
	\begin{array}[b]{c|cccccc|}
	\cline{2-7}
	1 & 1 & 1 & . & . & . & .\\
	2 & 1 & . & 1 & . & . & .\\
	3 & 1 & . & . & 1 & . & .\\
	4 & 1 & . & . & . & . & .\\
	5 & . & 1 & 1 & . & . & .\\
	6 & . & 1 & . & 1 & . & .\\
	7 & . & 1 & . & . & . & .\\
	8 & . & . & 1 & . & 1 & .\\
	9 & . & . & 1 & . & . & 1\\
	10 & . & . & . & 1 & 1 & .\\
	11 & . & . & . & 1 & . & 1\\
	12 & . & . & . & . & 1 & 1\\
	13 & . & . & . & . & 1 & .\\
	14 & . & . & . & . & . & 1\\
	\cline{2-7}
	\end{array}
	$}
}
\subfloat[(6,4)\{4\}]{
	\scalebox{\scaleboxes}{$
	\begin{array}[b]{c|cccccc|}
	\cline{2-7}
	1 & 1 & 1 & . & . & . & .\\
	2 & 1 & . & 1 & . & . & .\\
	3 & 1 & . & . & 1 & . & .\\
	4 & 1 & . & . & . & 1 & .\\
	5 & . & 1 & 1 & . & . & .\\
	6 & . & 1 & . & 1 & . & .\\
	7 & . & 1 & . & . & 1 & .\\
	8 & . & . & 1 & . & . & 1\\
	9 & . & . & 1 & . & . & .\\
	10 & . & . & . & 1 & . & 1\\
	11 & . & . & . & 1 & . & .\\
	12 & . & . & . & . & 1 & 1\\
	13 & . & . & . & . & 1 & .\\
	14 & . & . & . & . & . & 1\\
	\cline{2-7}
	\end{array}
	$}
}
\subfloat[(6,6)\{1\}]{
	\scalebox{\scaleboxes}{$
	\begin{array}[b]{c|cccccc|}
	\cline{2-7}
	1 & 1 & 1 & . & . & . & .\\
	2 & 1 & . & 1 & . & . & .\\
	3 & 1 & . & . & 1 & . & .\\
	4 & 1 & . & . & . & . & .\\
	5 & . & 1 & 1 & . & . & .\\
	6 & . & 1 & . & . & 1 & .\\
	7 & . & 1 & . & . & . & .\\
	8 & . & . & 1 & . & . & 1\\
	9 & . & . & 1 & . & . & .\\
	10 & . & . & . & 1 & 1 & .\\
	11 & . & . & . & 1 & . & 1\\
	12 & . & . & . & 1 & . & .\\
	13 & . & . & . & . & 1 & 1\\
	14 & . & . & . & . & 1 & .\\
	15 & . & . & . & . & . & 1\\
	\cline{2-7}
	\end{array}
	$}
}
\subfloat[(6,6)\{2\}]{
	\scalebox{\scaleboxes}{$
	\begin{array}[b]{c|cccccc|}
	\cline{2-7}
	1 & 1 & 1 & . & . & . & .\\
	2 & 1 & . & 1 & . & . & .\\
	3 & 1 & . & . & 1 & . & .\\
	4 & 1 & . & . & . & . & .\\
	5 & . & 1 & . & . & 1 & .\\
	6 & . & 1 & . & . & . & 1\\
	7 & . & 1 & . & . & . & .\\
	8 & . & . & 1 & . & 1 & .\\
	9 & . & . & 1 & . & . & 1\\
	10 & . & . & 1 & . & . & .\\
	11 & . & . & . & 1 & 1 & .\\
	12 & . & . & . & 1 & . & 1\\
	13 & . & . & . & 1 & . & .\\
	14 & . & . & . & . & 1 & .\\
	15 & . & . & . & . & . & 1\\
	\cline{2-7}
	\end{array}
	$}
}
\caption{Matrix representation of $(a,b)$ absorbing set candidates with $a\leq 6$ that may represent an absorbing set in a TD LDPC code of column weight $4$. All depicted candidates are 4-colourable. If there are multiple candidates of the same size $(a,b)$, we use a postfix $\{i\}$ to determine an order. The symbol ``$\succ$'' means that the left-hand candidate is an extension of the right-hand candidate. The dots represent zero-entries.}
\vspace{1cm}
\label{fig:absorbing_set_classification}
\end{figure*}

\subsection{Classification Process}

In order to classificate the absorbing set candidates in TD LDPC codes based on transversal designs with block size $k$, we have written a program that outputs an exhaustive list of non-isomorphic set systems of block size $k$ that satisfy the combinatorial constraints (A)-(D) of Lemma~\ref{lemma:combinatorial_constraints_of_absorbing_sets} (up to a given size of $t$ blocks).

\emph{Approach:}
By starting with an empty set system, we successively extend the current system by a further block (in all possible ways) in compliance with some combinatorial rules that are necessary to build up an absorbing set candidate. If an extension fulfills all constraints of Lemma~\ref{lemma:combinatorial_constraints_of_absorbing_sets}, we add it to the list of absorbing set candidates. We continue until a maximum number of $t$ blocks is reached.

\noindent\rule[2pt]{254pt}{0.5pt}\\
\noindent $\mathcal{O}$ = \textbf{Classification}$(k,t)$\\
\noindent\rule[4pt]{254pt}{0.5pt}

\noindent\textsc{Input}: 
\begin{itemize}[leftmargin=18pt]
\item $k$: desired block size of the absorbing set candidates
\item $t$: maximum number of blocks
\end{itemize}

\noindent\textsc{Output}: 
\begin{itemize}[leftmargin=18pt]
\item $\mathcal{O}$: 
The output list $\mathcal{O}$ fills up with all non-isomorphic absorbing set candidates of block size $k$ and at most $t$ blocks that may occur in any TD LDPC code of column weight $k$. The presence or absence of these absorbing sets finally depends on the specific structure of a concrete TD LDPC code.
\end{itemize}

\noindent\textsc{Notations and Invariants:}
\begin{itemize}[leftmargin=18pt]
\item We denote the current set system by $\mathcal{S}$.
\item A set system $(\mathcal{P}',\mathcal{B}')$ is called an \emph{extension} of $\mathcal{S}=(\mathcal{P},\mathcal{B})$, if $\mathcal{P}\subseteq \mathcal{P}'$ and  $\mathcal{B}\subseteq \mathcal{B}'$. 
\item Let $\mathcal{E}_\varpi$, $1\leq \varpi \leq t$, be $t$ global lists of non-isomorphic extensions of size $\varpi$ that have already been processed. Note that the set systems collected in $\mathcal{E}_\varpi$ are not necessarily absorbing set candidates. 
\end{itemize}

\noindent\textsc{Algorithm:}
\begin{enumerate}[label=(\arabic*),leftmargin=18pt]

\item \emph{Initialization}: We start with an empty set system $\mathcal{S}$. Define
$\mathcal{E}_\varpi=\varnothing$ for $1\leq \varpi \leq t$ and $\mathcal{O}=\varnothing$.

\item \emph{Find extensions of $\mathcal{S}$}:
We extend $\mathcal{S}$ by a further block of size $k$ in all possible ways. Let $\mathcal{S}_1, \mathcal{S}_2,\hdots, \mathcal{S}_\mu$ be the non-isomorphic extensions of $\mathcal{S}$ with the following restrictions:
\begin{enumerate}[label=\alph*),leftmargin=15pt]
\item The constraints (A)-(C) of Lemma~\ref{lemma:combinatorial_constraints_of_absorbing_sets} must be valid.
\item $\mathcal{S}$ must be connected, i.e, for any two distinct subsets of blocks (called components) there must be at least one point that is contained in both components. 
\end{enumerate}
Note that any violation of a) is irreparable by extending the set system such that the concerned extensions can be discarded. 
By contrast, the constraint b) is reparable by adding some new blocks properly such that the isolated components get connected. Nonetheless, we may discard unconnected set systems at this early stage since all extensions of these systems will be found by extending the blocks in a different order.
All constraints therefore reduce the processing complexity.

\item \emph{For each extension $\mathcal{S}_i$ do}:
\begin{enumerate}[label=\alph*),leftmargin=15pt,topsep=0pt]
\item Let $\varpi$ be the number of blocks of $\mathcal{S}_i$. If $\mathcal{S}_i$ is isomorphic to any set system of $\mathcal{E}_\varpi$, we discard $\mathcal{S}_i$ and continue with the next extension, else we add $\mathcal{S}_i$ to $\mathcal{E}_\varpi$.
\item We add $\mathcal{S}_i$ to the output $\mathcal{O}$, if
$\mathcal{S}_i$ satisfies constraint (D) of Lemma~\ref{lemma:combinatorial_constraints_of_absorbing_sets}.
\item If $\varpi=t$, we stop processing this extension, else we apply steps (2)-(3) recursively to $\mathcal{S}:= \mathcal{S}_i$.
\end{enumerate}
Note that the violation of constraint (D) is reparable by adding some new blocks properly, such that the extensions must be processed recursively even if they violate (D).

\end{enumerate}

\subsection{Classification for Column Weight $k=3$}

Fig.~\ref{fig:absorbing_set_classification_k3} shows the matrix representations of the $(a,b)$ absorbing set candidates which has been obtained by the procedure
\[
\mathcal{O} = \text{\textbf{Classification}} (3,6).
\]

\subsection{Classification for Column Weight $k=4$}

Fig.~\ref{fig:absorbing_set_classification} shows the matrix representations of the $(a,b)$ absorbing set candidates which has been obtained by the procedure
\[
\mathcal{O} = \text{\textbf{Classification}} (4,6).
\]

\section{TD LDPC Codes Based on MOLS\\ with Cyclic Structure}\label{codes_on_simple_structured_MOLS}

For the remainder of the paper, we employ a special class of cyclic-structured MOLS that are ideally suited for constructing TD LDPC codes with excellent decoding performances. More precisely, the presented MOLS provide an algebraic approach for the investigation and elimination of harmful absorbing sets, leading to codes with low error-floors and thus to near-optimal performances over the AWGN channel via SPA decoding. In \cite{GrunHub2013}, we have already used the same class of MOLS to produce powerful TD LDPC codes with improved stopping set distributions over the BEC. 

\subsection{Construction of Latin Squares with Cyclic Structure}

We first need the following straightforward lemma, giving mutually orthogonal Latin squares isomorphic to Cayley addition tables (cf.~\cite{crc_handbook}):

\begin{lemma}\label{simple_structured_MOLS}
Let $\mathbb{F}_q$ be the Galois field of any prime power order $q$. We obtain a Latin square $\CMcal{L}^{(\alpha,\beta)}_q$ of order $q$ and \emph{scale factors} $\alpha,\beta \in \mathbb{F}_{q}^{*}= \mathbb{F}_q\setminus\{0\}$ by
\[
\CMcal{L}^{(\alpha,\beta)}_q [x,y] = \alpha x + \beta y,\ x,y \in\mathbb{F}_q,
\]
with row set, column set and symbol set $X=Y=S=\mathbb{F}_q$, respectively.
If $\beta=1$, we simply write $\CMcal{L}^{(\alpha)}_q$ instead of $\CMcal{L}^{(\alpha,1)}_q$.  
The cyclic nature of the Latin square can be comprehended by considering that each column is obtained from the previous one by adding the same difference to each element of the previous column modulo $q$. The same holds for the rows.
In the remainder of the paper we consequently speak of cyclic Latin squares and of cyclic MOLS.

\subsection{Cyclic MOLS}
\label{cyclic_MOLS}

Now, we describe under which conditions sets of cyclic Latin squares are MOLS. We first need the following Lemma:

\begin{lemma} \label{lemma:orthogonality_latin_squares}
Two cyclic Latin squares $\CMcal{L}^{(\alpha_1,\beta_1)}_q$ and $\CMcal{L}^{(\alpha_2,\beta_2)}_q$ are orthogonal if and only if $\alpha_1 \beta_2 \neq \alpha_2 \beta_1$ over $\mathbb{F}_q$.
\end{lemma}
\begin{proof}
See Appendix~\ref{proof:orthogonality_of_latin_squares}.
\end{proof}
\end{lemma}

Clearly, every Latin square $\CMcal{L}^{(\alpha,\beta)}_q$ can be associated with a pair of scale factors $(\alpha,\beta)\in (\mathbb{F}_q^*)^2$. Now, we define an equivalence relation $\sim$ on $(\mathbb{F}_q^*)^2$, where $(u,v)\sim (u',v')$ if and only if $uv' = u'v$ over $\mathbb{F}_q$.
Let $\mathbb{F}_q^*=\{\phi_1,\phi_2,\hdots,\phi_{q-1}\}$. The equivalence classes of $\sim$ are given by
$\CMcal{U}_i:= \{(u,v)\in (\mathbb{F}_q^*)^2 : (u,v)\sim (\phi_i,1) \}=\{(x\phi_i, x) : x\in \mathbb{F}_q^* \}$ for $1\leq i\leq q-1$, where $(\phi_i,1)$ is a representative of the $i$-th class.
The classes $\CMcal{U}_i$ partition the set $(\mathbb{F}_q^*)^2$. 

\begin{theorem}
Let $(u_1,v_1),\hdots,(u_{q-1},v_{q-1})$ be any representative system, i.e., $(u_i,v_i)\in \CMcal{U}_i$. Then, the associated Latin squares $\{\CMcal{L}_q^{(u_i,v_i)} : 1\leq i\leq q-1\}$ are $q-1$ MOLS. We can use any $m$-subset of these MOLS ($1\leq m \leq q-1$) to build up a TD$(m+2,q)$ and thus to construct a TD LDPC code.
\end{theorem}
\begin{proof}
Since $(u_i,v_i)$ and $(u_j,v_j)$ for $i\neq j$ are in different equivalence classes, it follows that $u_i v_j \neq u_j v_i$. Hence, the associated Latin squares are orthogonal by Lemma~\ref{lemma:orthogonality_latin_squares}.
\end{proof}

\begin{figure}[!t]
\renewcommand{\arraystretch}{1}
\setlength{\arraycolsep}{3.5pt}
\centering
\begin{tikzpicture}[xscale=2.2,yscale=2.3]
\path [draw] (0.55,0.45) rectangle (1.45,-3.55);
\path [draw] (1.55,0.45) rectangle (2.45,-3.55);
\path [draw] (2.55,0.45) rectangle (3.45,-3.55);
\path [draw] (3.55,0.45) rectangle (4.45,-3.55);

\draw[line width=0.5pt, dashed] (0.52,0.48) -- ++(3.96,0) -- ++(0,-1.05) -- ++(-3.96,0) -- ++(0,1.05);

\node[scale=0.8] at (1,0) {
		$\begin{array}[b]{|ccccc|}
		\hline
		0&1&2&3&4\\
		1&2&3&4&0\\
		2&3&4&0&1\\
		3&4&0&1&2\\
		4&0&1&2&3\\
		\hline
		\end{array}$
};
\node[scale=0.8] at (1,-0.46) {$\CMcal{L}^{(1,1)}_5$};

\node[scale=0.8] at (2,0) {
		$\begin{array}[b]{|ccccc|}
		\hline
		0&1&2&3&4\\
		2&3&4&0&1\\
		4&0&1&2&3\\
		1&2&3&4&0\\
		3&4&0&1&2\\
		\hline
		\end{array}$
};
\node[scale=0.8] at (2,-0.46) {$\CMcal{L}^{(2,1)}_5$};

\node[scale=0.8] at (3,0) {
		$\begin{array}[b]{|ccccc|}
		\hline
		0&1&2&3&4\\
		3&4&0&1&2\\
		1&2&3&4&0\\
		4&0&1&2&3\\
		2&3&4&0&1\\
		\hline
		\end{array}$
};
\node[scale=0.8] at (3,-0.46) {$\CMcal{L}^{(3,1)}_5$};

\node[scale=0.8] at (4,0) {
		$\begin{array}[b]{|ccccc|}
		\hline
		0&1&2&3&4\\
		4&0&1&2&3\\
		3&4&0&1&2\\
		2&3&4&0&1\\
		1&2&3&4&0\\
		\hline
		\end{array}$
};
\node[scale=0.8] at (4,-0.46) {$\CMcal{L}^{(4,1)}_5$};

\node[scale=0.8] at (1,-1) {
		$\begin{array}[b]{|ccccc|}
		\hline
		0&2&4&1&3\\
		2&4&1&3&0\\
		4&1&3&0&2\\
		1&3&0&2&4\\
		3&0&2&4&1\\
		\hline
		\end{array}$
};
\node[scale=0.8] at (1,-1.46) {$\CMcal{L}^{(2,2)}_5$};

\node[scale=0.8] at (2,-1) {
		$\begin{array}[b]{|ccccc|}
		\hline
		0&2&4&1&3\\
		4&1&3&0&2\\
		3&0&2&4&1\\
		2&4&1&3&0\\
		1&3&0&2&4\\
		\hline
		\end{array}$
};
\node[scale=0.8] at (2,-1.46) {$\CMcal{L}^{(4,2)}_5$};

\node[scale=0.8] at (3,-1) {
		$\begin{array}[b]{|ccccc|}
		\hline
		0&2&4&1&3\\
		1&3&0&2&4\\
		2&4&1&3&0\\
		3&0&2&4&1\\
		4&1&3&0&2\\
		\hline
		\end{array}$
};
\node[scale=0.8] at (3,-1.46) {$\CMcal{L}^{(1,2)}_5$};

\node[scale=0.8] at (4,-1) {
		$\begin{array}[b]{|ccccc|}
		\hline
		0&2&4&1&3\\
		3&0&2&4&1\\
		1&3&0&2&4\\
		4&1&3&0&2\\
		2&4&1&3&0\\
		\hline
		\end{array}$
};
\node[scale=0.8] at (4,-1.46) {$\CMcal{L}^{(3,2)}_5$};

\node[scale=0.8] at (1,-2) {
		$\begin{array}[b]{|ccccc|}
		\hline
		0&3&1&4&2\\
		3&1&4&2&0\\
		1&4&2&0&3\\
		4&2&0&3&1\\
		2&0&3&1&4\\
		\hline
		\end{array}$
};
\node[scale=0.8] at (1,-2.46) {$\CMcal{L}^{(3,3)}_5$};

\node[scale=0.8] at (2,-2) {
		$\begin{array}[b]{|ccccc|}
		\hline
		0&3&1&4&2\\
		1&4&2&0&3\\
		2&0&3&1&4\\
		3&1&4&2&0\\
		4&2&0&3&1\\
		\hline
		\end{array}$
};
\node[scale=0.8] at (2,-2.46) {$\CMcal{L}^{(1,3)}_5$};

\node[scale=0.8] at (3,-2) {
		$\begin{array}[b]{|ccccc|}
		\hline
		0&3&1&4&2\\
		4&2&0&3&1\\
		3&1&4&2&0\\
		2&0&3&1&4\\
		1&4&2&0&3\\
		\hline
		\end{array}$
};
\node[scale=0.8] at (3,-2.46) {$\CMcal{L}^{(4,3)}_5$};

\node[scale=0.8] at (4,-2) {
		$\begin{array}[b]{|ccccc|}
		\hline
		0&3&1&4&2\\
		2&0&3&1&4\\
		4&2&0&3&1\\
		1&4&2&0&3\\
		3&1&4&2&0\\
		\hline
		\end{array}$
};
\node[scale=0.8] at (4,-2.46) {$\CMcal{L}^{(2,3)}_5$};

\node[scale=0.8] at (1,-3) {
		$\begin{array}[b]{|ccccc|}
		\hline
		0&4&3&2&1\\
		4&3&2&1&0\\
		3&2&1&0&4\\
		2&1&0&4&3\\
		1&0&4&3&2\\
		\hline
		\end{array}$
};
\node[scale=0.8] at (1,-3.46) {$\CMcal{L}^{(4,4)}_5$};

\node[scale=0.8] at (2,-3) {
		$\begin{array}[b]{|ccccc|}
		\hline
		0&4&3&2&1\\
		3&2&1&0&4\\
		1&0&4&3&2\\
		4&3&2&1&0\\
		2&1&0&4&3\\
		\hline
		\end{array}$
};
\node[scale=0.8] at (2,-3.46) {$\CMcal{L}^{(3,4)}_5$};

\node[scale=0.8] at (3,-3) {
		$\begin{array}[b]{|ccccc|}
		\hline
		0&4&3&2&1\\
		2&1&0&4&3\\
		4&3&2&1&0\\
		1&0&4&3&2\\
		3&2&1&0&4\\
		\hline
		\end{array}$
};
\node[scale=0.8] at (3,-3.46) {$\CMcal{L}^{(2,4)}_5$};

\node[scale=0.8] at (4,-3) {
		$\begin{array}[b]{|ccccc|}
		\hline
		0&4&3&2&1\\
		1&0&4&3&2\\
		2&1&0&4&3\\
		3&2&1&0&4\\
		4&3&2&1&0\\
		\hline
		\end{array}$
};
\node[scale=0.8] at (4,-3.46) {$\CMcal{L}^{(1,4)}_5$};
\end{tikzpicture}
\caption{The solid-line boxes represent the Latin squares associated with the equivalence classes $\CMcal{U}_1,\hdots,\CMcal{U}_4$ over $(\mathbb{F}_5^*)^2$. The Latin squares of the first row, bordered by a dashed line, constitute a possible representative system of these classes and thus are MOLS.}
\label{equivclasses}
\end{figure}

\begin{example}
Fig. \ref{equivclasses} depicts the four equivalence classes over $(\mathbb{F}_5^*)^2$, where the members of the equivalence classes are represented by the associated Latin squares. The Latin squares of the first row, bordered by a dashed line, constitute a possible representative system of these classes and thus are MOLS. 
\end{example}

\subsection{Transversal Designs Based on Cyclic MOLS}
\label{transversal_design_based_on_cyclic_MOLS}

Let $\CMcal{L}^{(\alpha_1,\beta_1)}_q,\hdots,\CMcal{L}^{(\alpha_m,\beta_m)}_q$ be $m$ cyclic MOLS with row, column and symbol sets $X=Y=S_i=\mathbb{F}_q$ for $1\leq i\leq m$. First, define $G_1:= X\times \{1\}$, $G_2:= Y\times \{2\}$ and $G_{i+2}:= S_i\times \{i+2\}$.
By applying the process given in the proof of Theorem~\ref{equiv_MOLS_TDs}, we obtain a transversal design $(\CMcal{P},\CMcal{G},\CMcal{B})$ with point set $\CMcal{P}=\{G_1 \cup\hdots\cup G_{m+2}\}$, groups $\CMcal{G}=\{G_1,\hdots, G_{m+2}\}$ and block set
$\CMcal{B}=\big\{\{(x,1),(y,2),(s_1,3), \hdots, (s_m,m+2)\}: x,y \in \mathbb{F}_q,\ s_1:= \alpha_1 x + \beta_1 y,\ s_m:= \alpha_m x + \beta_m y\big\}$.

\begin{example}
Let $\CMcal{L}^{(1,1)}_q,\CMcal{L}^{(2,1)}_q$ be two MOLS with $X=Y=S_1=S_2=\mathbb{F}_q$. We obtain a transversal design with 
$\CMcal{P}=\mathbb{F}_q \times \{1,2,3,4\}$, groups $\CMcal{G}=\big\{\mathbb{F}_q\times \{i\}: i=1,2,3,4 \big\}$ and blocks $\CMcal{B}=\big\{\{(x,1),(y,2), (x + y,3), (2x + y,3)\}:x,y\in \mathbb{F}_q\big\}$ with all computations over $\mathbb{F}_q$. Both MOLS and the arising transversal design are visualized in Fig.~\ref{fig:MOLS}.
\end{example}

\subsection{Properties of $\mathscr{L}^m_q$-TD LDPC Codes}

Define $\mathscr{L}^m_q=\big\{
\{ \CMcal{L}^{(\alpha_1,\beta_1)}_q, \hdots, \CMcal{L}^{(\alpha_m,\beta_m)}_q \}:
\alpha_i,\beta_i \in \mathbb{F}_q^*,\allowbreak \alpha_i\beta_j \neq \alpha_j\beta_i \text{ for } 1\leq i,j \leq m \text{ and } i\neq j
\big\}$
as the family of all $m$-sets of cyclic MOLS of order $q$. Each transversal design that is based on a set of $\mathscr{L}_q^m$ is referred to as an \emph{$\mathscr{L}^m_q$-TD}. By taking the incidence matrix of such a transversal design as the parity-check matrix of a code, we obtain an \emph{$\mathscr{L}^m_q$-TD LDPC code}. Note that the order of the Latin squares within a set of MOLS is irrelevant, since every order leads to the same code.

An $\mathscr{L}^m_q$-TD LDPC code has block length $N=q^2$, rate $R\geq (q-m-2)/q$, and the code's parity-check matrix is regular with column weight $m+2$ and row weight $q$. 
We show in \cite{GrunHub2013} that these codes have quasi-cyclic structure, leading to a low encoding complexity linear with the block length. The parity-check matrix of a  quasi-cyclic $\mathscr{L}^m_q$-TD LDPC code consists of $(m+2)\times q$ circulant submatrices (called circulants) of size $q\times q$. For a more flexible code design, it is possible to use any grid of circulants to generate a wide spectrum of block lengths and code rates for any prime power order $q$ \cite{GrunHub2013}.

\subsection{Simplification: Reduced Form of MOLS from $\mathscr{L}^m_q$}

Let $\CMcal{M}:=\big\{\CMcal{L}^{(\alpha_1,\beta_1)}_q,\hdots,\CMcal{L}^{(\alpha_m,\beta_m)}_q\big\}\in \mathscr{L}^m_q$ be any set of cyclic MOLS, $\CMcal{D}_{\CMcal{M}}$ be the corresponding $\mathscr{L}^m_q$-TD and $\CMcal{C}(\CMcal{D}_{\CMcal{M}})$ be the $\mathscr{L}^m_q$-TD LDPC code based on $\CMcal{D}_{\CMcal{M}}$.

\begin{lemma} \label{lemma:replacing_latin_squares}
Let $\ell_1,\hdots,\ell_m\in \mathbb{F}_q^*$ and $\CMcal{M}':=\big\{\CMcal{L}^{(\alpha'_1, \beta'_1)}_q,\hdots,\allowbreak \CMcal{L}^{(\alpha'_m,\beta'_m)}_q\big\}$ with $(\alpha_i',\beta_i'):=(\ell_i\alpha_i, \ell_i\beta_i)$ for $1\leq i\leq m$. Then, it follows that $\CMcal{C}(\CMcal{D}_{\CMcal{M}'})=\CMcal{C}(\CMcal{D}_{\CMcal{M}})$.
\end{lemma}
\begin{proof}
Observe that $\CMcal{L}^{(\alpha'_i,\beta'_i)}_q[x,y]=\ell_i \CMcal{L}^{(\alpha_i,\beta_i)}_q[x,y]$ over $\mathbb{F}_q$, i.e., there is a bijection between the symbols of $\CMcal{L}^{(\alpha_i,\beta_i)}_q$ and $\CMcal{L}^{(\alpha'_1,\beta'_1)}_q$ that preserves the structure of the Latin squares (renaming of the symbols). 
This implies a renaming of the points of $\CMcal{D}_{\CMcal{M}}$ and, equivalently, a reordering of the rows of the parity-check matrix of 
$\CMcal{C}(\CMcal{D}_{\CMcal{M}})$. It is obvious that any reordering of the rows of a code's parity-check matrix does not change the code.
\end{proof}

\begin{theorem}\label{reduced_form}
The set of cyclic MOLS $\CMcal{M}$ can be reduced to the form $\CMcal{M}':=\big\{\CMcal{L}^{(\alpha'_1,1)}_q,\hdots,\CMcal{L}^{(\alpha'_m,1)}_q\big\}$ with $\alpha'_i:= \alpha_i\beta_i^{-1}$ such that $\CMcal{C}(\CMcal{D}_{\CMcal{M}'})=\CMcal{C}(\CMcal{D}_{\CMcal{M}})$. Then, $\CMcal{M}'$ is called the \emph{reduced form} of $\CMcal{M}$.
\end{theorem}
\begin{proof}
It holds that $(\alpha_i,\beta_i)=\beta_i (\alpha'_i,1)$ and thus $\CMcal{L}^{(\alpha_i,\beta_i)}_q$ can be replaced by $\CMcal{L}^{(\alpha_i\beta_i^{-1},1)}_q$ according to Lemma~\ref{lemma:replacing_latin_squares} without changing the code.
\end{proof}

For the rest of the paper, we consider only MOLS in reduced form as a consequence of Theorem \ref{reduced_form}. This simplification is reasonable in order to investigate the structural properties of the arising codes. However, it is possible to construct a quasi-cyclic LDPC code by the appropriate choice of another representative system as described in \cite{GrunHub2013}, which allows encoding with low complexity.

\subsection{Equivalence of $\mathscr{L}^m_q$-TD LDPC Codes}

\begin{theorem}\label{theorem:equivalence_of_TD_LDPC_codes}
Let $\CMcal{M}:=\big\{\CMcal{L}^{(\alpha_1)}_q,\hdots,\CMcal{L}^{(\alpha_m)}_q\big\}\in \mathscr{L}^m_q$ be a set of MOLS in reduced form. 
For $\CMcal{M}_\ell:=\big\{\CMcal{L}^{(\ell \alpha_1)}_q,\hdots,\CMcal{L}^{(\ell \alpha_m)}_q\big\}$ with $\ell\in \mathbb{F}_q^*$ it follows that $\CMcal{C}(\CMcal{D}_{\CMcal{M}_\ell})=\CMcal{C}(\CMcal{D}_{\CMcal{M}})$.
\end{theorem}
\begin{proof}
It can easily be seen that the Latin square $\CMcal{L}^{(\ell \alpha_i)}_q$ can be obtained by reordering certain rows of $\CMcal{L}^{(\alpha_i)}_q$ and that this reordering is the same for each $1\leq i\leq m$. 
The simultaneous reordering of the rows of the given MOLS is equivalent to reorder the corresponding rows of the parity-check matrix of $\CMcal{C}(\CMcal{D}_{\CMcal{M}})$. Obviously, any reordering of the rows of a code's parity-check matrix does not change the code.
\end{proof}

\begin{example}
Let $\CMcal{M}=\big\{ \CMcal{L}^{(1)}_5,\CMcal{L}^{(2)}_5 \big\}$ be a set of two cyclic MOLS of order $5$. By Theorem~\ref{theorem:equivalence_of_TD_LDPC_codes}, the sets $\CMcal{M}_2=\big\{ \CMcal{L}^{(2)}_5,\CMcal{L}^{(4)}_5 \big\}$ ($\ell=2$), $\CMcal{M}_3=\big\{ \CMcal{L}^{(3)}_5,\CMcal{L}^{(1)}_5 \big\}$ ($\ell=3$) and $\CMcal{M}_4=\big\{ \CMcal{L}^{(4)}_5,\CMcal{L}^{(3)}_5 \big\}$ ($\ell=4$) can be constructed such that $\CMcal{C}(\CMcal{D}_{\CMcal{M}})=\CMcal{C}(\CMcal{D}_{\CMcal{M}_\ell})$.
\end{example}

\section{Eliminating Absorbing Sets \\ in $\mathscr{L}_q^m$-TD LDPC Codes}
\label{section:elimination}

The main target of the present paper is to eliminate harmful absorbing sets in $\mathscr{L}_q^m$-TD LDPC codes in order to improve the error-floor performance of these codes. 
Due to the intricate nature of iterative decoders over the AWGN channel, the current understanding about the exact failure mechanism is far from complete. However, the harmfulness of absorbing sets can be based on several conjectures that are in accordance with simulative results and with the current level of knowlege. An $(a,b)$ absorbing set is supposed to be harmful if
\begin{itemize}[leftmargin=15pt]
\item the size $a$ is small,
\item the syndrome $b$ is small compared to $a$,
\item the absorbing set is fully, and
\item the degrees of the variable nodes are small, in particular, the absorbing set is elementary.
\end{itemize}

The absorbing set candidates obtained by the classification in Section~\ref{section:classification} represent potentially harmful absorbing sets since they satisfy at least one of the conjectures listed above. 
We now demonstrate that some of these candidates can be avoided in $\mathscr{L}_q^m$-TD LDPC codes with well-chosen code parameters, more precisely, by a proper choice of the scale factors $\alpha_1,\hdots,\alpha_m$ of the underlying MOLS and of the prime power order $q$ of the Galois field $\mathbb{F}_q$.

For the rest of the section, let $\CMcal{M}:=\big\{\CMcal{L}_q^{(\alpha_1)},\hdots,\CMcal{L}_q^{(\alpha_m)}\big\}\in \mathscr{L}_q^m$ be a set of $m$ cyclic MOLS in reduced form, $\CMcal{D}_{\CMcal{M}}$ be the corresponding $\mathscr{L}_q^m$-TD of block size $k=m+2$ and $\CMcal{C}(\CMcal{D}_{\CMcal{M}})$ be the corresponding LDPC code of column weight $k$. 
We denote the groups of $\CMcal{D}_{\CMcal{M}}$ by $\CMcal{G}=\{G_1,\hdots,G_k\}$ such that the points of $G_1$ and $G_2$ correspond to the coordinates of the common row and column sets of $\CMcal{M}$, respectively, and the points of $G_{i+2}$ correspond to the symbols of $\CMcal{L}_q^{(\alpha_i)}$ for $1\leq i\leq m$.

\subsection{Describing Absorbing Sets by Linear Equation Systems}

Let $(\mathcal{P},\mathcal{B})$ be an absorbing set candidate from the classification of Section~\ref{section:classification}
and let $\varphi_1,\hdots,\varphi_\tau$ be the non-isomorphic $k$-colourings of $(\mathcal{P},\mathcal{B})$.
For each $k$-colouring $\varphi_j$, there exist $(k!)$ possibilities to associate the $k$ colours of $Q=\{1,\hdots,k\}$ with the $k$ groups of $\CMcal{G}=\{G_1,\hdots,G_k\}$.
More precisely, let $\pi_1,\hdots,\pi_{k!}$ be all possible bijections $Q\rightarrow\CMcal{G}$, called the \emph{colour-to-group mappings}.
As a short notation, we use $\hat{\pi}_\ell = (\pi_\ell^{-1}(G_1),\allowbreak\pi_\ell^{-1}(G_2),\hdots,\allowbreak\pi_\ell^{-1}(G_k))$. In Table~\ref{cps_k4}, we explicitely list the short notations of all $24$ colour-to-group mappings which will be important for the presentation of our results. For instance, from $\hat{\pi}_5=(1,4,3,2)$ we obtain $\pi_5(1)=G_1$, $\pi_5(4)=G_2$, $\pi_5(3)=G_3$ and $\pi_5(2)=G_4$.

\begin{theorem}\label{theorem:representing_an_absorbing_set}
An absorbing set of type $(\mathcal{P},\mathcal{B},\varphi_j)$ occurs in $\CMcal{C}(\CMcal{D}_{\CMcal{M}})$ if and only if there is an assignment of values from $\mathbb{F}_q$ to the variables $p_1,\hdots,p_\upsilon$ with $\upsilon:=|\mathcal{P}|$ and a colour-to-group mapping $\pi_\ell$ such that
\begin{enumerate}[label=(\arabic*),leftmargin=20pt]
\item the point $x\in\mathcal{P}$ is associated with variable $p_x$,
\item all elements of $\big\{\big(p_x,\varphi_j(x)\big):x\in\mathcal{P}\big\}$ are unique, 
\item for every block $\{x_1,\hdots,x_k\} \in \mathcal{B}$ with
$(\pi_\ell ( \varphi_j(x_1))=G_1,\hdots, \pi_\ell (\varphi_j(x_k))=G_k$
the linear equations $\alpha_i p_{x_1} + p_{x_2} - p_{x_{i+2}}=0$ are satisfied over $\mathbb{F}_q$ for $i=1,\hdots,m$.
\end{enumerate}
\end{theorem}
\begin{proof}
If there are values $p_1,\hdots,p_\upsilon \in \mathbb{F}_q$ that satisfy (1)-(3), then the elements of $\big\{\big(p_x,\varphi(x)\big):x\in\mathcal{P}\big\}$ can be considered as points of the transversal design $\CMcal{D}_\CMcal{M}$ and the elements of $\big\{\{(p_{x_1},1),\hdots,(p_{x_k},k)\}:\{x_1,\hdots,x_k\}\in\mathcal{B}\big\}$ as blocks of $\CMcal{D}_\CMcal{M}$.
\end{proof}

The linear equations obtained by Theorem~\ref{theorem:representing_an_absorbing_set} lead to a homogeneous linear
system over $\mathbb{F}_q$ with unknown variables $p_1,\hdots,p_\upsilon$ and coefficients depending on the scale factors $\alpha_i$. Let $\boldsymbol{E}$ be the coefficient matrix and $\boldsymbol{p}=(p_1,\hdots,p_\upsilon)^T$ be the column vector of the unknown variables.
Every solution of $\boldsymbol{Ep}=\boldsymbol{0}$ corresponds to an absorbing set of type $(\mathcal{P},\mathcal{B},\varphi_j)$ with colour-to-group mapping $\pi_\ell$ if
the pairs $\big\{\big(p_x,\varphi(x)\big):x\in\CMcal{P}\big\}$ are unique. Conversely, if such a solution does not exist, there can not be any absorbing set of this type.

\subsection{Elimination Process}

\emph{Approach}: 
The existence of any absorbing sets of type $(\mathcal{P},\mathcal{B},\varphi_j)$ with colour-to-group mapping $\pi_\ell$ in a code $\CMcal{C}(\CMcal{D}_{\CMcal{M}})$ is equivalent to the existence of certain solutions of the linear equation system $\boldsymbol{Ep}=\boldsymbol{0}$. 
We solve this system symbolically by a modified Gaussian elimination algorithm  in dependence of the scale factors $\alpha_i$ and the Galois field $\mathbb{F}_q$.

\noindent\rule[2pt]{254pt}{0.5pt}\\
\noindent $(\mathfrak{C}_1, \mathfrak{C}_2,\hdots, \mathfrak{C}_\tau)$ = \textbf{EliminationProcess}$(\mathcal{P},\mathcal{B})$\\
\noindent\rule[4pt]{254pt}{0.5pt}

\noindent\textsc{Input}: 

\begin{itemize}[leftmargin=18pt,labelindent=5pt]
\item $(\mathcal{P},\mathcal{B})$: absorbing set candidate of block size $k$
\end{itemize}

\noindent\textsc{Output}: 

\begin{itemize}[leftmargin=18pt,labelindent=5pt]
\item $\mathfrak{C}_j$: The process outputs the elimination constraints $\mathfrak{C}_j$ grouped by their $k$-colourings $\varphi_j$, i.e., all absorbing sets of type $(\mathcal{P},\mathcal{B},\varphi_j)$ can be eliminated if and only if $\mathfrak{C}_j=1$. Conversely, there exist absorbing sets of type $(\mathcal{P},\mathcal{B},\varphi_j)$ if and only if $\mathfrak{C}_j=0$.
\end{itemize}

\noindent\textsc{Algorithm}: 

\begin{enumerate}[label=(\arabic*),leftmargin=18pt]

\item For every $k$-colouring $\varphi_j$ of $(\mathcal{P},\mathcal{B})$, $1\leq j \leq \tau$, and every colour-to-group mapping $\pi_\ell$, $1\leq \ell \leq (k!)$, do:
\[
\mathfrak{C}_{j,\ell} \text{ =  \textbf{FindEliminationConstraint}} (\mathcal{P},\mathcal{B},\varphi_j,\pi_\ell).
\]

\item Output $\mathfrak{C}_j:=\bigwedge_{\ell=1}^{k!} \mathfrak{C}_{j,\ell}$ for $1\leq j \leq \tau$.

\end{enumerate}

\noindent\rule[2pt]{254pt}{0.5pt}\\
\noindent $\mathfrak{C}_{j,\ell}$ = \textbf{FindEliminationConstraint}$(\mathcal{P},\mathcal{B},\varphi_j,\pi_\ell)$\\
\noindent\rule[4pt]{254pt}{0.5pt}

\noindent\textsc{Input}: 

\begin{itemize}[leftmargin=18pt,labelindent=5pt]
\item $(\mathcal{P},\mathcal{B})$: absorbing set candidate of block size $k$
\item $\varphi_j$: the $j$-th $k$-colouring of $(\mathcal{P},\mathcal{B})$
\item $\pi_\ell$: the $\ell$-th colour-to-group mapping
\end{itemize}

\noindent\textsc{Output}: 
\begin{itemize}[leftmargin=18pt]
\item $\mathfrak{C}_{j,\ell}$: The output $\mathfrak{C}_{j,\ell}$ is an elimination constraint such that all absorbing sets of type $(\mathcal{P},\mathcal{B},\varphi_j)$ with colour-to-group mapping $\pi_\ell$ can be eliminated if and only if $\mathfrak{C}_{j,\ell}=1$. 
\end{itemize}

\noindent\textsc{Notations and Invariants}: 
\begin{itemize}[leftmargin=18pt]
\item We denote the type $(\mathcal{P},\mathcal{B},\varphi_j)$ with colour-to-group mapping $\pi_\ell$ from now on as $\mathscr{T}$.
\item The set $\Lambda^0$ contains all polynomials that are definitely zero, initialized by $\Lambda^0=\{0\}$.
\item The set $\Lambda^+$ contains all polynomials that are definitely non-zero. Initially, we have\footnote{Note that the scale factors $\alpha_i\in \mathbb{F}_q^*$ of any MOLS are non-zero and unique by definition such that $\alpha_i$ and any difference $\alpha_i-\alpha_j$ with $i\neq j$ are definitely non-zero. Furthermore, an entry that is equal to a prime power $p^i$ with $i\geq 1$ can be zero if the underlying Galois field has characteristic $p$. Conversely, an entry that is not a prime power must be definitely non-zero.  As an example, the term $6\alpha_1^2 \alpha_2^3(\alpha_1-\alpha_2)^2$ is definitely non-zero.}

\begin{itemize}[leftmargin=15pt]
\item $1,-1\in \Lambda^+$,
\item $\alpha_i\in \Lambda^+$ for $1\leq i\leq m$,
\item $(\alpha_i-\alpha_j)\in \Lambda^+$ for $1\leq i,j\leq m$ and $i\neq j$,
\item $u\in \Lambda^+$ if $u$ is not a prime power, and
\item $\lambda_1\lambda_2 \in \Lambda^+$ if $\lambda_1,\lambda_2\in\Lambda^+$.
\end{itemize}

\item The set $\Lambda^\text{\normalsize\textasteriskcentered}$ contains all polynomials that may be zero or non-zero depending on the choice of the scale factors $\alpha_i$ and $\mathbb{F}_q$. Consequently, it consists of all polynomials that are not in $\Lambda^0 \cup \Lambda^+$.

\item 
Let $\mathcal{T}$ be a rooted tree that consists of \emph{case nodes} of the form $(\lambda,v)^\text{C}$ and \emph{elimination nodes} of the form $(\lambda,v)^\text{E}$, where $\lambda\in\Lambda^\text{\normalsize\textasteriskcentered}$, $v\in \{0,1\}$  and where `C' and `E' stands for \emph{case} and \emph{elimination}, respectively. The elimination nodes must be leafs.
We say that a node is \emph{satisfied} if $\varrho(\lambda,v)=1$, where
\[
\varrho(\lambda,v)=
\begin{cases}
1, & \mbox{if } (\lambda\neq 0) \mbox{ and } (v=1)\\
   & \mbox{or } (\lambda = 0) \mbox{ and } (v=0),\\
0, & \mbox{otherwise.}
\end{cases}
\]
Note that it depends on the choice of the scale factors $\alpha_i$, $1\leq i\leq m$ and the choice of $\mathbb{F}_q$ if a node is satisfied. 
\end{itemize}

\noindent\textsc{Algorithm}: 

\begin{enumerate}[label=(\arabic*),leftmargin=18pt]

\item 
Build up an equation system $\boldsymbol{Ep}=\boldsymbol{0}$ with the equations obtained by Theorem~\ref{theorem:representing_an_absorbing_set} for type $\mathscr{T}$.
Let $\boldsymbol{E}_\eta$ be the upper right submatrix of the coefficient matrix $\boldsymbol{E}$ including the $\eta$-th row and column. We start with $\eta=1$.
Initialize the ouput tree $\mathcal{T}$ with an empty root node and and assume that this root node is trivially satisfied. 

\item
Find a column of $\boldsymbol{E}_\eta$ with all entries from $\Lambda^0\cup\Lambda^+$ and with at least one entry from $\Lambda^+$. 
If such a column does not exist, continue with step (3). Otherwise swap the rows and columns\footnote{If we swap rows of $\boldsymbol{E}$, we must also swap the corresponding entries in $\boldsymbol{p}$.} of $\boldsymbol{E}$ in such a way that there is an entry of $\Lambda^+$ in the left top corner of $\boldsymbol{E}_\eta$. Then, apply row eliminations to $\boldsymbol{E}$ such that all other entries of the first column of $\boldsymbol{E}_\eta$ become zero (except the first) and continue with (4). 

\item Find the column of $\boldsymbol{E}_\eta$ with the smallest positive number of entries in $\Lambda^\text{\normalsize\textasteriskcentered}$. If such a column does not exist, i.e., if all entries are from $\Lambda^0$, continue with step (6).
Otherwise, choose an entry $\lambda\in\Lambda^\text{\normalsize\textasteriskcentered}$ of this column and do a case differentation. For this, split the current process into two subprocesses. For the first subprocess,
\begin{itemize}[leftmargin=15pt]
\item assume that $\lambda=0$,
\item append the node $(\lambda,0)^\text{C}$ to the current node of $\mathcal{T}$,\footnote{By the current node of $\mathcal{T}$ we mean the case node that has been inserted by the parent process. Initially, the current node is the root of $\mathcal{T}$.}
\item move $\lambda$ to $\Lambda^0$, and
\item continue with step (5).
\end{itemize}
For the second subprocess,
\begin{itemize}[leftmargin=15pt]
\item assume that $\lambda\neq 0$,
\item append the node $(\lambda,1)^\text{C}$ to the current node of $\mathcal{T}$,
\item move $\lambda$ to $\Lambda^+$, and
\item continue with step (5).
\end{itemize}

Notice that the sets $\Lambda^0,\Lambda^+$ and $\Lambda^\text{\normalsize\textasteriskcentered}$ are process invariants that are stored for each subprocess separately, whereas $\mathcal{T}$ is stored as a global entity.

\item

After every step search for linear equations of the form $\lambda (p_{x_1}-p_{x_2})=0$ with $\varphi({x_1})=\varphi({x_2})$ and $\lambda\in\Lambda^\text{\normalsize\textasteriskcentered}$. 
By choosing the scale factors $\alpha_i$ and $\mathbb{F}_q$ in such a way that $\lambda\neq 0$ over $\mathbb{F}_q$, it follows that $(p_{x_1},\varphi(x_1))=(p_{x_2},\varphi(x_2))$ which violates the second condition of Theorem~\ref{theorem:representing_an_absorbing_set} for~all possible solutions. Hence, we can eliminate all absorbing sets of type $\mathscr{T}$.
Consequently, we append the elimination node $(\lambda,1)^\text{E}$ to the current node of $\mathcal{T}$.
For further processing, we start a new process and
\begin{itemize}[leftmargin=15pt]
\item assume that $\lambda=0$,
\item remove the row of $\boldsymbol{E}$ that corresponds to the equation (which is trivially satisfied),
\item move $\lambda$ to $\Lambda^0$, and
\item continue with step (5).
\end{itemize}

\item
Repeat step (2) with $\eta:=\eta+1$ until the matrix is in row echelon form. If the matrix is in row echelon form, we solve the system symbolically by back substitution such that all unknown variables depend on the scale factors $\alpha_i$ and on some free variables if the system is underdetermined. We obtain symbolic expressions for $p_1,\hdots,p_\upsilon$. 

\item Compute the symbolic differences $p_{x_1}-p_{x_2}$ for all $x_1,\allowbreak x_2\in \mathcal{P}$ with $\varphi(x_1)=\varphi(x_2)$ and evaluate under which conditions these differences are zero. 
More precisely, search for differences of the form $p_{x_1}-p_{x_2}=\lambda(p_{x_3}-p_{x_4})=0$ with $\lambda\in\Lambda^\text{\normalsize\textasteriskcentered}$ and $x_3,x_4\in\mathcal{P}\setminus\{x_1,x_2\}$. For the case of $\lambda=0$, the points $(p_{x_1},\varphi(x_1))$ and $(p_{x_2},\varphi(x_2))$ coincide for all possible solutions and thus, the obtained set system can not be of type $\mathscr{T}$. Consequently, append the elimination node $(\lambda,0)^\text{E}$ to the current node of $\mathcal{T}$.

\item 

Let $\mathfrak{p}=\big(\text{root}, (\lambda_1,v_1)^\text{C}, (\lambda_2,v_2)^\text{C}, \hdots, (\lambda_{n-1},v_{n-1})^\text{C},\allowbreak (\lambda_n,v_n)^\text{E}\big)$ be a path of $\mathcal{T}$ that ends up in an elimination node. Define $\varrho(\mathfrak{p}):=\bigwedge_{i=1}^n \varrho(\lambda_i,v_i)$.
Then, the absorbing sets of type $\mathscr{T}$ can be eliminated if 
$\varrho(\mathfrak{p}) = 1$. 
Let $\mathfrak{p}_1,\hdots,\mathfrak{p}_\xi$ be all paths of $\mathcal{T}$ that ends up in an elimination node. Define $\mathfrak{C}_{j,\ell}:=\bigvee_{i=1}^\xi \varrho(\mathfrak{p}_i)$.
Then, the absorbing sets of type $\mathscr{T}$ can be eliminated if and only if $\mathfrak{C}_{j,\ell} = 1$. 
Finally, we simplify the elimination constraint $\mathfrak{C}_{j,\ell}$ by using the following rules:
\begin{itemize}[leftmargin=15pt]
\item $\varrho(\varepsilon\lambda,v)=\varrho(\lambda,v)$,
\item $\varrho(\lambda^i,v)=\varrho(\lambda,v)$,
\item if $\varrho(\lambda,v)\Rightarrow\varrho(\lambda',v')$, $\varrho(\lambda,v)\vee \varrho(\lambda',v')=\varrho(\lambda',v')$,
\item if $\varrho(\lambda,v)\Rightarrow\varrho(\lambda',v')$, $\varrho(\lambda,v)\wedge \varrho(\lambda',v')=\varrho(\lambda,v)$,
\item $\varrho(\lambda,0)\vee \varrho(\lambda',0)=\varrho(\lambda \lambda',0)$
\item $\varrho(\lambda\lambda',1)\vee \varrho(\lambda,1)=\varrho(\lambda,1)$
\item $\varrho(\lambda,1) \vee \varrho(\lambda',0)=\varrho(\lambda,1)\vee\varrho(\varepsilon\lambda+\varepsilon'\lambda',0)$,
\item $\varrho(\lambda,1) \vee \varrho(\lambda',1)=\varrho(\lambda,1)\vee \varrho(\varepsilon\lambda+\varepsilon'\lambda',1)$,
\item $\varrho(\lambda,0) \wedge \varrho(\lambda',0)=\varrho(\lambda,0)\wedge \varrho(\varepsilon\lambda+\varepsilon'\lambda',0)$,
\item $\varrho(\lambda,0) \wedge \varrho(\lambda',1)=\varrho(\lambda,0)\wedge \varrho(\varepsilon\lambda+\varepsilon'\lambda',1)$,
\end{itemize}
where $\lambda,\lambda'\in \Lambda^\text{\normalsize\textasteriskcentered}$, $\varepsilon,\varepsilon'\in \Lambda^+$, and $v,v'\in\{0,1\}$.
\end{enumerate}

\begin{example}\label{example:absorbing_pattern_as_subsystem}
Let $(\mathcal{P},\mathcal{B})$ be the $(4,4)$ absorbing set candidate given by the points $\mathcal{P}=\{1,\hdots,10\}$ and blocks $\mathcal{B}=\big\{\{1,2,3,4\},\allowbreak\{1,5,6,7\},\allowbreak\{2,5,8,9\},\allowbreak\{3,6,8,10\}\big\}$. 
With $4$-colouring $\hat{\varphi}_1=(\{1,8\},\allowbreak\{2,6\},\allowbreak\{3,7,9\},\{4,5,10\})$ and colour-to-group mapping $\hat{\pi}_5=(1,4,2,3)$ from Table~\ref{cps_k4}, we obtain the following system of linear equations over $\mathbb{F}_q$:
\begin{equation*}
\begin{alignedat}{8}
\alpha_1 p_1 &+ p_4 -& p_2 & = 0, & \ \ \ \ \ \ \alpha_2 p_1 &+ p_4 &- p_3 & = 0,\\
\alpha_1 p_1 &+ p_5 -& p_6 & = 0, &              \alpha_2 p_1 &+ p_5 &- p_7 & = 0,\\
\alpha_1 p_8 &+ p_5 -& p_2 & = 0, &              \alpha_2 p_8 &+ p_5 &- p_9 & = 0,\\
\alpha_1 p_8 &+ p_{10} -& p_6 & = 0, &              \alpha_2 p_8 &+ p_{10} &- p_3 & = 0.\\
\end{alignedat}
\end{equation*}
This system can also be described by the matrix equation
\[
\boldsymbol{Ep}=
\left[
\begin{smallmatrix}
 \alpha_1 &  . & 1 & . & . & -1 &  . &  . &  . &  .\\
 \alpha_2 &  . & 1 & . & . &  . &  . & -1 &  . &  .\\
 \alpha_1 &  . & . & 1 & . &  . & -1 &  . &  . &  .\\
 \alpha_2 &  . & . & 1 & . &  . &  . &  . & -1 &  .\\
  . & \alpha_1 & . & 1 & . & -1 &  . &  . &  . &  .\\
  . & \alpha_2 & . & 1 & . &  . &  . &  . &  . & -1\\
  . & \alpha_1 & . & . & 1 &  . & -1 &  . &  . &  .\\
  . & \alpha_2 & . & . & 1 &  . &  . & -1 &  . &  .\\
\end{smallmatrix}
\right]
\left[
\begin{smallmatrix}
p_1\\ p_8\\ p_4\\ p_5\\ p_{10}\\ p_2\\ p_6\\ p_3\\ p_7\\ p_9
\end{smallmatrix}
\right]
\footnotesize
=\boldsymbol{0},
\]
where the dots represent zeros. By applying step (2)-(5) of the elimination process, we obtain the coefficient matrix $\boldsymbol{E}$ in row echelon form. The matrix equation is as follows:
\begin{align*}
\left[
\begin{smallmatrix}
\alpha_1&.&1&.&.&-1&.&.&.&.\\
.&\alpha_1&.&1&.&-1&.&.&.&.\\
.&.&-1&1&.&1&.&-1&.&.\\
.&.&.&\alpha_2-\alpha_1&.&.&.&-\alpha_2&\alpha_1&.\\
.&.&.&.&\alpha_1-\alpha_2&\alpha_1-\alpha_2&.&2\alpha_2-\alpha_1&-\alpha_1&.\\
.&.&.&.&.&\alpha_2&.&-\alpha_2&\alpha_1&-\alpha_1\\
.&.&.&.&.&.&-\alpha_2&.&\alpha_2-\alpha_1&\alpha_1\\
.&.&.&.&.&.&.&.&\alpha_2-2\alpha_1&2\alpha_1-\alpha_2\\
\end{smallmatrix}
\right]
\\
\left[
\begin{smallmatrix}
p_1,\ p_8,\ p_4,\ p_5,\ p_{10},\ p_2,\ p_6,\ p_7,\ p_3,\ p_9
\end{smallmatrix}
\right]^T
\footnotesize
=\boldsymbol{0}.
\end{align*}

Step (4) of the elimination process detects the linear equation of the form $(2\alpha_1-\alpha_2)(p_9-p_3)=0$ and hence, we have $2\alpha_1-\alpha_2=0$ or $p_9-p_3=0$. By a proper choice of $\alpha_1$, $\alpha_2$ and $\mathbb{F}_q$, we may ensure that $2\alpha_1-\alpha_2\neq 0$ over $\mathbb{F}_q$ and hence, it must hold that $p_9=p_3$.
Their equality leads to a degraded set system where the two points $3$ and $9$ of $\mathcal{P}$ coincide.
Therefore, we append the elimination node $(2\alpha_1-\alpha_2,1)^{\text{E}}$ to the root of $\mathcal{T}$. 
Step (6) does not find further elimination constraints. Finally, we obtain by step (7) that all absorbing sets of type $(\mathcal{P},\mathcal{B},\varphi_1)$ with colour-to-group mapping $\pi_5$ can be avoided if and only if the elimination constraint $\mathfrak{C}_{1,5}:=\varrho(2\alpha_1-\alpha_2,1)$ is satisfied.

\end{example}

\begin{table}[t!]
\caption{Non-isomorphic 3-colourings of the absorbing set candidates of Fig.~\ref{fig:absorbing_set_classification_k3} and their elimination constraints over $\mathbb{F}_q$. An asterisk indicates that the 3-colouring leads to a fully absorbing set.}
\label{elimination_constraints_k3}
\setlength{\tabcolsep}{3pt}
\renewcommand{\arraystretch}{1.3}
\centering
	\begin{tabular}{ccc}
	\toprule
	$\boldsymbol{(a,b)}$ & \textbf{3-colourings} & \textbf{elimination}\\
	\midrule
	$(3,3)$ & $(\{1,6\},\{2,5\},\{3,4\})$ & no\\
	$(4,0)$ & $(\{1,6\},\{2,5\},\{3,4\})$ ({\large\textasteriskcentered}) & $\omega_q\neq 2$\\
	$(4,2)$ & $(\{1,6,7\},\{2,5\},\{3,4\})$ ({\large\textasteriskcentered}) & $\omega_q = 2$\\
	$(4,4)$ & $(\{1,6\},\{2,5,8\},\{3,4,7\})$ & no\\
	  & $(\{1,6\},\{2,4\},\{3,5,7,8\})$ ({\large\textasteriskcentered}) & no\\
	$(5,3)\{1\}$ & $(\{1,6,7\},\{2,5,8\},\{3,4,9\})$ & no\\
	$(5,3)\{2\}$ & $(\{1,6,9\},\{2,5,8\},\{3,4,7\})$ & no\\
	$(5,5)$ & $(\{1,7,9,10\},\{2,5,8\},\{3,4,6\})$ & no\\
	$(6,0)\{1\}$ & $(\{1,6,7\},\{2,5,8\},\{3,4,9\})$ ({\large\textasteriskcentered}) & $\omega_q\neq 2$\\
	$(6,0)\{2\}$ & $(\{1,6,9\},\{2,5,8\},\{3,4,7\})$ ({\large\textasteriskcentered}) & no\\
	$(6,2)\{1\}$ & $(\{1,8\},\{2,5,6\},\{3,4,7\})$ ({\large\textasteriskcentered}) & $\omega_q\neq 3$\\
	$(6,2)\{2\}$ & $(\{1,8,9\},\{2,5,6\},\{3,4,7\})$ ({\large\textasteriskcentered}) &  $\omega_q= 2$ or $3$\\
	$(6,2)\{3\}$ & $(\{1,6,7\},\{2,5,9,10\},\{3,4,8\})$ ({\large\textasteriskcentered}) & $\omega_q = 2$\\
	$(6,2)\{4\}$ & $(\{1,6,7\},\{2,5,8\},\{3,4,9,10\})$ ({\large\textasteriskcentered}) & $\omega_q = 2$\\
	$(6,2)\{5\}$ & $(\{1,6,7,10\},\{2,5,8\},\{3,4,9\})$ ({\large\textasteriskcentered}) & $\omega_q = 2$\\
	$(6,2)\{6\}$ & $(\{1,6,9,10\},\{2,5,8\},\{3,4,7\})$ ({\large\textasteriskcentered}) & always\\
	$(6,4)\{1\}$ & $(\{1,10\},\{2,5,7,8\},\{3,4,6,9\})$ & no\\
	$(6,4)\{2\}$ & $(\{1,8,9,10\},\{2,5,6\},\{3,4,7\})$ ({\large\textasteriskcentered}) & no\\
	$(6,4)\{3\}$ & $(\{1,6,8,10,11\},\{2,5,7\},\{3,4,9\})$ ({\large\textasteriskcentered}) & no\\
	  & $(\{1,6,8,9\},\{2,5,7\},\{3,4,10,11\})$ & no\\
	$(6,4)\{4\}$ & $(\{1,6,8,10\},\{2,5,7,11\},\{3,4,9\})$ & no\\
	$(6,4)\{5\}$ & $(\{1,6,8\},\{2,5,9,10\},\{3,4,7,11\})$ & no\\
	  & $(\{1,7,9,10,11\},\{2,4,8\},\{3,5,6\})$ ({\large\textasteriskcentered}) & no\\
	  & $(\{1,7,8,10\},\{2,4,9,11\},\{3,5,6\})$ & no\\
	  & $(\{1,6,8\},\{2,4,9,11\},\{3,5,7,10\})$ & no\\
	$(6,4)\{6\}$ & $(\{1,6,8\},\{2,5,9,10\},\{3,4,7,11\})$ & no\\
	   & $(\{1,6,9,11\},\{2,4,8\},\{3,5,7,10\})$ & no\\
	$(6,6)\{1\}$ & $(\{1,10,11\},\{2,5,7,8\},\{3,4,6,9\})$ & no\\
	$(6,6)\{2\}$ & $(\{1,7,9,10\},\{2,5,8,11\},\{3,4,6,12\})$ & no\\
	  & $(\{1,7,8,11\},\{2,5,9,10\},\{3,4,6,12\})$ & no\\
	  & $(\{1,6,8\},\{2,5,9,11,12\},\{3,4,7,10\})$ & no\\
	  & $(\{1,6,8\},\{2,4,10\},\{3,5,7,9,11,12\})$ ({\large\textasteriskcentered}) & no\\
	\bottomrule
	\end{tabular}
\end{table}

\section{Main Results}
\label{section:main_results}

This section summarizes the main results of our paper and describes how the potentially harmful absorbing sets of the classification of Section~\ref{section:classification} can be eliminated in $\mathscr{L}_q^m$-TD LDPC codes for the cases $k=3$ and $4$. When we refer to the smallest absorbing sets, we mean those with the smallest number of bit nodes and, among these, with the smallest syndrome.

\subsection{Results for $\mathscr{L}_q^1$-TD LDPC Codes of Column Weight $k=3$}

In Table~\ref{elimination_constraints_k3} we give the possible elimination constraints of the smallest absorbing sets that may occur in $\mathscr{L}_q^1$-TD LDPC codes of column weight $k=3$. The given constraints depend on the characteristic $\omega_q$ of the underlying Galois field $\mathbb{F}_q$.

\begin{itemize}[leftmargin=10pt,labelindent=5pt]

\item The smallest $(3,3)$ absorbing sets are unavoidable.

\item The smallest fully absorbing sets have size $(4,0)$ and can be eliminated by choosing $\mathbb{F}_q$ such that $\omega_q\neq 2$.

\item The absorbing sets of size $(4,2)$ can be avoided by choosing $\mathbb{F}_q$ such that $\omega_q = 2$. Hence, the $(4,0)$ and $(4,2)$ absorbing sets can not be eliminated simultaneously.

\item The absorbing sets of size $5$ can not be avoided, but they are supposed to be harmless since they are non-fully and have relatively large syndromes. 

\item The $(6,0)\{1\}$ absorbing sets can be avoided by $\omega_q\neq 2$, whereas the $(6,0)\{2\}$ absorbing sets can not be avoided.

\item The $(6,2)\{1\}$ absorbing sets can be avoided by $\omega_q\neq 3$, the $(6,2)\{2\}$ absorbing sets by $\omega_q=2$ or $3$ and the $(6,2)\{3\}$, $(6,2)\{4\}$ and $(6,2)\{5\}$ absorbing sets by $\omega_q=2$. The $(6,2)\{6\}$ absorbing sets do never occur. 

\item The $(6,4)\{i\}$ and $(6,6)\{i\}$ absorbing sets can not be avoided, but they are supposed to be harmless due to their large syndrome. 

\end{itemize}

We conjecture that we obtain an excellent $\mathscr{L}_q^1$-TD LDPC code by choosing $\mathbb{F}_q$ in such a way that $\omega_q\neq 2$, since the most harmful $(4,0)$ absorbing sets can be eliminated and the absorbing sets of size $(6,0)$ can partially be avoided.

\subsection{Results for $\mathscr{L}_q^2$-TD LDPC Codes of Column Weight $k=4$}

In Fig.~\ref{classification} we present the exact elimination constraints for the smallest absorbing sets that may occur in an $\mathscr{L}_q^2$-TD LDPC code of column weight $k=4$. The given constraints depend on the scale factors $\alpha_1,\alpha_2$ and on the underlying Galois field $\mathbb{F}_q$ with characteristic $\omega_q$ and can be satisfied by a proper choice of these parameters.

\begin{itemize}[leftmargin=10pt,labelindent=5pt]

	\item The smallest possible absorbing sets have size $(4,4)$ and can be eliminated if and only if the constraints C1-C4 of Table~\ref{constraints_k4} are satisfied. The $(6,2)\{1\}$ and $(6,2)\{3\}$ absorbing sets can be simultaneously avoided since they contain a $(4,4)$ absorbing set.

	\item The smallest possible fully absorbing sets have size $(4,4)$ and can be eliminated if and only if C4 is satisfied.
	
	\item The smallest (fully) absorbing sets with syndrome 0 are of size $(6,0)$ and can be avoided if and only if the constraint $\text{C4}\vee\text{C16}$ is satisfied. Note, that $(6,0)$ absorbing sets correspond to codewords of minimum weight 6 and also define stopping sets of size $6$. Hence, by avoiding these entities, we also raise the minimum and stopping distance of the code.
	
	\item The $(5,4)$ absorbing sets can be avoided if and only if the constraints C5-C7 are satisfied. The $(6,0)$ absorbing sets can be simultaneously avoided since they contain a $(5,4)$ absorbing set.

	\item The fully $(6,2)\{2\}$ and $(6,2)\{4\}$ absorbing sets do never occur. 

	\item The fully $(6,4)\{1\}$ absorbing sets can be avoided if and only if the constraint $\text{C8} \vee \overline{\text{C1}}$ is satisfied.

	\item The $(6,4)\{2\}$ absorbing sets can be eliminated if and only if the constraints C4 and C8-C15 are satisfied. 

	\item The $(6,4)\{3\}$ absorbing sets can be avoided if the constraints C16 and C18-C25 are satisfied.

	\item The $(6,4)\{4\}$ absorbing sets can be avoided if and only if the constraint $\overline{\text{C1}}\vee\overline{\text{C2}}\vee\overline{\text{C3}}\vee\overline{\text{C8}}$ is satisfied.

	\item The $(6,6)\{1\}$ absorbing sets can partially be eliminated if the constraints C1-C4 and C8 are satisfied. In particular, the fully $(6,6)\{1\}$ absorbing sets can be avoided if and only if the constraint C4 is satisfied. 

	\item The $(6,6)\{2\}$ absorbing sets can partially be eliminated if the conditions C1-C4 are satisfied, but the fully $(6,6)\{2\}$ absorbing sets can not be avoided.

\end{itemize}

We conjecture that we obtain an excellent $\mathscr{L}_q^2$-TD LDPC code by choosing $\alpha_1$, $\alpha_2$ and $\mathbb{F}_q$ in such a way that the constraints C1-16 and C18-C25 are satisfied. In this case, the most harmful absorbing sets of size  $(4,4)$, $(5,4)$, $(6,0)$ and $(6,2)$ are eliminated and absorbing sets of size $(6,4)$ and $(6,6)$ are partially avoided.

\begin{table}[t!]
\setlength{\tabcolsep}{3pt}
\renewcommand{\arraystretch}{1.3}
\caption{All colour-permutations $\hat{\pi}_\ell$ of $4$ colours (in short notation)}
\centering
	\begin{tabular}{*{6}{c}}
	\toprule
	$\hat{\pi}_1$ & $\hat{\pi}_2$ & $\hat{\pi}_3$ & $\hat{\pi}_4$ & $\hat{\pi}_5$ & $\hat{\pi}_6$\\
	$(1,2,3,4)$ & $(1,2,4,3)$  & $(1,3,2,4)$  & $(1,3,4,2)$  & $(1,4,2,3)$  & $(1,4,3,2)$\\
	\midrule
	$\hat{\pi}_7$ & $\hat{\pi}_8$ & $\hat{\pi}_9$ & $\hat{\pi}_{10}$ & $\hat{\pi}_{11}$ & $\hat{\pi}_{12}$\\
	$(2,1,3,4)$  & $(2,1,4,3)$ & $(2,3,1,4)$ & $(2,3,4,1)$ & $(2,4,1,3)$ & $(2,4,3,1)$\\
	\midrule
	$\hat{\pi}_{13}$ & $\hat{\pi}_{14}$ & $\hat{\pi}_{15}$ & $\hat{\pi}_{16}$ & $\hat{\pi}_{17}$ & $\hat{\pi}_{18}$\\
	$(3,1,2,4)$  & $(3,1,4,2)$ & $(3,2,1,4)$ & $(3,2,4,1)$ & $(3,4,1,2)$ & $(3,4,2,1)$\\
	\midrule
	$\hat{\pi}_{19}$ & $\hat{\pi}_{20}$ & $\hat{\pi}_{21}$ & $\hat{\pi}_{22}$ & $\hat{\pi}_{23}$ & $\hat{\pi}_{24}$\\
	$(4,1,2,3)$  & $(4,1,3,2)$ & $(4,2,1,3)$ & $(4,2,3,1)$ & $(4,3,1,2)$ & $(4,3,2,1)$\\
	\bottomrule
	\end{tabular}
\label{cps_k4}
\end{table}

\begin{table}[t!]
\caption{List of Constraints over $\mathbb{F}_q$}
\label{constraints_k4}
\setlength{\tabcolsep}{3pt}
\renewcommand{\arraystretch}{1.3}
\centering
	\begin{tabular}{cc;{1pt/3pt}cc}
	\hline
	\textbf{Label} & \textbf{Constraint} & \textbf{Label} & \textbf{Constraint}\\
	\hline
	C1 & $\alpha_1 + \alpha_2 \neq 0$ & C15 & $3\alpha_1 - \alpha_2 \neq 0$\\
	C2 & $2\alpha_1 - \alpha_2 \neq 0$ & C16 & $\alpha_1^2 + \alpha_1\alpha_2 + \alpha_2^2 \neq 0$\\
	C3 & $\alpha_1 - 2\alpha_2 \neq 0$ & C17 & $\omega_q \neq 5$\\
	C4 & $\omega_q \neq 2$ & C18 & $3\alpha_1^2 - 3\alpha_1\alpha_2 + \alpha_2^2 \neq 0$\\
	C5 & $\alpha_1^2 + \alpha_1\alpha_2 - \alpha_2^2 \neq 0$ & C19 & $\alpha_1^2 - 3\alpha_1\alpha_2 + 3\alpha_2^2 \neq 0$\\
	C6 & $\alpha_2^2 + \alpha_1\alpha_2 - \alpha_1^2 \neq 0$ & C20 & $3\alpha_1 - 4\alpha_2 \neq 0$\\
	C7 & $\alpha_1^2 - 3\alpha_1\alpha_2 + \alpha_2^2 \neq 0$ & C21 & $4\alpha_1 - 3\alpha_2 \neq 0$\\
	C8 & $\alpha_1^2 - \alpha_1\alpha_2 + \alpha_2^2 \neq 0$ & C22 & $\alpha_1 - 4\alpha_2 \neq 0$\\
	C9 & $\omega_q \neq 3$ & C23 & $4\alpha_1 - \alpha_2 \neq 0$\\
	C10 & $3\alpha_1 - 2\alpha_2 \neq 0$ & C24 & $\alpha_1 + 3\alpha_2 \neq 0$\\
	C11 & $2\alpha_1 - 3\alpha_2 \neq 0$ & C25 & $3\alpha_1 + \alpha_2 \neq 0$\\
	C12 & $\alpha_1 + 2\alpha_2 \neq 0$ & C26 & $\alpha_1^2 + \alpha_2^2 \neq 0$\\
	C13 & $2\alpha_1 + \alpha_2 \neq 0$ & C27 & $2\alpha_1^2 - 2\alpha_1\alpha_2 +\alpha_2^2 \neq 0$\\
	C14 & $\alpha_1 - 3\alpha_2 \neq 0$ & C28 & $\alpha_1^2 - 2\alpha_1\alpha_2 +2\alpha_2^2 \neq 0$\\
	\hline
	\end{tabular}
\end{table}

\begin{figure*}[!t]

\vspace{0.5cm}

\def\scaleboxes {0.65}

\begin{minipage}[t]{0.32\textwidth}
\mbox{}\\[-\baselineskip]

\renewcommand{\arraystretch}{1.3}
\newcolumntype{L}[1]{>{\raggedright\let\newline\\\arraybackslash\hspace{0pt}}m{#1}}
\newcolumntype{C}[1]{>{\centering\let\newline\\\arraybackslash\hspace{0pt}}m{#1}}
\newcolumntype{R}[1]{>{\raggedleft\let\newline\\\arraybackslash\hspace{0pt}}m{#1}}

	\scalebox{\scaleboxes}{
	\begin{tabular}{|C{3.7cm}C{3.9cm}|}
	\multicolumn{2}{c}{$\boldsymbol{(4,4)}$}\\
	\hline
	\multicolumn{2}{|c|}{$\hat{\varphi}_1=(\{1,8\},\{2,6\},\{3,7,9\},\{4,5,10\})$}\\
	\hdashline[1pt/3pt]
	$\ell$ of $\pi_\ell$ & elimination \\
	$1,2,7,8,17,18,23,24$ & $\text{C1}$ \\
	$3,5,9,11,14,16,20,22$ & $\text{C2}$ \\
	$4,6,10,12,13,15,19,21$ & $\text{C3}$ \\
	\hline
	\multicolumn{2}{|c|}{$\hat{\varphi}_2=(\{1,8\},\{2,6\},\{3,5\},\{4,7,9,10\})$ ({\Large\textasteriskcentered})}\\
	\hdashline[1pt/3pt]
	$\ell$ of $\pi_\ell$ & elimination \\
	$1-24$ & $\text{C4}$\\
	\hline
	\end{tabular}
	}
\\
\vspace{0.37cm}
\\
	\scalebox{\scaleboxes}{
	\begin{tabular}{|C{3.7cm}C{3.9cm}|}
	\multicolumn{2}{c}{$\boldsymbol{(5,4)}$}\\
	\hline
	\multicolumn{2}{|c|}{$\hat{\varphi}_1=(\{1,9,10\},\{2,6,12\},\{3,7,8\},\{4,5,11\})$}\\
	\hdashline[1pt/3pt]
	$\ell$ of $\pi_\ell$ & elimination \\
	$1,3,8,11,14,17,22,24$ & $\text{C5}$\\
	$2,4,7,12,13,18,21,23$ & $\text{C6}$\\
	$5,6,9,10,15,16,19,20$ & $\text{C7}$\\
	\hline
	\end{tabular}
	}
\\
\vspace{0.37cm}
\\
	\scalebox{\scaleboxes}{
	\begin{tabular}{|C{3.7cm}C{3.9cm}|}
	\multicolumn{2}{c}{$\boldsymbol{(6,0)}\succ(5,4)$}\\
	\hline
	\multicolumn{2}{|c|}{$\hat{\varphi}_1=(\{1,9,10\},\{2,6,12\},\{3,7,8\},\{4,5,11\})$ ({\Large\textasteriskcentered})}\\
	\hdashline[1pt/3pt]
	$\ell$ of $\pi_\ell$ & elimination \\
	$1-24$ & $\text{C4} \vee \text{C16}$\\
	\hline
	\end{tabular}
	}
\\
\vspace{0.37cm}
\\
	\scalebox{\scaleboxes}{
	\begin{tabular}{|C{3.7cm}C{3.9cm}|}
	\multicolumn{2}{c}{$\boldsymbol{(6,2)\{1\}}\succ(4,4)$}\\
	\hline
	\multicolumn{2}{|c|}{$\hat{\varphi}_1=(\{1,11\},\{2,6,10\},\{3,7,8\},\{4,5,9\})$ ({\Large\textasteriskcentered})}\\
	\hdashline[1pt/3pt]
	$\ell$ of $\pi_\ell$ & elimination \\
	$1-24$ & $\text{C}1 \vee \text{C9}$\\
	\hline
	\end{tabular}
	}
\\
\vspace{0.37cm}
\\
	\scalebox{\scaleboxes}{
	\begin{tabular}{|C{3.7cm}C{3.9cm}|}
	\multicolumn{2}{c}{$\boldsymbol{(6,2)\{2\}}\succ(4,4)$}\\
	\hline
	\multicolumn{2}{|c|}{$\hat{\varphi}_1=(\{1,11,12\},\{2,6,10\},\{3,7,8\},\{4,5,9\})$ ({\Large\textasteriskcentered})}\\
	\hdashline[1pt/3pt]
	$\ell$ of $\pi_\ell$ & elimination \\
	$1-24$ & always\\
	\hline
	\end{tabular}
	}
\\
\vspace{0.37cm}
\\
	\scalebox{\scaleboxes}{
	\begin{tabular}{|C{3.7cm}C{3.9cm}|}
	\multicolumn{2}{c}{$\boldsymbol{(6,2)\{3\}}\succ(4,4)$}\\
	\hline
	\multicolumn{2}{|c|}{$\hat{\varphi}_1=(\{1,8,12,13\},\{2,6,11\},$}\\
	\multicolumn{2}{|c|}{$\{3,7,9\},\{4,5,10\})$ ({\Large\textasteriskcentered})}\\
	\hdashline[1pt/3pt]
	$\ell$ of $\pi_\ell$ & elimination \\
	$1,2,7,8,17,18,23,24$ & $\text{C1} \vee \overline{\text{C9}}$\\
	$3,5,9,11,14,16,20,22$ & $\text{C2} \vee \overline{\text{C9}}$\\
	$4,6,10,12,13,15,19,21$ & $\text{C3} \vee \overline{\text{C9}}$\\
	\hline
	\multicolumn{2}{|c|}{$\hat{\varphi}_1=(\{1,8,11\},\{2,6,12,13\},$}\\
	\multicolumn{2}{|c|}{$\{3,7,9\},\{4,5,10\})$ ({\Large\textasteriskcentered})}\\
	\hdashline[1pt/3pt]
	$\ell$ of $\pi_\ell$ & elimination \\
	$1-24$ & always\\
	\hline
	\end{tabular}
	}
\\
\vspace{0.37cm}
\\
	\scalebox{\scaleboxes}{
	\begin{tabular}{|C{3.7cm}C{3.9cm}|}
	\multicolumn{2}{c}{$\boldsymbol{(6,2)\{4\}}\succ(5,4)$}\\
	\hline
	\multicolumn{2}{|c|}{$\hat{\varphi}_1=(\{1,9,10\},\{2,6,12,13\},$}\\
	\multicolumn{2}{|c|}{$\{3,7,8\},\{4,5,11\})$ ({\Large\textasteriskcentered})}\\
	\hdashline[1pt/3pt]
	$\ell$ of $\pi_\ell$ & elimination \\
	$1-24$ & always\\
	\hline
	\end{tabular}
	}
\\
\vspace{0.37cm}
\\
	\scalebox{\scaleboxes}{
	\begin{tabular}{|C{3.7cm}C{3.9cm}|}
	\multicolumn{2}{c}{$\boldsymbol{(6,4)\{1\}}$}\\
	\hline
	\multicolumn{2}{|c|}{$\hat{\varphi}_1=(\{1,11,12,13\},\{2,6,10\},$}\\
	\multicolumn{2}{|c|}{$\{3,7,8\},\{4,5,9\})$ ({\Large\textasteriskcentered})}\\
	\hdashline[1pt/3pt]
	$\ell$ of $\pi_\ell$ & elimination \\
	$1-24$ & $\text{C}8 \vee \overline{\text{C}1}$ \\
	\hline
	\end{tabular}
	}
\end{minipage}
\begin{minipage}[t]{0.337\textwidth}
\mbox{}\\[-\baselineskip]

\renewcommand{\arraystretch}{1.3}
\newcolumntype{L}[1]{>{\raggedright\let\newline\\\arraybackslash\hspace{0pt}}m{#1}}
\newcolumntype{C}[1]{>{\centering\let\newline\\\arraybackslash\hspace{0pt}}m{#1}}
\newcolumntype{R}[1]{>{\raggedleft\let\newline\\\arraybackslash\hspace{0pt}}m{#1}}

	\scalebox{\scaleboxes}{
	\begin{tabular}{|C{3.7cm}C{4.4cm}|}
	\multicolumn{2}{c}{$\boldsymbol{(6,4)\{2\}}$}\\
	\hline
	\multicolumn{2}{|c|}{$\hat{\varphi}_1=(\{1,9,11,13,14\},\{2,6,12\},\{3,7,8\},\{4,5,10\})$ ({\Large\textasteriskcentered})}\\
	\hdashline[1pt/3pt]
	$\ell$ of $\pi_\ell$ & elimination \\
	$1-24$ & $\text{C4} \vee \overline{\text{C16}}$\\
	\hline
	\multicolumn{2}{|c|}{$\hat{\varphi}_2=(\{1,9,11,12\},\{2,6,13,14\},\{3,7,8\},\{4,5,10\})$}\\
	\hdashline[1pt/3pt]
	$\ell$ of $\pi_\ell$ & elimination \\
	$1-24$ & always\\
	\hline
	\multicolumn{2}{|c|}{$\hat{\varphi}_3=(\{1,9,10,13\},\{2,6,12\},\{3,7,8\},\{4,5,11,14\})$}\\
	\hdashline[1pt/3pt]
	$\ell$ of $\pi_\ell$ & elimination \\
	$1,8,17,24$ & $\text{C10} \vee \overline{\text{C17}}$\\
	$2,7,18,23$ & $\text{C11} \vee \overline{\text{C17}}$\\
	$3,11,14,22$ & $\text{C12} \vee \overline{\text{C17}}$\\
	$4,12,13,21$ & $\text{C13} \vee \overline{\text{C17}}$\\
	$5,9,16,20$ & $\text{C14} \vee \overline{\text{C17}}$\\
	$6,10,15,19$ & $\text{C15} \vee \overline{\text{C17}}$\\
	\hline
	\multicolumn{2}{|c|}{$\hat{\varphi}_4=(\{1,8,10\},\{2,6,12\},\{3,7,9,13\},\{4,5,11,14\})$}\\
	\hdashline[1pt/3pt]
	$\ell$ of $\pi_\ell$ & elimination \\
	$1-24$ & $\text{C9} \vee \overline{\text{C1}}$\\
	\hline
	\multicolumn{2}{|c|}{$\hat{\varphi}_5=(\{1,8,10\},\{2,7,11,13\},\{3,5,12\},\{4,6,9,14\})$}\\
	\hdashline[1pt/3pt]
	$\ell$ of $\pi_\ell$ & elimination \\
	$1-24$ & $\text{C8} \vee \overline{\text{C1}} \vee \overline{\text{C4}}$\\
	\hline
	\multicolumn{2}{|c|}{$\hat{\varphi}_6=(\{1,8,10\},\{2,6,13,14\},\{3,5,12\},\{4,7,9,11\})$}\\
	\hdashline[1pt/3pt]
	$\ell$ of $\pi_\ell$ & elimination \\
	$1-24$ & always\\
	\hline
	\multicolumn{2}{|c|}{$\hat{\varphi}_7=(\{1,8,10\},\{2,6,12\},\{3,5,13,14\},\{4,7,9,11\})$}\\
	\hdashline[1pt/3pt]
	$\ell$ of $\pi_\ell$ & elimination \\
	$1,8,17,24$ & $\text{C12}$\\
	$2,7,18,23$ & $\text{C13}$\\
	$3,11,14,22$ & $\text{C10}$\\
	$4,12,13,21$ & $\text{C11}$\\
	$5,9,16,20$ & $\text{C15}$\\
	$6,10,15,19$ & $\text{C14}$\\
	\hline
	\end{tabular}
	}
\\
\vspace{0.32cm}
\\
	\scalebox{\scaleboxes}{
	\begin{tabular}{|C{3.7cm}C{4.4cm}|}
	\multicolumn{2}{c}{$\boldsymbol{(6,4)\{3\}}$}\\
	\hline
	\multicolumn{2}{|c|}{$\hat{\varphi}_1=(\{1,9,10\},\{2,6,12\},\{3,7,8,14\},\{4,5,11,13\})$}\\
	\hdashline[1pt/3pt]
	$\ell$ of $\pi_\ell$ & elimination \\
	$1,2,7,8,17,18,23,24$ & $\text{C16} \vee \overline{\text{C4}}$\\
	$3,5,9,11,14,16,20,22$ & $\text{C18} \vee \overline{\text{C4}}$\\
	$4,6,10,12,13,15,19,21$ & $\text{C19} \vee \overline{\text{C4}}$\\
	\hline
	\multicolumn{2}{|c|}{$\hat{\varphi}_2=(\{1,9,10\},\{2,6,13,14\},\{3,5,12\},\{4,7,8,11\})$}\\
	\hdashline[1pt/3pt]
	$\ell$ of $\pi_\ell$ & elimination \\
	$1,8,17,24$ & $\text{C20}$\\
	$2,7,18,23$ & $\text{C21}$\\
	$3,11,14,22$ & $\text{C22}$\\
	$4,12,13,21$ & $\text{C23}$\\
	$5,9,16,20$ & $\text{C24}$\\
	$6,10,15,19$ & $\text{C25}$\\
	\hline
	\end{tabular}
	}
\\
\vspace{0.32cm}
\\
	\scalebox{\scaleboxes}{
	\begin{tabular}{|C{3.7cm}C{4.4cm}|}
	\multicolumn{2}{c}{$\boldsymbol{(6,4)\{4\}}$}\\
	\hline
	\multicolumn{2}{|c|}{$\hat{\varphi}_1=(\{1,9,11,13,14\},\{2,6,12\},\{3,7,8\},\{4,5,10\})$ ({\Large\textasteriskcentered})}\\
	\hdashline[1pt/3pt]
	$\ell$ of $\pi_\ell$ & elimination \\
	$1-24$ & $\overline{\text{C}1} \vee \overline{\text{C2}} \vee \overline{\text{C3}} \vee \overline{\text{C8}}$ \\
	\hline
	\multicolumn{2}{|c|}{$\hat{\varphi}_2=(\{1,9,11,12\},\{2,6,13,14\},\{3,7,8\},\{4,5,10\})$}\\
	\hdashline[1pt/3pt]
	$\ell$ of $\pi_\ell$ & elimination \\
	$1-24$ & always\\
	\hline
	\end{tabular}
	}
\end{minipage}
\begin{minipage}[t]{0.31\textwidth}
\mbox{}\\[-\baselineskip]

\renewcommand{\arraystretch}{1.3}
\newcolumntype{L}[1]{>{\raggedright\let\newline\\\arraybackslash\hspace{0pt}}m{#1}}
\newcolumntype{C}[1]{>{\centering\let\newline\\\arraybackslash\hspace{0pt}}m{#1}}
\newcolumntype{R}[1]{>{\raggedleft\let\newline\\\arraybackslash\hspace{0pt}}m{#1}}

	\scalebox{\scaleboxes}{
	\begin{tabular}{|C{3.7cm}C{4.4cm}|}
	\multicolumn{2}{c}{$\boldsymbol{(6,6)\{1\}}$}\\
	\hline
	\multicolumn{2}{|c|}{$\hat{\varphi}_1=(\{1,9,12,13\},\{2,7,11,14\},\{3,6,8\},\{4,5,10,15\})$}\\
	\hdashline[1pt/3pt]
	$\ell$ of $\pi_\ell$ & elimination \\
	$1,3,8,11,14,17,22,24$ & $\text{C3}$ \\
	$2,4,7,12,13,18,21,23$ & $\text{C2}$ \\
	$5,6,9,10,15,16,19,20$ & $\text{C1}$ \\
	\hline
	\multicolumn{2}{|c|}{$\hat{\varphi}_2=(\{1,9,11,14\},\{2,7,12,13\},\{3,6,8\},\{4,5,10,15\})$}\\
	\hdashline[1pt/3pt]
	$\ell$ of $\pi_\ell$ & elimination \\
	$1-24$ & $\text{C8}$ \\
	\hline
	\multicolumn{2}{|c|}{$\hat{\varphi}_3=(\{1,9,10,15\},\{2,7,11,14\},\{3,6,8\},\{4,5,12,13\})$}\\
	\hdashline[1pt/3pt]
	$\ell$ of $\pi_\ell$ & elimination \\
	$1-24$ & no \\
	\hline
	\multicolumn{2}{|c|}{$\hat{\varphi}_4=(\{1,8,10\},\{2,7,11,14\},\{3,6,9,15\},\{4,5,12,13\})$}\\
	\hdashline[1pt/3pt]
	$\ell$ of $\pi_\ell$ & elimination \\
	$1-24$ & always \\
	\hline
	\multicolumn{2}{|c|}{$\hat{\varphi}_5=(\{1,8,10\},\{2,7,12,13\},\{3,6,9,15\},\{4,5,11,14\})$}\\
	\hdashline[1pt/3pt]
	$\ell$ of $\pi_\ell$ & elimination \\
	$1,6,8,10,15,17,19,24$ & $\overline{\text{C2}} \vee \overline{\text{C4}}\vee \overline{\text{C11}}\vee \overline{\text{C13}}$\\
	$2,5,7,9,16,18,20,23$ & $\overline{\text{C3}} \vee \overline{\text{C4}}\vee \overline{\text{C10}}\vee \overline{\text{C12}}$\\
	$3,4,11,12,13,14,21,22$ & $\overline{\text{C1}} \vee \overline{\text{C4}}\vee \overline{\text{C14}}$\\
	\hline
	\multicolumn{2}{|c|}{$\hat{\varphi}_6=(\{1,8,12,14\},\{2,6,11\},\{3,7,9,13\},\{4,5,10,15\})$}\\
	\hdashline[1pt/3pt]
	$\ell$ of $\pi_\ell$ & elimination \\
	$1-24$ & $\text{C1} \wedge (\overline{\text{C2}} \vee \overline{\text{C3}} \vee \overline{\text{C8}} \vee \overline{\text{C9}}) \vee (\overline{\text{C1}} \wedge \text{C9})$ \\
	\hline
	\multicolumn{2}{|c|}{$\hat{\varphi}_7=(\{1,8,10\},\{2,6,11\},\{3,7,9,14,15\},\{4,5,12,13\})$}\\
	\hdashline[1pt/3pt]
	$\ell$ of $\pi_\ell$ & elimination \\
	$1,2,7,8,17,18,23,24$ & $\text{C1}$\\
	$3,5,9,11,14,16,20,22$ & $\text{C2}$\\
	$4,6,10,12,13,15,19,21$ & $\text{C3}$\\
	\hline
	\multicolumn{2}{|c|}{$\hat{\varphi}_8=(\{1,8,10\},\{2,6,11\},\{3,7,9,13\},\{4,5,12,14,15\})$}\\
	\hdashline[1pt/3pt]
	$\ell$ of $\pi_\ell$ & elimination \\
	$1,2,7,8,17,18,23,24$ & $\overline{\text{C1}} \vee \overline{\text{C4}} \vee \overline{\text{C14}} \vee \overline{\text{C26}}$\\
	$3,5,9,11,14,16,20,22$ & $\overline{\text{C2}} \vee \overline{\text{C4}} \vee \overline{\text{C11}} \vee \overline{\text{C13}} \vee \overline{\text{C27}}$\\
	$4,6,10,12,13,15,19,21$ & $\overline{\text{C3}} \vee \overline{\text{C4}} \vee \overline{\text{C10}} \vee \overline{\text{C12}} \vee \overline{\text{C28}}$\\
	\hline
	\multicolumn{2}{|c|}{$\hat{\varphi}_9=(\{1,8,10\},\{2,6,11\},\{3,5,13\},$}\\
	\multicolumn{2}{|c|}{$\{4,7,9,12,14,15\})$ ({\Large\textasteriskcentered})}\\
	\hdashline[1pt/3pt]
	CP's & elimination \\
	$1-24$ & $\text{C4}$\\
	\hline
	\end{tabular}
	}
\\
\vspace{0.39cm}
\\
	\scalebox{\scaleboxes}{
	\begin{tabular}{|C{3.7cm}C{4.4cm}|}
	\multicolumn{2}{c}{$\boldsymbol{(6,6)\{2\}}$}\\
	\hline
	\multicolumn{2}{|c|}{$\hat{\varphi}_1=(\{1,10,12,14\},\{2,7,11,15\},\{3,6,8\},\{4,5,9,13\})$}\\
	\hdashline[1pt/3pt]
	$\ell$ of $\pi_\ell$ & elimination \\
	$1-24$ & $\text{C4}$ \\
	\hline
	\multicolumn{2}{|c|}{$\hat{\varphi}_2=(\{1,10,11,15\},\{2,7,12,14\},\{3,6,8\},\{4,5,9,13\})$}\\
	\hdashline[1pt/3pt]
	$\ell$ of $\pi_\ell$ & elimination \\
	$1,2,7,8,17,18,23,24$ & $\text{C1}$ \\
	$3,6,9,11,14,16,20,22$ & $\text{C2}$ \\
	$4,5,10,12,13,15,19,21$ & $\text{C3}$ \\
	\hline
	\multicolumn{2}{|c|}{$\hat{\varphi}_3=(\{1,9,11\},\{2,7,12,14\},\{3,6,8\},\{4,5,10,13,15\})$}\\
	\hdashline[1pt/3pt]
	$\ell$ of $\pi_\ell$ & elimination \\
	$1-24$ & always \\
	\hline
	\multicolumn{2}{|c|}{$\hat{\varphi}_4=(\{1,8,12\},\{2,6,11\},\{3,5,9\},$}\\
	\multicolumn{2}{|c|}{$\{4,7,10,13,14,15\})$ ({\Large\textasteriskcentered})}\\
	\hdashline[1pt/3pt]
	$\ell$ of $\pi_\ell$ & elimination \\
	$1-24$ & no \\
	\hline
	\end{tabular}
	}
\vspace{0.5cm}
\end{minipage}

\renewcommand{\baselinestretch}{1.2}
\caption{Each table deals with an $(a,b)\{i\}$ absorbing set candidate $(\mathcal{P},\mathcal{B})$ of block size $k=4$ listed in Fig.~\ref{fig:absorbing_set_classification}. Each table is divided into subtables presenting all non-isomorphic $4$-colourings $\hat{\varphi}_j$ of the candidate (in short notation). For each 4-colouring $\varphi_j$, the subtable gives the exact constraints in order to eliminate the absorbings sets of type $(\mathcal{P},\mathcal{B},\varphi_j)$ with colour-to-group mapping $\pi_\ell$ in an $\mathscr{L}_q^{2}$-TD LDPC code of column weight $4$. The colour-to-group mappings are explicitely given in Table~\ref{cps_k4} and the referred constraints C1-C28 are listed in Table~\ref{constraints_k4}. If the $4$-colouring $\varphi_j$ is marked by an asterisk ({\normalsize\textasteriskcentered}), then the absorbing sets of type $(\mathcal{P},\mathcal{B},\varphi_j)$ are fully. The symbol `$\succ$' means that the left-hand candidate is an extension of the right-hand candidate such that the elimination of the smaller candidate leads to the elimination of the extension.
}
\label{classification}
\end{figure*}

\subsection{Strategy for $\mathscr{L}_q^m$-TD Codes of Higher Column Weights}
\label{strategy_for_higher_column_weighs}

Let $\CMcal{D}$ be an $\mathscr{L}_q^m$-TD of block size $k=m+2$.
By removing the points of any $t$ groups of $\CMcal{D}$, we obtain an $\mathscr{L}_q^{(m-t)}$-TD $\CMcal{D}'$ of block size $k-t$. 
Since $\CMcal{D}'$ is embedded in $\CMcal{D}$, it is highly possible that an absorbing set $\CMcal{A}$ of $\CMcal{C}(\CMcal{D})$ also represents an absorbing set of $\CMcal{C}(\CMcal{D}')$ or contains a smaller absorbing set $\CMcal{A}'$ of $\CMcal{C}(\CMcal{D}')$.
Hence, the elimination of absorbing sets in $\CMcal{C}(\CMcal{D}')$ of column weight $k-t$ leads to the avoidance of absorbing sets in $\CMcal{C}(\CMcal{D})$ of column weight $k$.
Therefore, we conjecture that we obtain beneficial $\mathscr{L}_q^m$-TD LDPC codes of higher column weights if the scale factors $\alpha_1,\hdots,\alpha_m$ pairwise satisfy the constraints C1-C16 and C18-C25 of Table~\ref{constraints_k4} which is the design strategy for the case $k=4$.

\begin{table*}[!t]
\renewcommand{\arraystretch}{1.25}
\newcolumntype{C}{>{\centering\arraybackslash $}m{0.49cm}<{$}}
\newcommand{\nm}[1]{\textnormal{#1}}
\caption{Absorbing Set Detections for all $\mathscr{L}_{13}^2$-TD LDPC Codes over the AWGN channel with $E_b/N_0=5 \textnormal{ dB}$ and $10^8$ Codewords.}
\label{table_absorbingsets}
\begin{center}
\scalebox{0.85}{
\begin{tabular}{|@{}>{\centering\ \ \,$}p{1.3cm}<{$}|*{5}{CCC;{1pt/3pt}}CCC|}
\hline
&
\multicolumn{18}{c|}{Each pair $[1,\alpha]$ represents an $\mathscr{L}_{13}^2$-TD LDPC code based on the MOLS $\{\CMcal{L}_{13}^{(1)},\CMcal{L}_{13}^{(\alpha)}\}\in \mathscr{L}_{13}^2$}
\\
\cline{2-19}
&
\multicolumn{3}{c;{1pt/3pt}}{$\boldsymbol{[1,2],[1,7]}$}&
\multicolumn{3}{c;{1pt/3pt}}{$\boldsymbol{[1,3],[1,9]}$}&
\multicolumn{3}{c;{1pt/3pt}}{$\boldsymbol{[1,4],[1,10]}$}&
\multicolumn{3}{c;{1pt/3pt}}{$\boldsymbol{[1,5],[1,8]}$}&
\multicolumn{3}{c;{1pt/3pt}}{$\boldsymbol{[1,6],[1,11]}$}&
\multicolumn{3}{c|}{$\boldsymbol{[1,12]}$}
\\
\boldsymbol{(a,b)} &
\nm{total}&\nm{fully}&\nm{elem.}&
\nm{total}&\nm{fully}&\nm{elem.}&
\nm{total}&\nm{fully}&\nm{elem.}&
\nm{total}&\nm{fully}&\nm{elem.}&
\nm{total}&\nm{fully}&\nm{elem.}&
\nm{total}&\nm{fully}&\nm{elem.}
\\
\hline
(4,4)&
1 & 0 & 1 &
0 & 0 & 0 &
0 & 0 & 0 &
0 & 0 & 0 &
0 & 0 & 0 &
3 & 0 & 3
\\
(6,2)&
4129 & 4129 & 4129 &
0 & 0 & 0 &
0 & 0 & 0 &
0 & 0 & 0 &
0 & 0 & 0 &
3977 & 3977 & 3977
\\
(6,4)&
0 & 0 & 0 &
0 & 0 & 0 &
3 & 0 & 3 &
1 & 0 & 1 &
5 & 3 & 5 &
0 & 0 & 0
\\
(6,6)&
0 & 0 & 0 &
0 & 0 & 0 &
0 & 0 & 0 &
7 & 0 & 7 &
0 & 0 & 0 &
0 & 0 & 0
\\
(7,4)&
0 & 0 & 0 &
0 & 0 & 0 &
0 & 0 & 0 &
0 & 0 & 0 &
1 & 0 & 1 &
0 & 0 & 0
\\
(7,6)&
0 & 0 & 0 &
0 & 0 & 0 &
0 & 0 & 0 &
0 & 0 & 0 &
2 & 0 & 2 &
0 & 0 & 0
\\
(8,0)&
551 & 551 & 551 &
0 & 0 & 0 &
0 & 0 & 0 &
0 & 0 & 0 &
0 & 0 & 0 &
517 & 517 & 517
\\
(8,2)&
54 & 54 & 54 &
257 & 257 & 220 &
134 & 134 & 134 &
241 & 241 & 217 &
225 & 225 & 205 &
47 & 47 & 47
\\
(8,4)&
0 & 0 & 0 &
4 & 0 & 4 &
0 & 0 & 0 &
0 & 0 & 0 &
3 & 0 & 3 &
0 & 0 & 0
\\
(9,4)&
0 & 0 & 0 &
0 & 0 & 0 &
0 & 0 & 0 &
0 & 0 & 0 &
5 & 0 & 5 &
1 & 0 & 1
\\
(9,6)&
0 & 0 & 0 &
0 & 0 & 0 &
0 & 0 & 0 &
1 & 0 & 1 &
0 & 0 & 0 &
0 & 0 & 0
\\
(10,0)&
0 & 0 & 0 &
19 & 19 & 19 &
42 & 42 & 28 &
21 & 21 & 21 &
20 & 20 & 20 &
0 & 0 & 0
\\
(10,2)&
51 & 51 & 31 &
54 & 54 & 24 &
70 & 70 & 42 &
61 & 61 & 36 &
63 & 63 & 43 &
77 & 77 & 29
\\
(10,4)&
0 & 0 & 0 &
0 & 0 & 0 &
1 & 0 & 1 &
5 & 3 & 1 &
0 & 0 & 0 &
0 & 0 & 0
\\
(11,4)&
1 & 0 & 0 &
0 & 0 & 0 &
0 & 0 & 0 &
0 & 0 & 0 &
0 & 0 & 0 &
0 & 0 & 0
\\
(11,6)&
0 & 0 & 0 &
0 & 0 & 0 &
2 & 0 & 0 &
0 & 0 & 0 &
0 & 0 & 0 &
0 & 0 & 0
\\
(12,0)&
4 & 4 & 2 &
7 & 7 & 3 &
4 & 4 & 3 &
2 & 2 & 0 &
0 & 0 & 0 &
9 & 9 & 1
\\
(12,2)&
8 & 8 & 1 &
26 & 26 & 14 &
5 & 5 & 2 &
12 & 12 & 5 &
19 & 19 & 3 &
16 & 16 & 3
\\
(14,0)&
2 & 2 & 0 &
2 & 2 & 1 &
0 & 0 & 0 &
1 & 1 & 1 &
0 & 0 & 0 &
0 & 0 & 0
\\
(14,2)&
2 & 2 & 2 &
9 & 9 & 0 &
4 & 4 & 1 &
6 & 6 & 1 &
2 & 2 & 0 &
5 & 5 & 2
\\
(16,0)&
2 & 2 & 0 &
0 & 0 & 0 &
0 & 0 & 0 &
1 & 1 & 0 &
0 & 0 & 0 &
0 & 0 & 0
\\
(16,2)&
0 & 0 & 0 &
0 & 0 & 0 &
0 & 0 & 0 &
2 & 2 & 0 &
0 & 0 & 0 &
4 & 4 & 0
\\
(18,0)&
0 & 0 & 0 &
0 & 0 & 0 &
0 & 0 & 0 &
1 & 1 & 0 &
0 & 0 & 0 &
0 & 0 & 0
\\
(18,2)&
0 & 0 & 0 &
0 & 0 & 0 &
1 & 1 & 0 &
0 & 0 & 0 &
0 & 0 & 0 &
0 & 0 & 0
\\
\hdashline[1pt/3pt]
\textbf{BER}&
\multicolumn{3}{c;{1pt/3pt}}{$20.624\cdot 10^{-5}$}&
\multicolumn{3}{c;{1pt/3pt}}{$5.058\cdot 10^{-5}$}&
\multicolumn{3}{c;{1pt/3pt}}{$4.560\cdot 10^{-5}$}&
\multicolumn{3}{c;{1pt/3pt}}{$4.965\cdot 10^{-5}$}&
\multicolumn{3}{c;{1pt/3pt}}{$4.968\cdot 10^{-5}$}&
\multicolumn{3}{c|}{$20.217 \cdot 10^{-5}$}
\\
\textbf{FER}&
\multicolumn{3}{c;{1pt/3pt}}{$51.96\cdot 10^{-4}$}&
\multicolumn{3}{c;{1pt/3pt}}{$8.39\cdot 10^{-4}$}&
\multicolumn{3}{c;{1pt/3pt}}{$7.30\cdot 10^{-4}$}&
\multicolumn{3}{c;{1pt/3pt}}{$8.07\cdot 10^{-4}$}&
\multicolumn{3}{c;{1pt/3pt}}{$8.16\cdot 10^{-4}$}&
\multicolumn{3}{c|}{$50.4 \cdot 10^{-4}$}
\\
\hdashline[1pt/3pt]
\multirow{4}{*}{\!\!\!\textbf{Violations}}&
\multicolumn{3}{c;{1pt/3pt}}{$\boldsymbol{[1,2]}$:}&
\multicolumn{3}{c;{1pt/3pt}}{$\boldsymbol{[1,3]}$:}&
\multicolumn{3}{c;{1pt/3pt}}{$\boldsymbol{[1,4]}$:}&
\multicolumn{3}{c;{1pt/3pt}}{$\boldsymbol{[1,5]}$:}&
\multicolumn{3}{c;{1pt/3pt}}{$\boldsymbol{[1,6]}$:}&
\multicolumn{3}{c|}{$\boldsymbol{[1,12]}$:}
\\
&
\multicolumn{3}{c;{1pt/3pt}}{C2}&
\multicolumn{3}{c;{1pt/3pt}}{C15, C16, C28}&
\multicolumn{3}{c;{1pt/3pt}}{C8, C20, C23, C24}&
\multicolumn{3}{c;{1pt/3pt}}{C11, C18, C26}&
\multicolumn{3}{c;{1pt/3pt}}{C12, C19, C27}&
\multicolumn{3}{c|}{C1}
\\
&
\multicolumn{3}{c;{1pt/3pt}}{$\boldsymbol{[1,7]}$:}&
\multicolumn{3}{c;{1pt/3pt}}{$\boldsymbol{[1,9]}$:}&
\multicolumn{3}{c;{1pt/3pt}}{$\boldsymbol{[1,10]}$:}&
\multicolumn{3}{c;{1pt/3pt}}{$\boldsymbol{[1,8]}$:}&
\multicolumn{3}{c;{1pt/3pt}}{$\boldsymbol{[1,11]}$:}&
\multicolumn{3}{c|}{}
\\
&
\multicolumn{3}{c;{1pt/3pt}}{C3}&
\multicolumn{3}{c;{1pt/3pt}}{C14, C16, C27}&
\multicolumn{3}{c;{1pt/3pt}}{C8, C21, C22, C25}&
\multicolumn{3}{c;{1pt/3pt}}{C10, C19, C26}&
\multicolumn{3}{c;{1pt/3pt}}{C13, C18, C28}&
\multicolumn{3}{c|}{}
\\
\hline
\end{tabular}
}
\end{center}
\end{table*}

\section{Simulation Results}\label{simulations}

\subsection{Absorbing Set Detections for $\mathscr{L}_{13}^2$-TD LDPC Codes}

In Table~\ref{table_absorbingsets}, we simulated the transmission and decoding of $10^8$ codewords for all possible $\mathscr{L}_{13}^2$-TD LDPC codes over the AWGN channel with $E_b/N_0=5$~dB and under the standard sum-product algorithm with a maximum of 2000 iterations per codeword. 
For every code, we counted the number of absorbing set detections grouped by their sizes $(a,b)$ in three categories: The column ``total'' counts the total number of absorbing set detections throughout the simulations and the columns ``fully'' and ``elem.'' count the number of fully and elementary absorbing set detections, respectively. 
We also list the bit error rate (BER), the frame error rate (FER) and any violations of the constraints of Table~\ref{constraints_k4} for each code. 
Note that there are some codes that are equal according to Theorem~\ref{theorem:equivalence_of_TD_LDPC_codes} and thus are listed in the same column. 
For example, let $\CMcal{M}_1=\big\{ \CMcal{L}_{13}^{(1)}, \CMcal{L}_{13}^{(2)} \big\}$. For $\lambda=7$, we obtain that $\CMcal{C}(\CMcal{D}_{\CMcal{M}_1})=\CMcal{C}(\CMcal{D}_{\CMcal{M}_2})$ with $\CMcal{M}_2=\big\{ \CMcal{L}_{13}^{(7)}, \CMcal{L}_{13}^{(1)} \big\}$ by Theorem~\ref{theorem:equivalence_of_TD_LDPC_codes}.
These codes clearly show the same decoding behaviour. Notice that the Latin squares of $\CMcal{M}_2$ can be swapped without changing the code $\CMcal{C}(\CMcal{D}_{\CMcal{M}_2})$,
but since we have fixed the order of the MOLS by assigning the scale factors $\alpha_1$ and $\alpha_2$, the swapping of the Latin squares has the effect of interchanging the roles of $\alpha_1$ and $\alpha_2$. 
Hence, the codes $\CMcal{C}(\CMcal{D}_{\CMcal{M}_1})$ and $\CMcal{C}(\CMcal{D}_{\CMcal{M}_2})$ violate different elimination constraints although they are equivalent. In particular, the sets of constraints can be converted into each other by interchanging the roles of $\alpha_1$ and $\alpha_2$. 
Now, by considering the simulation results in Table~\ref{table_absorbingsets}, we can make the following observations:
\begin{itemize}[leftmargin=10pt,labelindent=5pt]
\item If an $(a,b)$ absorbing set has been detected frequently (say, more that 10 times), then it is fully.
\item If an $(a,b)$ absorbing set has been detected frequently, then a large fraction of the detected absorbing sets are elementary. 
\item The most harmful absorbing sets in this simulation are of size $(6,2)$ and $(8,0)$ which are fully and elementary. These absorbing sets can be avoided by satisfying C1-C3.
\item The conditions C4-C7 are not violated by any listed code such that there are no $(5,4)$ and $(6,0)$ absorbing sets. 
\item The codes represented by $[1, 4]$ and $[1, 10]$ seem to be the best choice among the codes with $q=13$, since they avoid the most harmful $(6,2)$ and $(8,0)$ absorbing sets and reduce the number of $(8,2)$ absorbing sets.
\item Some types of absorbing sets can not be avoided simultaneously. For instance, by avoiding the $(6,2)$ and $(8,0)$ absorbing sets, some $(10,0)$ absorbing sets occur and vice versa. A general explanation is that the count of all non-isomorphic configurations of size $t$ must sum up to $N\choose t$, where $N$ is the block length of the code, such that the avoidance of any configuration of size $t$ automatically increases the number other configurations of this size.
\end{itemize}

\subsection{Decoding Performance of Various $\mathscr{L}_{q}^m$-TD LDPC Codes}

In Fig.~\ref{Simu}, we again employed the standard SPA decoder with a maximum of 50 iterations per codeword to demonstrate the performance of our optimized TD LDPC codes compared to similar codes known in the literature. 
Firstly, we can observe that our methods of eliminating absorbing sets achieve an enormous performance gain between the TD LDPC codes with well and badly chosen scale factors. Secondly, the optimized TD LDPC codes outperform some known LDPC codes based on the random PEG algorithm \cite{Hu64}, on a construction by Zhang et al. \cite{Zhang2010} and on the Lattice construction by Vasic and Milenkovic \cite{Vasic04}.

\begin{figure*}[t!]
	\renewcommand{\baselinestretch}{1.2}
	\centering
	\subfloat[Performance of two $\mathscr{L}_{29}^2$-TD LDPC codes based on the MOLS $\{\CMcal{L}_{29}^{(1)}, \CMcal{L}_{29}^{(2)}\}$ (bad choice) and $\{\CMcal{L}_{29}^{(1)}, \CMcal{L}_{29}^{(12)}\}$ (good choice), respectively, compared to a random LDPC code constructed by the PEG algorithm \cite{Hu64}.]{
		\includegraphics[scale = 0.7, trim=5cm 8cm 5.5cm 8cm]{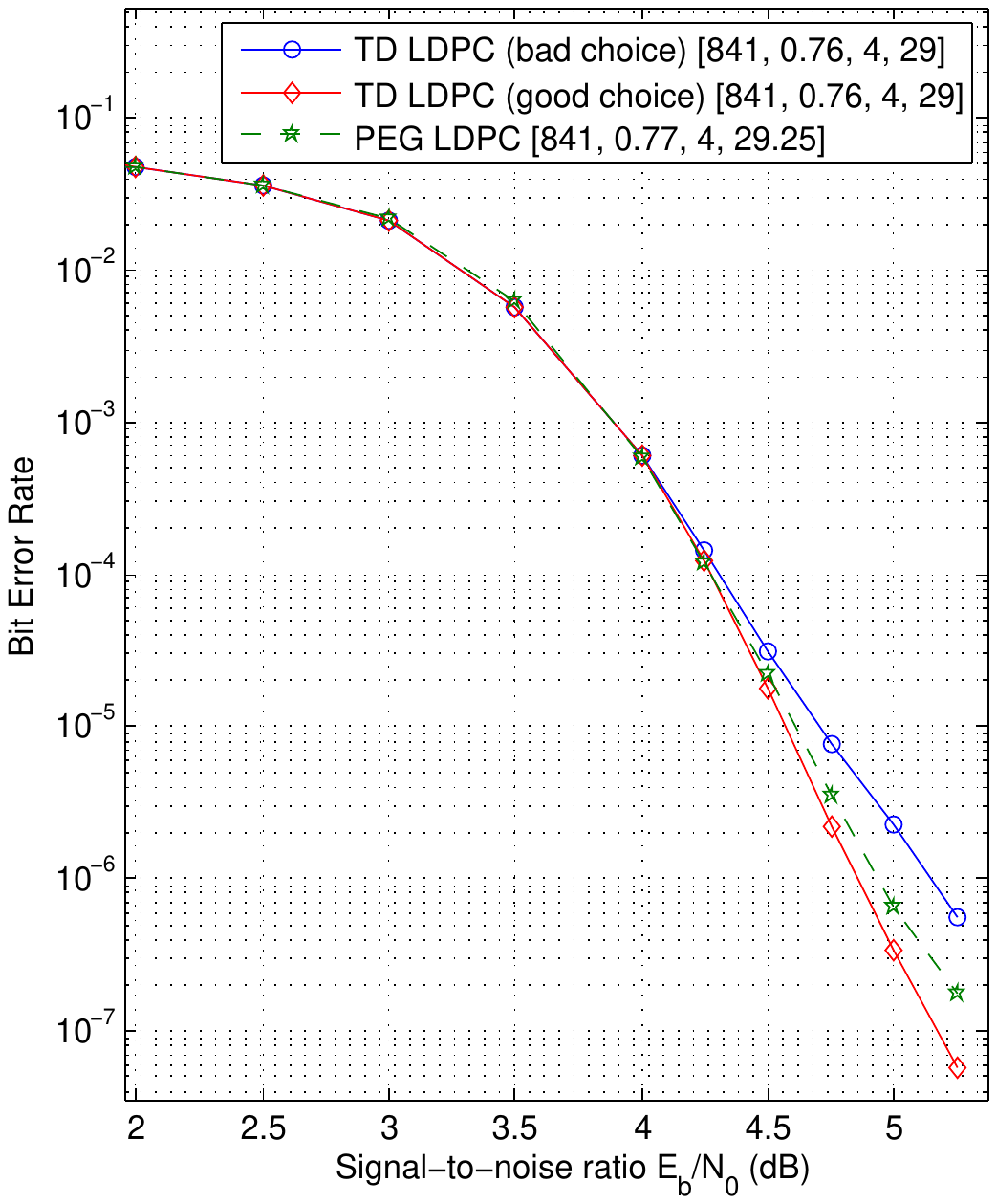}
	}
	\hspace{0.5cm}
	\subfloat[Performance of two $\mathscr{L}_{41}^2$-TD LDPC codes based on the MOLS $\{\CMcal{L}_{41}^{(1)}, \CMcal{L}_{41}^{(2)}\}$ (bad choice) and $\{\CMcal{L}_{41}^{(1)}, \CMcal{L}_{41}^{(9)}\}$ (good choice), respectively, compared to a random LDPC code constructed by the PEG algorithm and a quasi-cyclic LDPC code constructed by Zhang et al. \cite{Zhang2010}.]{
		\includegraphics[scale = 0.7, trim=5cm 8cm 5.5cm 8cm]{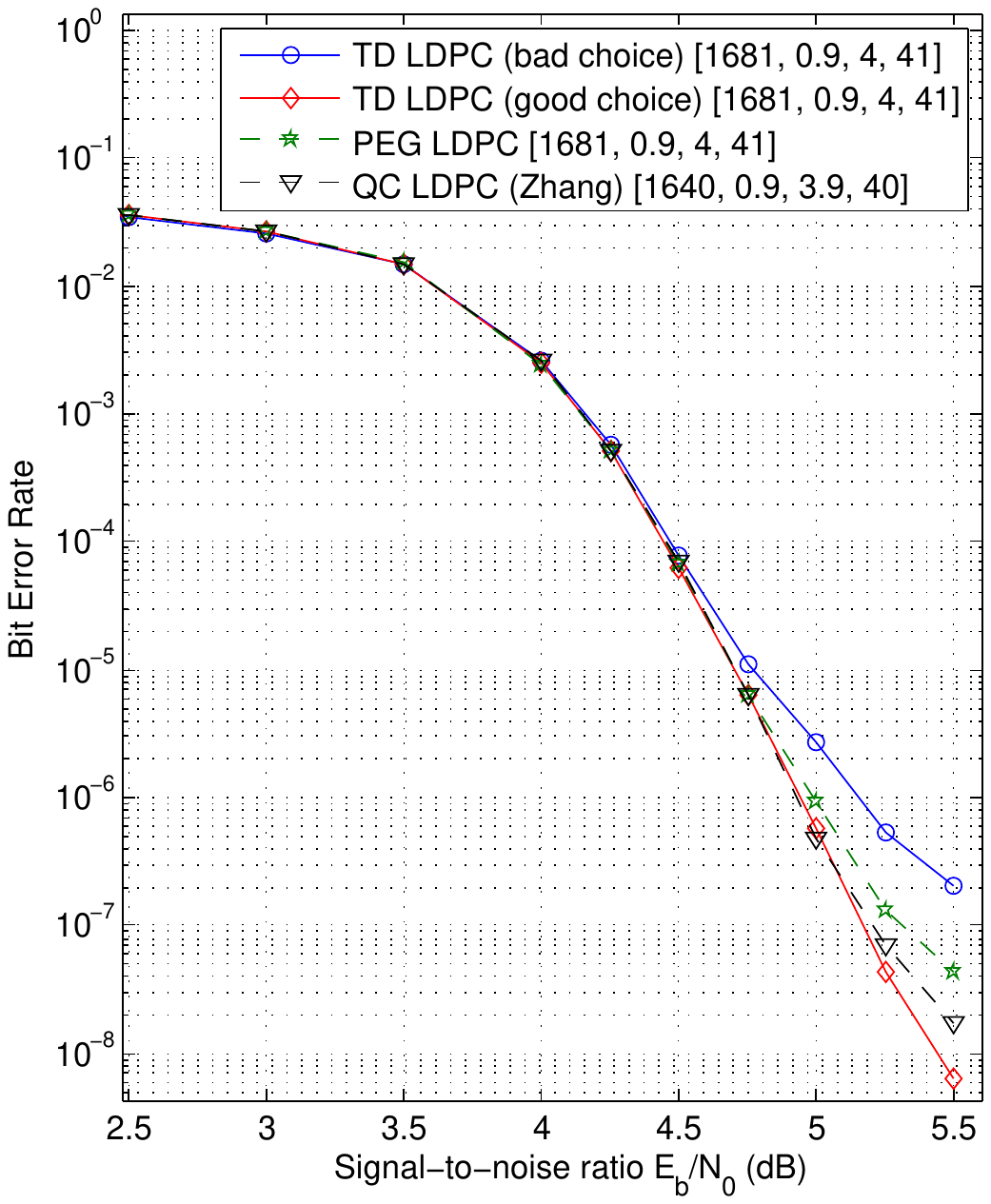}
	}
	\vspace{-0.6cm}
	\subfloat[Performance of two $\mathscr{L}_{47}^2$-TD LDPC codes based on the MOLS $\{\CMcal{L}_{47}^{(1)}, \CMcal{L}_{47}^{(2)}\}$ (bad choice) and $\{\CMcal{L}_{47}^{(1)}, \CMcal{L}_{47}^{(5)}\}$ (good choice), respectively, compared to a random LDPC code constructed by the PEG algorithm and an LDPC code based on the Lattice construction by Vasic and Milenkovic \cite{Vasic04}.]{
		\includegraphics[scale = 0.7, trim=5cm 8cm 5.5cm 8cm]{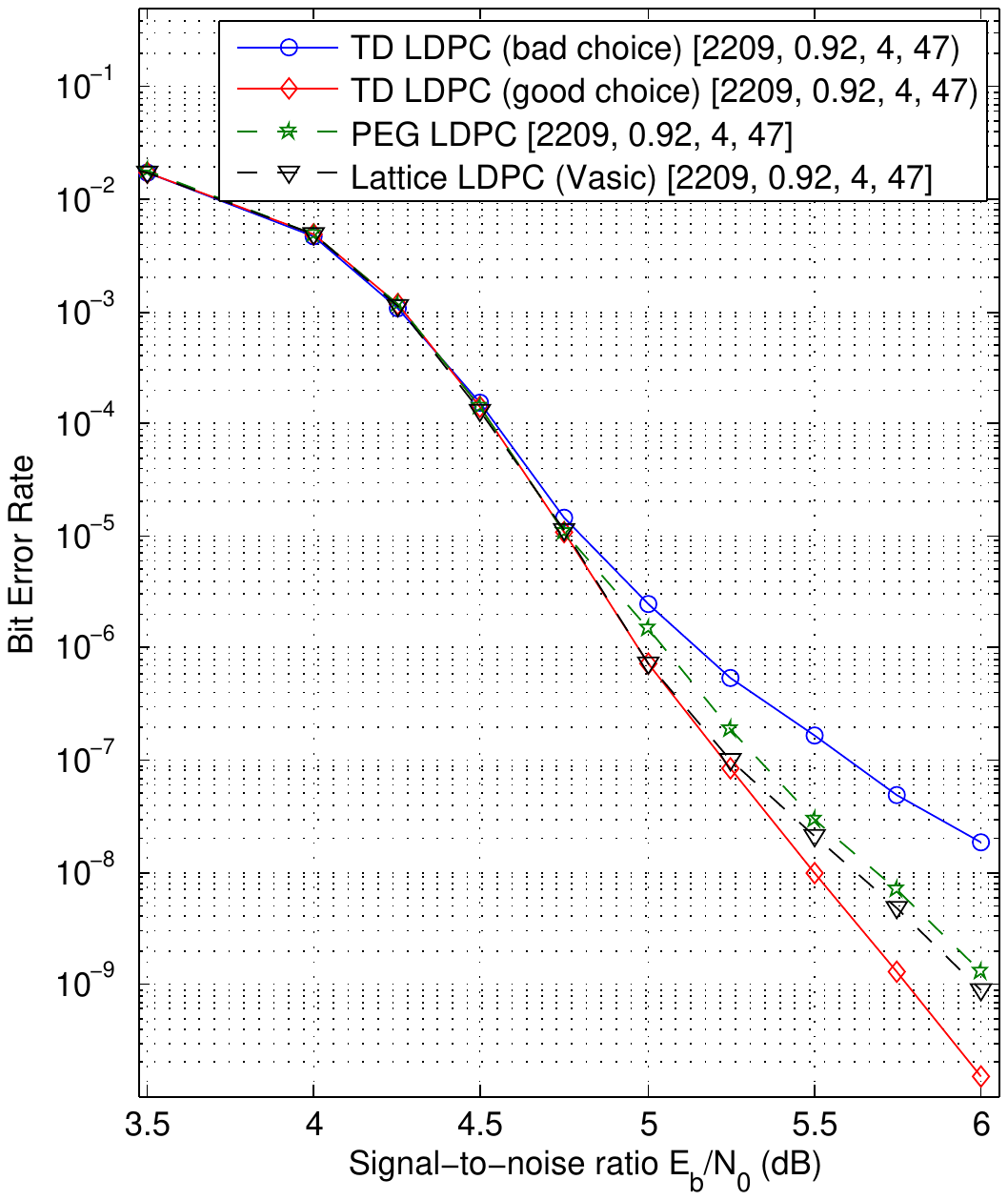}
	}
	\hspace{0.5cm}
	\subfloat[Performance of two $\mathscr{L}_{41}^3$-TD LDPC codes of column weight 5 based on the MOLS $\{\CMcal{L}_{41}^{(1)}, \CMcal{L}_{41}^{(2)} , \CMcal{L}_{41}^{(40)}\}$ (bad choice) and $\{\CMcal{L}_{41}^{(1)}, \CMcal{L}_{41}^{(5)}, \CMcal{L}_{41}^{(9)}\}$ (good choice). Note that the scale factors $\alpha_1=1,\alpha_2=5$ and $\alpha_3=9$ pairwise satisfy the constraints C1-C16 and C18-C25 as proposed in Subsection~\ref{strategy_for_higher_column_weighs}.]{
		\includegraphics[scale = 0.7, trim=5cm 8cm 5.5cm 8cm]{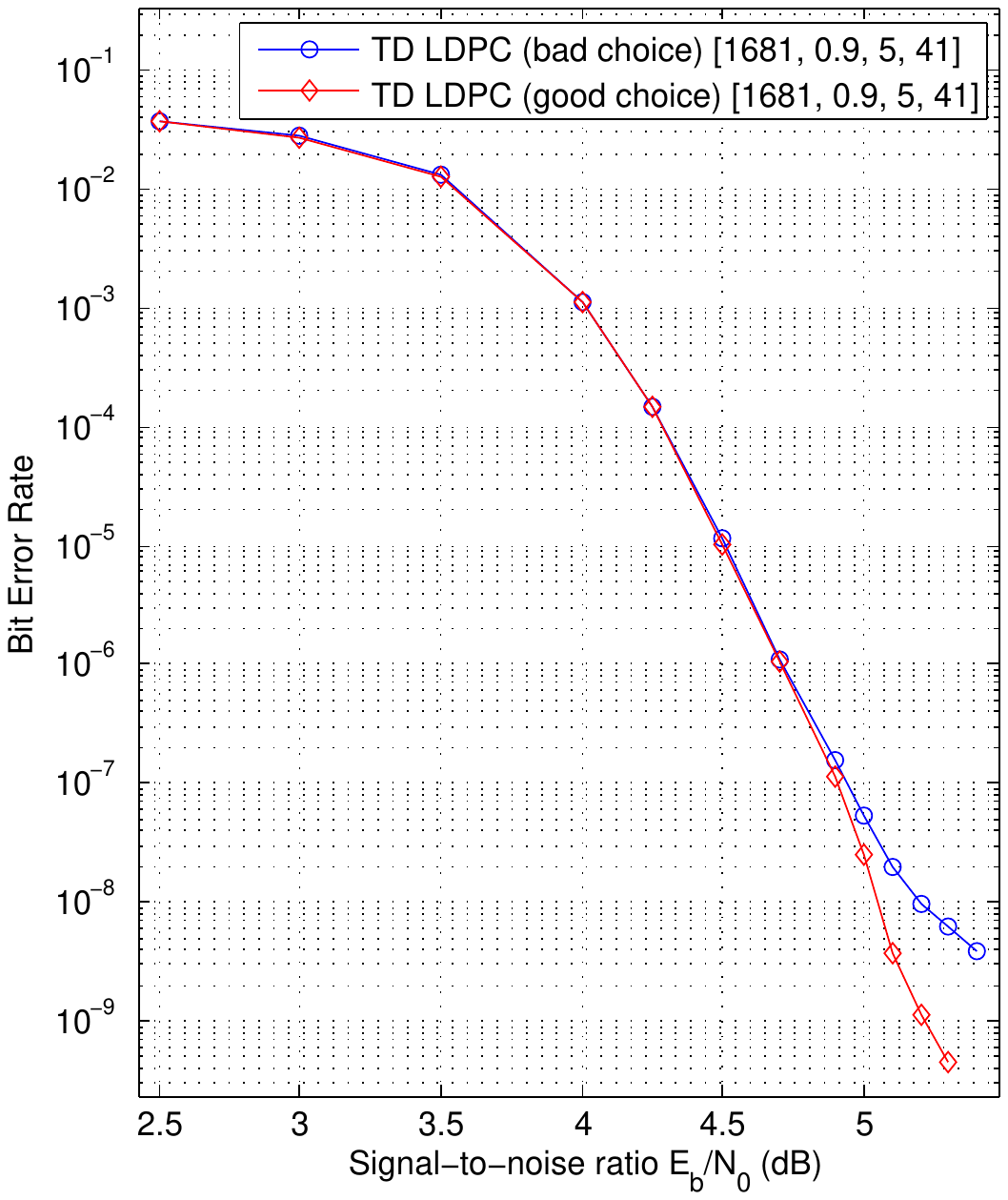}
	}
	\vspace{0.5cm}
	\caption{Decoding performance of various TD LDPC codes over the AWGN channel}
	\label{Simu}
\end{figure*}

\section{Conclusion}\label{conclusion}

In this paper we have demonstrated that the class of transversal designs based on cyclic-structured MOLS generates a wide range of quasi-cyclic LDPC codes with excellent decoding performances over the AWGN channel via SPA decoding, in particular with low error-floors. By investigating the presence or absence of absorbing sets in these codes, we have derived powerful constraints for the proper choice of code parameters in order to eliminate small and harmful absorbing sets. The presented constraints are derived for LDPC codes of column weight $3$ and $4$ but are also potentially beneficial for codes of higher column weights. Since absorbing sets are also known to be stable under bit-flipping decoders, the presented codes should also reveal excellent performances over the BSC via bit-flipping decoding.

\appendices

\section{} \label{proof:orthogonality_of_latin_squares}

\begin{proof}
We prove both directions by contraposition.

\subsubsection{If $\alpha_1 \beta_2 \neq \alpha_2 \beta_1$, then the Latin squares are orthogonal} We prove the contrapositive of this statement: Suppose that the Latin squares are \textbf{not} orthogonal, then there must be two cell positions $[x_1,y_1]$ and $[x_2,y_2]$ such that $\CMcal{L}^{(\alpha_1,\beta_1)}_q[x_1,y_1]=\CMcal{L}^{(\alpha_1,\beta_1)}_q[x_2,y_2]$ and $\CMcal{L}^{(\alpha_2,\beta_2)}_q[x_1,y_1]=\CMcal{L}^{(\alpha_2,\beta_2)}_q[x_2,y_2]$. It follows with Lemma~\ref{simple_structured_MOLS} that
(1) $\alpha_1(x_1-x_2)+\beta_1(y_1-y_2)=0$ and (2) $\alpha_2(x_1-x_2)+\beta_2(y_1-y_2)=0$. 
After multiplicating (1) with $(-\alpha_2)$, (2) with $\alpha_1$ and adding both results, we obtain $(y_1-y_2)(\alpha_1\beta_2-\alpha_2\beta_1)=0$. It follows that $y_1=y_2$ or $\alpha_1\beta_2=\alpha_2\beta_1$. The first condition can never be satisfied, since $y_1$ and $y_2$ represent two separate columns. Hence, the second condition must be satisfied. 

\subsubsection{If the Latin squares are orthogonal, then $\alpha_1 \beta_2 \neq \alpha_2 \beta_1$} Again, we prove the contrapositive of this statement: Suppose that $\alpha_1 \beta_2 = \alpha_2 \beta_1$. Let $[x_1,y_1]$ and $[x_2,y_2]$ be any two cell positions such that $\CMcal{L}^{(\alpha_1,\beta_1)}_q[x_1,y_1]=\CMcal{L}^{(\alpha_1,\beta_1)}_q[x_2,y_2]$. With Lemma~\ref{simple_structured_MOLS}, it follows that $\alpha_1(x_1-x_2)+\beta_1(y_1-y_2)=0$. By multiplicating with $\beta_2$ and replacing $\alpha_1 \beta_2$ with $\alpha_2 \beta_1$, we obtain $\beta_1(\alpha_2(x_1-x_2)+\beta_2(y_1-y_2))=0$. Since $\beta_1$ must be positive, it follows that $\alpha_2(x_1-x_2)+\beta_2(y_1-y_2)=0$. Hence, $\alpha_2 x_1 + \beta_2 y_1 = \alpha_2 x_2 + \beta_2 y_2$ and thus $\CMcal{L}^{(\alpha_2,\beta_2)}_q[x_1,y_1]=\CMcal{L}^{(\alpha_2,\beta_2)}_q[x_2,y_2]$. Consequently, the Latin squares can not be orthogonal. 
\end{proof}

\begin{IEEEbiography}[{\includegraphics[width=1in, height=1.25in,clip,keepaspectratio]{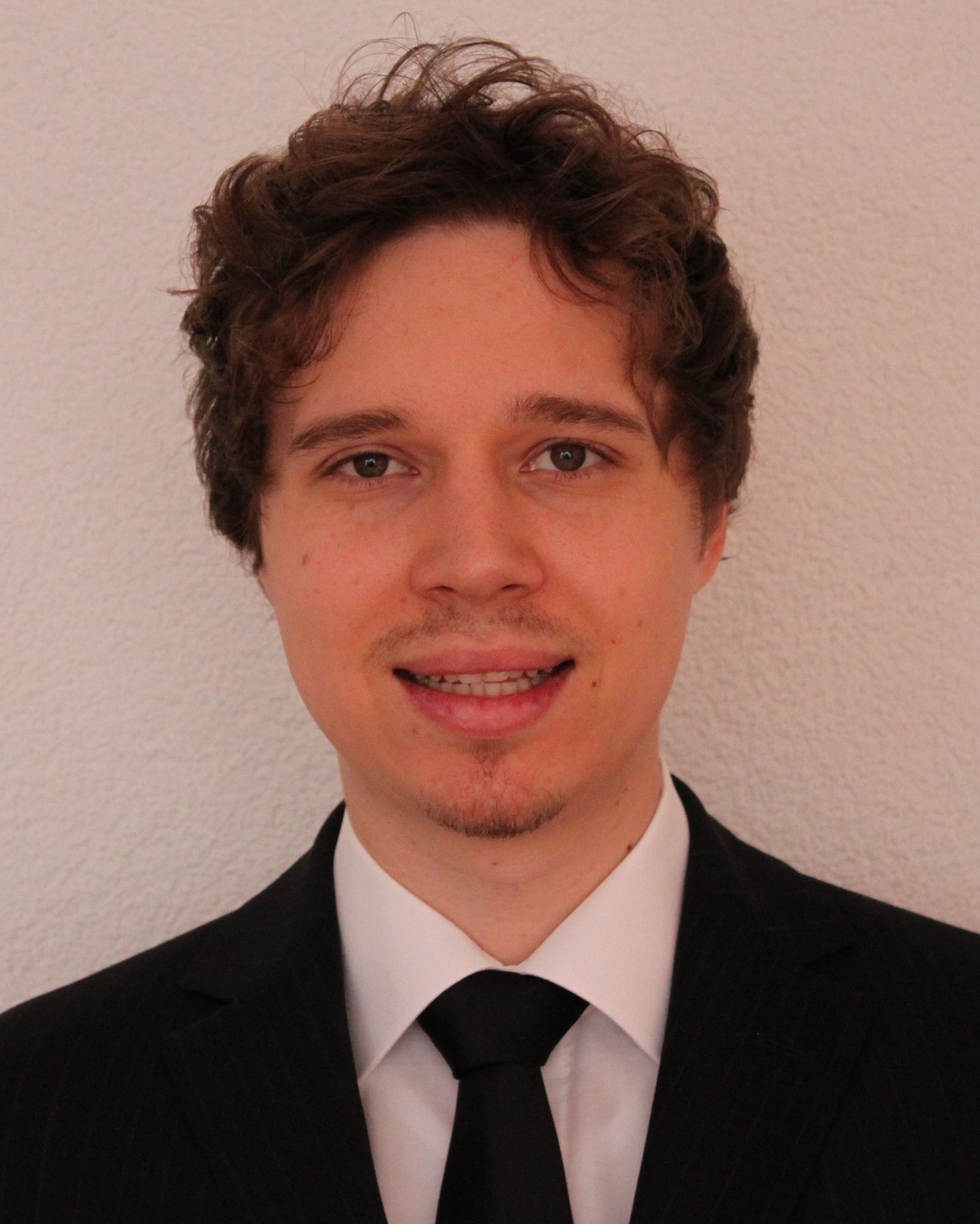}}]{Alexander Gruner} is a Ph.D. student in computer science at the Wilhelm Schickard Institute for Computer Science, University of T{\"u}bingen, Germany, where he is part of an interdisciplinary research training group in computer science and mathematics. He received the Diploma degree in computer science from the University of T{\"u}bingen in 2011. His research interests are in the field of coding and information theory with special emphasis on turbo-like codes, codes on graphs and iterative decoding.
\end{IEEEbiography}

\begin{IEEEbiography}[{\includegraphics[width=1in, height=1.25in,clip,keepaspectratio]{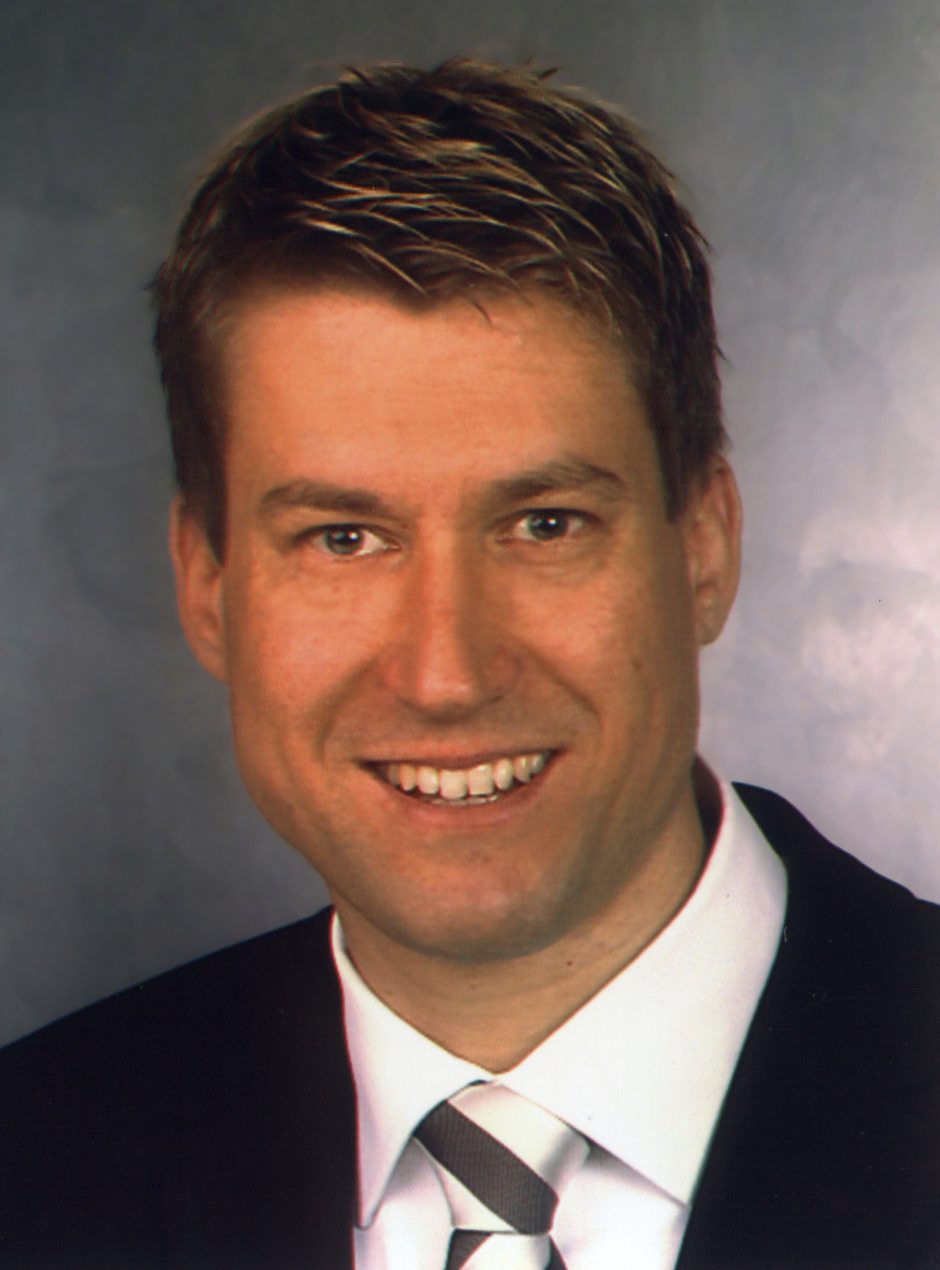}}]{Michael Huber} (M.'09) is currently a manager for data quality and business intelligence at Mercedes-Benz Bank AG, Stuttgart, and adjunct professor at the Wilhelm Schickard Institute for Computer Science, University of T\"ubingen, Germany. He was a Heisenberg fellow of the German Research Foundation (DFG) from 2008-12 and a visiting full professor in mathematics at Berlin Technical University from 2007-08. He obtained the Diploma, Ph.D. and Habilitation degrees in mathematics from the University of T\"ubingen in 1999, 2001 and 2006, respectively. He was awarded the 2008 Heinz Maier Leibnitz Prize by the DFG and the German Ministry of Education and Research (BMBF). He became a Fellow of the Institute of Combinatorics and Its Applications (ICA), Winnipeg, Canada, in 2009.

Prof. Huber's research interests are in the areas of coding and information theory, cryptography and information security, discrete mathematics and combinatorics, algorithms, and data analysis and visualization with applications to business and biological sciences. Among his publications in these areas are two books, Flag-transitive Steiner Designs (Birkh\"auser Verlag, Frontiers in Mathematics, 2009) and Combinatorial Designs for Authentication and Secrecy Codes (NOW Publishers, Foundations and Trends in Communications and Information Theory, 2010). He is a Co-Investigator of an interdisciplinary research training group in computer science and mathematics at the University of T\"ubingen.
\end{IEEEbiography}

\enlargethispage{-3in}


\begin{thebibliography}{19}

\bibitem{Richardson2003}
T.~J.~Richardson, ``Error-floors of LDPC codes,'' in \emph{Proc. 41st Annu. Allerton Conf. Commun., Control, Comput.,} Monticello, IL, Oct. 1--3, 2003, pp. 1426--1435.

\bibitem{Di02}
C.~Di, D.~Proietti, E.~Telatar, T.~Richardson, and R.~Urbanke, ``Finite length
  analysis of low-density parity-check codes,'' \emph{IEEE Trans. on
  Information Theory}, vol.~48, no.~6, pp. 1570--1579, 2002.

\bibitem{Ka03}
N.~Kashyap and A.~Vardy, ``Stopping sets in codes from designs,''
  \emph{Proc. IEEE Int. Symp. Information Theory}, p. 122, 2003.

\bibitem{SchwVar06}
M.~Schwartz and A.~Vardy, ``On the stopping distance and the stopping redundancy of codes,''
  \emph{IEEE Trans. on Information Theory}, vol.~52, no.~3, pp. 922--932, 2006.

\bibitem{Zhang2006}
Z.~Zhang, L.~Dolecek, B.~Nikolic, V.~Anantharam, and M.~J.~Wainwright, ``Investigation of error floors of structured low-density parity-check codes
by hardware emulation,'' in \emph{Proc. IEEE Global Telecommunications
Conference (GLOBECOM)}, 2006, pp. 1--6.

\bibitem{Zhang2009}
Z.~Zhang, L.~Dolecek, B.~Nikolic, V.~Anantharam, and M.~J.~Wainwright, ``Design of LDPC decoders for improved low error rate performance:
quantization and algorithm choices,'' in \emph{IEEE Trans. Wireless Commun.}, vol.~57, no.~11, pp. 3258--3268, 2009.

\bibitem{Dolecek2010}
L.~Dolecek, Z.~Zhang, V.~Anantharam, M.~J.~Wainwright and B.~Nikoli\'{c}, ``Analysis of Absorbing Sets and Fully Absorbing Sets
of Array-Based LDPC Codes,''
  \emph{IEEE Trans. on Information Theory}, vol.~56, no.~1, pp. 181--201,
  2010.

\bibitem{Wang2013}
J.~Wang, L.~Dolecek and R.~Wesel, ``Controlling LDPC Absorbing Sets via the Null Space of the Cycle Consistency Matrix,'' \emph{IEEE Trans. on
  Information Theory}, vol.~59, no.~4, pp. 2293--2314, 2013.

\bibitem{DolWang2010}
L.~Dolecek, J.~Wang and Z.~Zhang, ``Towards improved LDPC code designs using absorbing set spectrum properties,'' in \emph{Proc. 6th International Symposium on Turbo Codes and Iterative Information Processing (ISTC)}, 2010, pp. 477--481.

\bibitem{GrunHub2013}
A.~Gruner and M.~Huber, ``Low-Density Parity-Check Codes from Transversal Designs with Improved Stopping Set Distributions,'' \emph{IEEE Trans. on Communications}, vol.~61, no.~6,
  pp. 2190--2200, 2013.

\bibitem{TowWel67}
R.~L. Townsend and E.~J.~Weldon, Jr., ``Self-orthogonal quasi-cylic codes,''
  \emph{IEEE Trans. on Information Theory}, vol. IT-13, no.~2, pp. 183--195,
  1967.

\bibitem{JohnWell2004}
S.~J. Johnson and S.~R. Weller, ``Codes for iterative decoding from partial geometries,'' \emph{IEEE Trans. on Communications}, vol.~52, no.~2,
  pp. 236--243, 2004.

\bibitem{JohnsonDiss}
S.~J. Johnson, ``Low-density parity-check codes from combinatorial designs,''
  Ph.D. dissertation, School of Electrical Engineering and Computer Science,
  University of Newcastle, 2004.


\bibitem{crc_handbook}
C.~J. Colbourn and J.~Dinitz, Eds., \emph{The CRC Handbook of Combinatorial
  Designs}, 2nd~ed.,\hskip 1em plus 0.5em minus 0.4em\relax CRC Press, 2006.

\bibitem{BosShri90}
R.~C. Bose and S.~S. Shrikhande, ``On the construction of sets of mutually
  orthogonal latin squares and the falsity of a conjecture of {E}uler,''
  \emph{Transactions of the American Mathematical Society}, vol.~95, pp.
  191--209, 1960.

\bibitem{Wilson74}
R.~M. Wilson, ``Concerning the number of mutually orthogonal latin squares,''
  \emph{Discrete Mathematics}, vol.~9, pp. 181--198, 1974.

\bibitem{Hu64}
X.-Y.~Hu, E.~Eleftheriou and D.~M.~Arnold, ``Regular and irregular progressive edge-growth tanner graphs,''
  \emph{IEEE Trans. on Information Theory}, vol.~51, no.~1, pp. 386--398,
  1964.

\bibitem{Zhang2010}
L. Zhang, Q. Huang, S. Lin, K. Abdel-Ghaffar and I.~F.~Blake, ``Quasi-Cyclic LDPC Codes: An Algebraic Construction, Rank Analysis, and Codes on Latin Squares,'' \emph{IEEE Trans. on Communications}, vol.~58, no.~11,
  pp. 3126--3139, 2010.

\bibitem{Vasic04}
B.~Vasic and O.~Milenkovic, ``Combinatorial constructions of low-density parity check codes for
  iterative decoding,'' \emph{IEEE Trans. on Communications}, vol.~50, no.~6,
  pp. 1156--1176, 2004.

\end{thebibliography}
\end{document}